\documentclass[english]{article}
\usepackage[T1]{fontenc}
\usepackage[latin9]{inputenc}
\usepackage{geometry}
\usepackage{babel}
\usepackage{float}
\usepackage{textcomp}
\usepackage{amsthm}
\usepackage{amsmath}
\usepackage{amssymb}
\usepackage{graphicx}
\usepackage[unicode=true,pdfusetitle,
 bookmarks=true,bookmarksnumbered=false,bookmarksopen=true,bookmarksopenlevel=1,
 breaklinks=false,pdfborder={0 0 0},backref=page,colorlinks=false]
 {hyperref}
\usepackage{breakurl}

\makeatletter

\floatstyle{ruled}
\newfloat{algorithm}{tbp}{loa}
\providecommand{\algorithmname}{Algorithm}
\floatname{algorithm}{\protect\algorithmname}

\theoremstyle{plain}
\newtheorem{thm}{\protect\theoremname}[section]
  \theoremstyle{plain}
  \newtheorem{lem}[thm]{\protect\lemmaname}
  \theoremstyle{plain}
  \newtheorem{prop}[thm]{\protect\propositionname}
  \theoremstyle{definition}
  \newtheorem{defn}[thm]{\protect\definitionname}
  \theoremstyle{remark}
  \newtheorem{claim}[thm]{\protect\claimname}
  \theoremstyle{plain}
  \newtheorem{cor}[thm]{\protect\corollaryname}

\makeatother

  \providecommand{\claimname}{Claim}
  \providecommand{\corollaryname}{Corollary}
  \providecommand{\definitionname}{Definition}
  \providecommand{\lemmaname}{Lemma}
  \providecommand{\propositionname}{Proposition}
\providecommand{\theoremname}{Theorem}

\begin{document}

\title{On the effect of randomness on planted 3-coloring models}
\author{
	Roee David\thanks{
		Department of Computer Science and Applied Mathematics, Weizmann Institute of Science, Rehovot 76100,
		Israel. E-mail: {\tt roee.david@weizmann.ac.il}.}
	\and Uriel Feige\thanks{
		Department of Computer Science and Applied Mathematics, Weizmann Institute of Science, Rehovot 76100,
		Israel. E-mail: {\tt uriel.feige@weizmann.ac.il}.}
	}

\maketitle
\begin{abstract}
We present the {\emph hosted coloring} framework for studying algorithmic
and hardness results for the $k$-coloring problem. There is a class
${\cal H}$ of host graphs. One selects a graph $H\in{\cal H}$ and
plants in it a balanced $k$-coloring (by partitioning the vertex
set into $k$ roughly equal parts, and removing all edges within each
part). The resulting graph $G$ is given as input to a polynomial
time algorithm that needs to $k$-color $G$ (any legal $k$-coloring
would do -- the algorithm is not required to recover the planted $k$-coloring).
Earlier planted models correspond to the case that ${\cal H}$ is
the class of all $n$-vertex $d$-regular graphs, a member $H\in{\cal H}$
is chosen at random, and then a balanced $k$-coloring is planted
at random. Blum and Spencer~[1995] designed algorithms
for this model when $d=n^{\delta}$ (for $0<\delta\le1$), and Alon
and Kahale~[1997] managed to do so even when $d$
is a sufficiently large constant.


The new aspect in our framework is that it need not involve randomness.
In one model within the framework (with $k=3$)
$H$ is a $d$ regular spectral expander (meaning that except for the largest eigenvalue of its adjacency
matrix, every other eigenvalue has absolute value much smaller than $d$) chosen by an adversary, and the planted 3-coloring is random.
We show that the 3-coloring algorithm of Alon and Kahale~[1997] can be modified to apply to this case.
In another model $H$ is a random $d$-regular graph but the planted balanced $3$-coloring is chosen by an adversary,
after seeing $H$. We show that for a certain range of average degrees somewhat below $\sqrt{n}$, finding a 3-coloring is NP-hard.
Together these results (and other results that we have) help clarify which aspects of randomness in the planted coloring model are the key to successful 3-coloring algorithms.


\end{abstract}
\newpage{}

\section{Introduction}

A $k$-coloring of a graph $G(V,E)$ is an assignment $\chi:V\longrightarrow[k]$
of colors to vertices such that for every edge $(u,v)\in E$ one has
$\chi(u)\not=\chi(v)$. The problem of deciding whether a given graph
is $k$-colorable is NP-hard for every $k\ge3$~\cite{DBLP:conf/coco/Karp72,garey1976some}.
Moreover, even the most sophisticated coloring algorithms known require
(on worst case instances) $|V|^{\delta}$ colors (for some $\delta\simeq0.2$)
in order to properly color a 3-colorable graph~\cite{kawarabayashi2014coloring}.

An approach for coping with NP-hardness is by restricting the class
of input instances in a way that either excludes the most difficult
instances, or makes them unlikely to appear (in models in which there
is a probability distribution over inputs). Along this line, a model
that is very relevant to our current work is the so called {\em
random planted coloring model} $G_{n,k,p}$ in which the vertex set
(of cardinality $n$) is partitioned at random into $k$ parts, and
edges between vertices in different parts are placed independently
with probability $p$. Such graphs are necessarily $k$-colorable.
Following initial work by Blum and Spencer~\cite{JALGO::BlumS1995},
it was shown by Alon and Kahale~\cite{alon1997spectral} that for
every $k$ there is a polynomial time algorithm that with high probability
$k$-colors such input graphs (the probability is taken over random
choice of input graphs) provided that $p>\frac{c_{k}}{n}$, where
$c_{k}$ is some constant that depends only on $k$.

In the current work we propose a framework that contains several different
models for generating instances of $k$-colorable graphs. We call
this framework the {\em hosted coloring} framework. The random
planted coloring model is one of the models that is contained in the
hosted coloring framework. We consider several other planted coloring
models within our framework, and obtain both new algorithmic results
and new hardness results. In particular, our results help clarify
the role that randomness plays in the random planted model.

\subsection{The hosted coloring framework}

We describe our framework for generating instances with planted solutions.
In the current manuscript, the framework is described only in the
special case of the $k$-coloring problem, though it is not difficult
to extend it to other NP-hard problems.

The hosted coloring framework is a framework for generating $k$-colorable
graphs. We alert the reader that graphs within this framework are
{\em labeled}, meaning that every $n$-vertex graph is given together
with a naming of its vertices from~1 to~$n$. A model within this
framework involves two components:
\begin{enumerate}
\item A class ${\cal {H}}$ of host graphs. Let ${\cal H}_{n}$ denote the
set of graphs in ${\cal {H}}$ that have $n$ vertices. 

\item A class ${\cal {P}}$ of planted solutions. Formally, in the context
of $k$-coloring, a planted solution can be thought of as a complete
$k$-partite graph. Let ${\cal P}_{n}$ denote the set of planted
solutions in ${\cal {P}}$ that have $n$ vertices. 

\end{enumerate}
To generate a $k$-colorable graph with $n$ vertices, one selects
one graph $H$ from ${\cal H}_{n}$, and plants in it one solution
$P$ from ${\cal P}_{n}$. Formally, the planting can be described
as generating the graph $G(V,E)$ whose edge set is the intersection
of the edge sets of the host graph $H$ and of the complete $k$-partite
graph $P$. Namely, the vertex set of $G(V,E)$ is $V=[n]$, and $(u,v)\in E$
iff both $(u,v)\in E(H)$ and $(u,v)\in E(P)$.

To complete the description of the hosted framework, we explain how
the host graph $H\in{\cal H}_{n}$ is selected, and how the planted
solution $P\in{\cal P}_{n}$ is selected. Here, the framework allows
for four {\em selection rules}:
\begin{enumerate}
\item Adversarial/adversarial. An adversary selects $H\in{\cal H}_{n}$
and $P\in{\cal P}_{n}$.
\item Random/random. The class of host graphs is equipped with a probability
distribution (typically simply the uniform distribution) and likewise
for the class of planted solutions. The selections of host graph and
planted solution are done independently at random, each according
to its own distribution. We use the notation $H\in_{R}{\cal H}_{n}$
and $P\in_{R}{\cal P}_{n}$ to describe such selections.
\item Adversarial/random. An adversary first selects $H\in{\cal H}_{n}$,
and then $P\in_{R}{\cal P}_{n}$ is selected at random.
\item Random/adversarial. A host graph $H\in_{R}{\cal H}_{n}$ is selected
at random, and then an adversary, upon seeing $H$, selects $P\in{\cal P}_{n}$.
\end{enumerate}
The hosted $k$-coloring framework allows for many different planted
$k$-coloring models, depending on the choice of ${\cal H}$, ${\cal P}$
and the selection rule. The random planted $G_{n,k,p}$ model can
be described within the hosted $k$-coloring framework by taking ${\cal H}_{n}$
to be the class of all $n$-vertex graphs equipped with the Erdos-Renyi
probability distribution $G_{n,p}$, taking ${\cal P}_{n}$ to be
the class of all $n$-vertex complete $k$-partite graphs equipped
with the uniform distribution, and using the random/random selection
rule.

One of the goals of our work is to remove randomness from planted
models. The hosted $k$-coloring framework allows us to do this. Moreover,
one can control separately different aspects of randomness. Let us
explain how this is done in our work.

\textbf{Regular expanders as host graphs.} Randomness is eliminated
from the choice of host graph by allowing the adversary to select
an arbitrary host graph from the class of $d$-regular $\lambda$-expanders,
for given $d$ (that may be a function of $n$) and $\lambda$ (which
is a function of $d$). The term $\lambda$-expander refers to the
spectral manifestation of graph expansion. Namely, let $\lambda_{1}\ge\ldots\ge\lambda_{n}$
denote the eigenvalues of the adjacency matrix of an $n$ node graph
$G$, and let $\lambda=\max[\lambda_{2},|\lambda_{n}|]$. As is well
known, a $d$-regular graph has $\lambda_{1}=d$ and $\lambda\ge\Omega(\sqrt{d})$~\cite{nilli1991second}.
A $d$-regular graph is referred to as a spectral expander if $\lambda$
is significantly smaller than $d$ -- the smaller $\lambda$ is the
better the guaranteed expansion properties are~\cite{hoory2006expander}.

Random $d$-regular graphs are essentially the best possible spectral
expanders, satisfying $\lambda=O(\sqrt{d})$ almost surely~\cite{feige2005spectral,furedi1981eigenvalues,friedman1989second}.
The same holds for random graphs in the $G_{n,p}$ model, taking $d$
to be the average degree (roughly $pn$). We remark that our results
extend to graphs that are approximately $d$-regular (e.g., with degree
distribution similar to $G_{n,p}$), and regularity is postulated
only so as to keep the presentation simple.

\textbf{Balanced coloring.} Randomness is eliminated from the choice
of planted coloring (the choice of $P$) by allowing the adversary,
after seeing $H$, to plant an arbitrary {\em balanced} $k$-coloring,
namely, to select an arbitrary $k$-partite graph in which all parts
are of size $\frac{n}{k}$. (Also here, our results extend to having
part sizes of roughly $\frac{n}{k}$ rather than exactly, and exact
balance is postulated only for simplicity.) Observe that a random
$k$-partition is nearly balanced almost surely.

\subsection{Main results}

For simplicity we focus here on the special case of $k=3$, namely,
3-coloring. Extensions of our results to $k>3$ are discussed in Section~\ref{sec:k>3}.
In our main set of results, we shall consider four related planted
models, all within the hosted $k$-coloring framework.

The four models will have selection rules referred to as $H_{A}/P_{A}$,
$H_{A}/P_{R}$, $H_{R}/P_{A}$, and $H_{R}/P_{R}$, where $H$ and
$P$ refer to host graph and planted coloring respectively, and $A$
and $R$ stand for {\em adversarial} and {\em random} respectively.
\begin{itemize}
\item $H_{A}$ means that the adversary chooses an arbitrary $d$-regular
$\lambda$-expander (for an appropriate choice of $d$ and $\lambda$)
host graph -- we refer to this as an {\em adversarial expander};
\item $H_{R}$ means that the host graph is chosen as a random Erdos-Renyi
random graph $G_{n,p}$ (for an appropriate choice of $p$) -- we
refer to this as a {\em random host graph}.
\item $P_{R}$ refers to a random planted coloring (complete tripartite
graph) chosen uniformly at random -- we refer to this is {\em random
planting};
\item $P_{A}$ refers to a balanced planted coloring chosen adversarially
(after the adversary sees $H$) -- we refer to this is {\em adversarial
planting}.
\end{itemize}

In all cases, $n$ denotes the number of vertices in the graph, and
$d$ denotes the average degree of the host graph (where $d\simeq pn$
for random host graphs). We shall say that two colorings of the same
set of vertices are identical if the partitions that the color classes
induce on the vertices are the same (the actual names of colors are
irrelevant).

In presenting our results it will be instructive to consider the following
notions of coloring:
\begin{itemize}
\item The {\em planted} 3-coloring $P$.
\item A {\em legal} 3-coloring (but not necessarily the planted one).
\item For a given $b<n$, a {\em $b$-approximated} coloring is a 3-coloring
that is not necessarily legal, but it is identical to the planted
3-coloring on a set of at least $n-b$ vertices.
\item For a given $b<n$, a {\em $b$-partial} coloring is a 3-coloring
of $n-b$ vertices from the graph that is identical to the planted
coloring on these vertices. The remaining $b$ vertices are left uncolored
and are referred to as {\em free}.
\end{itemize}
Our first theorem offers a unifying theme for all four models. For
the $H_{R}/P_{R}$ model a similar theorem was known~\cite{alon1997spectral}.
(Recall that $\lambda$, the second largest in absolute value eigenvalue,
is a measure of expansion and satisfies $\lambda=\Theta(\sqrt{d})$
for random graphs.)
\begin{thm}
\label{thm:partial} For a sufficiently large constant $c$, let the
average degree in the host graph satisfy $c<d<n$. Then in all four
models ($H_{A}/P_{A}$, $H_{A}/P_{R}$, $H_{R}/P_{A}$, $H_{R}/P_{R}$)
there is a polynomial time algorithm that finds a $b$-partial coloring
for $b=O(\left(\frac{\lambda}{d}\right)^{2}n)$. For the models with
random host graphs ($H_{R}$) and/or random planted colorings ($P_{R}$),
the algorithm succeeds with high probability over choice of random
host graph $H$ and/or random planted coloring $P$.
\end{thm}

Given that Theorem~\ref{thm:partial} obtains a $b$-partial coloring,
the task that remains is to 3-color the set of $b$ free vertices,
in a way that is both internally consistent (for edges between free
vertices) and externally consistent (for edges with only one endpoint
free). From~\cite{alon1997spectral} it is known that this task can
be completed in the $H_{R}/P_{R}$ model. Here are our main results
for the other planted models. In all cases, $d$ is the average degree
of the host graph.
\begin{thm}
\label{thm:A/A} In the $H_{A}/P_{A}$ model, for every $d$ in the
range $C<d<n^{1-\epsilon}$ (where $C$ is a sufficiently large constant
and $\epsilon>0$ is arbitrarily small), it is NP-hard to 3-color
a graph with a planted 3-coloring, even when $\lambda=O(\sqrt{d})$.
\end{thm}

\begin{thm}
\label{thm:A/R} In the $H_{A}/P_{R}$ model, for some constants $0<c<1$
and $C>1$ there is a polynomial time algorithm with the following
properties. For every $d$ in the range $C<d\le n-1$ and every $\lambda\le cd$,
for every host graph within the model, the algorithm with high probability
(over the choice of random planted coloring) finds a legal 3-coloring.
\end{thm}

\begin{thm}
\label{thm:R/A} In the $H_{R}/P_{A}$ model:
\begin{description}
\item [{a}] There is a polynomial time algorithm that with high probability
(over the choice of host graph) finds a legal 3-coloring whenever
$d\ge Cn^{2/3}$ (for a sufficiently large constant $C$).
\item [{b}] There is a constant $\frac{1}{3}<\delta_{0}<\frac{1}{2}$ such
that for $\delta_{0}<\delta<\frac{1}{2}$, 
 no polynomial time algorithm has constant probability (over the random
choice of host graph of average degree $n^{\delta}$, for adversarially
planted coloring) to produce a legal 3-coloring, unless NP has expected
polynomial time algorithms.
\end{description}
\end{thm}

Let us briefly summarize our main findings as to the role of randomness
in planted 3-coloring models. For partial coloring, randomness in
the model can be replaced by degree and expansion requirements for
the host graph, and balance requirements for the planted coloring
(see Theorem~\ref{thm:partial}). For finding a legal (complete)
3-coloring, randomness of the planted coloring is the key issue, in
which case it suffices that the host graph is an arbitrary spectral
expander, and in fact, quite a weak one ($\lambda$ can even be linear
in $d$ -- see Theorem~\ref{thm:A/R}). If the planted 3-coloring
is not random, then spectral expansion does not suffice (not even
$\lambda=O(\sqrt{d})$ -- see Theorem~\ref{thm:A/A}), and moreover,
even randomness of the host graph does not suffice (for some range
of degrees -- see Theorem~\ref{thm:R/A}b).
Finally, comparing Theorem~\ref{thm:R/A}a to Theorem~\ref{thm:A/A}
shows that spectral expansion cannot always replace randomness of
the host graph.







\subsection{Related work}


There is a vast body of work on models with planted solutions, and here we shall survey only a sample of it that suffices in order to understand the context for our results.

In our framework we ask for algorithms that $k$-color a graph that has a planted $k$-coloring, and allow the algorithm to return any legal $k$-coloring, not necessarily the planted one. We refer to this as the {\em optimization} version. The optimization version regards planted models as a framework for studying possibly tractable instances for otherwise NP-hard problems. In certain other contexts (signal processing, statistical inference) the goal in planted models is to recover the planted object, either exactly, or approximately. We refer to this as the {\em recovery} version. It is often motivated by practical needs (e.g., to recover a a true signal from a noisy version, to cluster noisy data, etc.). Our Theorem~\ref{thm:partial} addresses the approximate recovery question, but for many of our models and settings of parameters, exact recovery of the planted $k$-coloring is information theoretically impossible (one reason being that the input graph might have multiple legal $k$-coloring with no indication which is the planted one). In general, optimization becomes more difficult when exact recovery is impossible, but still in many of our models we manage to exactly solve the optimization problem (e.g., in Theorem~\ref{thm:A/R}). In the context of random planted $k$-coloring, the $k$-coloring algorithms of~\cite{JALGO::BlumS1995} could work only in the regime in which exact recovery is possible, whereas the algorithms of~\cite{alon1997spectral} work also in regimes where exact recovery is not possible.

There are planted models and corresponding algorithms for many other optimization problems, including graph bisection~\cite{boppana1987eigenvalues}, 3-SAT~\cite{flaxman2003spectral}, general graph partitioning~\cite{mcsherry2001spectral}, and others. A common algorithmic theme used in many of this works is that (depending on the parameters) exact or approximate solutions can be found using spectral techniques. The reason why the class of host graphs that we consider is that of spectral expanders is precisely because we can hope that spectral techniques will be applicable in this case. Indeed, the first step of our algorithm (in the proof of Theorem~\ref{thm:partial}) employs spectral techniques. Nevertheless, its proof differs from previous proofs in some of its parts, because it uses only weak assumptions on the planted model (there is no randomness involved, the host graph is an arbitrary spectral expander, and moreover not necessarily a very good one, and the k-coloring is planted in a worst case manner in the host graph).

Our hosted coloring framework allows the input to be generated in a way that is partly random and partly adversarial. Such models are often referred to as semi-random models~\cite{JALGO::BlumS1995,feige2001heuristics}. Among the motivations to semi-random models one can mention attempts to capture real life instances better than purely random models ({\em smoothed analysis}~\cite{spielman2004smoothed} is a prominent example of this line of reasoning, but there are also other recent attempts such as~\cite{makarychev2012approximation} and~\cite{makarychev2014constant}), and attempts to understand better how worst case input instances to a problem may look like (e.g., see the semirandom models for unique games in~\cite{kolla2011play}, which consider four different aspects of an input instance, and study all combinations in which three of these aspects are adversarial and only one is random). Another attractive aspect of semi-random models is the possibility of matching algorithmic results by NP-hardness results, which become possible (in principle) due to the presence of an adversary (there are no known NP-hardness results in purely random planted models). NP-hardness results for certain semi-random models of $k$-coloring (in which an adversary is allowed to add arbitrary edges between color classes in a random planted $G_{n,k,p}$ graph) have been shown in~\cite{feige2001heuristics}, thus explaining why the algorithmic results that were obtained there for certain values of $p$ cannot be pushed to considerably lower values of $p$. Typically these NP-hardness results are relatively easy to prove, due to a strong adversarial component in the planted model. In contrast, our NP-hardness result in Theorem~\ref{thm:R/A}b is proved in a context in which the adversary seems to have relatively little power (has no control whatsoever over the host graph and can only choose the planted coloring). Its proof appears to be different from any previous NP-hardness proof that we are aware of.

In \cite{arora2011new} an algorithm that outputs an $\Omega(n)$ size independent set for $d$-regular 3-colorable graphs is designed. The running time of the algorithm is $n^{O(D)}$, where $D$ is the \emph{threshold rank} of the input graph (namely, the adjacency matrix of the input graph has at most $D$ eigenvalues more negative than $-t$, where $t = \Omega(d)$). Graphs generated in our $H_A/P_R$ model are nearly regular and have low threshold rank. Graphs generated in our $H_A/P_A$ model might have vertices of degree much lower than the average degree, but they have large subgraphs of low threshold rank in which all degrees are within constant multiplicative factors of each other (and presumably the algorithms of \cite{arora2011new} can be adapted to such graphs). The hardness results presented in the current paper do not directly give regular graphs, but one can modify the hardness results for $k$-coloring in the $H_A/P_A$ model (specifically, Theorem~\ref{thm:A/A-k}, for $k > 3$) to obtain hardness of $k$-coloring regular graphs of low threshold rank.


Let us end this survey with two unpublished works (available online) that are related to the line of research presented in the current paper. A certain model for planted 3SAT was studied in~\cite{alina}. It turns out that in that model a $b$-partial solution (even for a very small value of $b$) can be found efficiently, but it is not known whether a satisfying assignment can be found. As that model is purely random, it is unlikely that one can prove that finding a satisfying assignment is NP-hard. In~\cite{palntedColoring} a particular model within the hosted coloring framework was introduced. In that model the class of host graphs is that of so called anti-geometric graphs, and both the choice of host graph and planted coloring are random. The motivation for choosing the class of anti-geometric graphs as host graphs is that these are the graphs on which~\cite{FeiLanSch04} showed integrality gaps for the semi-definite program of~\cite{karger1998approximate}. Hence spectral algorithms appear to be helpless in these planted model settings. Algorithms for 3-coloring were presented in~\cite{palntedColoring} for this class of planted models when the average degree is sufficiently large (above $n^{0.29}$), and it is an interesting open question whether this can be pushed down to lower degrees. If so, this may give 3-coloring algorithms that do well on instances on which semidefinite programming seems helpless.

\section{Overview of proofs}

In this section we explain the main ideas in the proof. The full proofs,
which often include additional technical content beyond the ideas
overviewed in this section, appear in the appendix.

\subsection{An algorithm for partial colorings}

\label{sec:partial}

Here we explain how Theorem~\ref{thm:partial} (an algorithm for
partial coloring) is proved. Our algorithm can be thought of as having
the following steps, which mimic the steps in the algorithm of Alon
and Kahale~\cite{alon1997spectral} who addressed the $H_{R}/P_{R}$
model.
\begin{enumerate}
\item {\em Spectral clustering.} Given an input graph $G$, compute the
eigenvectors corresponding to the two most negative eigenvalues of
the adjacency matrix of $G$. The outcome can be thought of as describing
an embedding of the vertices of $G$ in the plane (the coordinates
of each vertex are its corresponding entries in the eigenvectors).
Based on this embedding, use a distance based clustering algorithm
to partition the vertices into three classes. These classes form the
(not necessarily legal) coloring $\chi_{1}$.
\item {\em Iterative recoloring.} Given some (illegal) coloring, a {\em
local improvement} step moves a vertex $v$ from its current color
class to a class where $v$ has fewer neighbors, thus reducing the
number of illegally colored edges. Perform local improvement steps
(in parallel) until no longer possible. At this point one has a new
(not necessarily legal) coloring $\chi_{2}$.
\item {\em Cautious uncoloring.} Uncolor some of the vertices, making
them {\em free}. Specifically, using an iterative procedure, every
{\em suspect} vertex is uncolored, where a vertex $v$ is suspect
if it either has significantly less than $\frac{2d}{3}$ colored neighbors,
or there is a color class other than $\chi_{2}(v)$ with fewer than
$\frac{d}{6}$ neighbors of $v$. The resulting partial coloring is
referred to as $\chi_{3}$.
\end{enumerate}
The analysis of the three steps of the algorithm is based on that
of~\cite{alon1997spectral}, but with modifications due to the need
to address adversarial settings. Consequently, the values that we
obtain for the parameter $b$ after the iterative recoloring and cautious
uncoloring steps are weaker than the corresponding bounds in~\cite{alon1997spectral}.
\begin{lem}
\label{lem:step1} The coloring $\chi_{1}$ is a $b$-approximated
coloring for $b\le O(\frac{\lambda}{d}n)$.
\end{lem}
For the proof of Lemma~\ref{lem:step1}, the underlying idea is that
$d$-regular $\lambda$-expander graphs do not have any eigenvalues
more negative than $-\lambda$. On the other hand, planting a 3-coloring
can be shown to create exactly two eigenvalues of value roughly $-\frac{d}{3}$.
This is quite easy to show in the random planting model such as the
one used in~\cite{alon1997spectral} because the resulting graph
is nearly $\frac{2d}{3}$-regular. We show that this also holds in
the adversarial planted model. Thereafter, if $\frac{d}{3}>\lambda$,
it makes sense (though of course it needs a proof) that the eigenvalues
corresponding to the two most negative eigenvalues contain some information
about the planted coloring. An appropriate choice of clustering algorithm
can be used to extract this information. We remark that our choice
of clustering algorithm differs from and is more efficient than that
of~\cite{alon1997spectral}, a fact that is of little importance
in the context of planted 3-coloring, but does offer significant advantages
for planted $k$-coloring when $k$ is large. 

\begin{lem}
\label{lem:step2} The coloring $\chi_{2}$ is a $b$-approximated
coloring for $b\le O(\frac{\lambda^{2}}{d^{2}}n)$.
\end{lem}
In~\cite{alon1997spectral} a statement similar to Lemma~\ref{lem:step2}
was proved using probabilistic arguments (their setting is equivalent
to $H_{R}/P_{R}$). Our setting (specifically, that of $H_{A}/P_{A}$)
involves no randomness. We replace the proof of~\cite{alon1997spectral}
by a proof that uses only deterministic arguments. Specifically, we
use the well known expander mixing lemma~\cite{alon1988explicit}.
\begin{lem}
\label{lem:step3} The partial coloring $\chi_{3}$ is a $b$-partial
coloring for $b\le O(\frac{\lambda^{2}}{d^{2}}n)$.
\end{lem}
For the proof of Lemma~\ref{lem:step3}, the definition of {\em
suspect} vertex strikes the right balance between 
two conflicting requirements. One is ensuring that no colored vertex
remaining is wrongly colored. The other is that most of the graph
should remain remain colored. In~\cite{alon1997spectral} a statement
similar to Lemma~\ref{lem:step2} was proved using probabilistic
arguments, whereas our proof uses only deterministic arguments.

The full proof of Theorem~\ref{thm:partial} appears in Section~\ref{sec:Coloring Expander Graphs With Adversarial Planting}.

\subsection{Adversarial expanders with adversarial planting}

\label{sec:A/A}

Here we sketch how Theorem~\ref{thm:A/A} is proved, when the average
degree of the host graph is $d=n^{\delta}$ for some $0<\delta<1$.
Suppose (for the sake of contradiction) that there is a polynomial
time 3-coloring algorithm $\mbox{ALG}$ for the planted $H_{A}/P_{A}$ model.
We show how $\mbox{ALG}$ could be used to solve NP-hard problems, thus implying
P=NP.

Let ${\cal {Q}}$ be a class of sparse graphs on which the problem
of 3-coloring is NP-hard. For concreteness, we can take ${\cal {Q}}$
to be the class of 4-regular graphs. For simplicity, assume further
that if a graph in ${\cal {Q}}$ is 3-colorable, all color classes
are of the same size. (This can easily be enforced, e.g., by making
three copies of the graph.)

Given a graph $Q\in{\cal Q}$ on $n_{1}\simeq n^{1-\delta}<\frac{n}{4d}$
vertices for which one wishes to determine 3-colorability, do the
following. Construct an arbitrary spectral expander $Z$ on $n_{2}=n-n_{1}$
vertices, in which $n_{1}(d-4)$ vertices (called {\em connectors})
have degree $d-1$ and the rest of the vertices have degree $d$.
Plant an arbitrary balanced 3-coloring in $Z$ (each color class has
a third of the connector vertices and a third of the other vertices),
obtaining a graph that we call $Z_{3}$. Now give the graph $G$ that
is a disjoint union of $Q$ and $Z_{3}$ as input to $\mbox{ALG}$.
If $\mbox{ALG}$
finds a 3-coloring in $G$ declare $Q$ to be 3-colorable, and else declare
$Q$ as not having a 3-coloring.

Let us now prove correctness of the above procedure. If $Q$ is not 3-colorable, then clearly $\mbox{ALG}$ cannot 3-color $G$.
It remains to show that if $Q$ is 3-colorable,
then we can trust $\mbox{ALG}$ to find a 3-coloring of $G$. Namely, we need to show the existence of an expander host graph
$H$ and a planted 3-coloring in $H$ that after the removal of the
monochromatic edges produces exactly the graph $G$. An adversary
with unlimited computation power can derive $H$ from $Q$ and $Z$
as follows. It finds a balanced 3-coloring $\chi$ in $Q$. Then it
connects each vertex $v$ of $Q$ to $d-4$ distinct connector vertices
that have exactly the same color as $v$ (under $Z_{3}$). This gives
the graph $H$ which is $d$-regular, and for which planting the 3-coloring
$\chi$ on its $Q$ part and $Z_{3}$ on its $Z$ part gives the graph
$G$. It only remains to prove that $H$ is a spectral expander, but
this is not difficult.

The full proof of Theorem~\ref{thm:A/A} appears in Section~\ref{sec:On the hardness of coloring expander graphs with adversarial planting}.

\subsection{Adversarial expanders with random planting}

\label{sec:A/R}

Here we sketch the proof of Theorem~\ref{thm:A/R} (concerning $H_{A}/P_{R}$).
It would be instructive to first recall how~\cite{alon1997spectral}
completed the 3-coloring algorithm in the $H_{R}/P_{R}$ case. First,
in this case Theorem~\ref{thm:partial} can be considerably strengthened,
showing that one gets a $b$-partial coloring with $\frac{b}{n}$ exponentially
small in $d$. Hence when $d>>\log n$ this by itself recovers the
planted 3-coloring. The difficult case that remains is when $d$ is
sublogarithmic (e.g., $d$ is some large constant independent of $n$).
In this case it is shown in~\cite{alon1997spectral} that the subgraph
induced on the free vertices decomposes into connected components
each of which is smaller than $\log n$. Then each component by itself
can be 3-colored in polynomial time by exhaustive search, finding
a legal 3-coloring (not necessarily the planted one) for the whole
graph.

In the $H_{A}/P_{R}$ model it is still true that Theorem~\ref{thm:partial}
can be strengthened to show that one gets a $b$-partial coloring
with $\frac{b}{n}$ exponentially small in $d$. However, we do not know if
it is true that the subgraph induced on the free vertices decomposes
into connected components smaller than $\log n$. To overcome this,
we add another step to the algorithm (which is not required in the~\cite{alon1997spectral}
setting), which we refer to as {\em safe recoloring}. In this step,
iteratively, if an uncolored vertex $v$ has neighbors colored by
two different colors, then $v$ is colored by the remaining color.
Clearly, if one starts with a $b$-partial coloring, meaning that
all colored vertices agree with the planted coloring, this property
is maintained by safe recoloring. We prove that with high probability
(over choice of random planting), after the recoloring stage the remaining
free vertices break up into connected components of size $O\left(\frac{\lambda^{2}}{d^{2}}\log n\right)$.
Thereafter, a legal 3-coloring can be obtained in polynomial time
using exhaustive search.

The full proof of Theorem~\ref{thm:A/R} appears in Section~\ref{sec:Coloring expander graphs with random color planting}.

\subsection{Algorithm for random graphs with adversarial planting}

\label{sec:R/Aalg}

Here we sketch the proof of Theorem~\ref{thm:R/A}a (concerning an
algorithm for $H_{R}/P_{A}$). Given the negative result for $H_{A}/P_{A}$,
our algorithm must use a property that holds for random host graphs
but need not hold for expander graphs. The property that we use is
that when the degree $d$ is very large, the number of common neighbors
of every two vertices is larger than $O(\frac{n}{d})$. For random
graphs this holds (w.h.p.) whenever $d\ge Cn^{2/3}$ for a sufficiently
large constant $C$, but for expander graphs this property need not
hold. Recall that $b\le O(\frac{n}{d})$ in the $b$-partial coloring
from Theorem~\ref{thm:partial} (because $\lambda = O(\sqrt{d})$). Hence the pigeon-hole principle
implies that every two free vertices $u$ and $v$ had at least one
common neighbor $w$ (common neighbor in the host graph $H$) that
is not free, namely, it is colored. Hence if in the input graph $G$
neither of them have a colored neighbor in the $b$-partial coloring,
it must be that both of them lost their edge to $w$ because of the
planted coloring, meaning that $u$ and $v$ have the same color in
the planted coloring. Consequently, the set $F_{0}$ of all free vertices
with no colored neighbor must be monochromatic in the planted coloring.
This leads to the following algorithm for legally 3-coloring the free
vertices. Guess the color that should be given to the set $F_{0}$.
There are only three possibilities for this. Each of the remaining
free vertices has at least one colored neighbor, and hence at most two
possible colors. Hence we are left with a list-coloring problem with
at most two colors per list. This problem can be solved in polynomial
time by reduction to 2SAT.

The full proof of Theorem~\ref{thm:R/A}a appears in Section~\ref{sec:F 3-coloring random graphs with adversarial color planting}.

\subsection{Hardness for random graphs with adversarial planting}

\label{R/Ahard}

Here we sketch the proof of Theorem~\ref{thm:R/A}b (hardness for
$H_{R}/P_{A}$). The proof plan is similar to that of the proof of
Theorem~\ref{thm:A/A} (hardness for $H_{A}/P_{A}$, see Section~\ref{sec:A/A}),
but making this plan work is considerably more difficult.

Let us start with the proof plan. Suppose (for the sake of contradiction)
that there is a polynomial time 3-coloring algorithm $\mbox{ALG}$ for the
planted $H_{R}/P_{A}$ model with host graphs coming from $G_{n,p}$,
and hence of average degree $d\simeq pn$. We show how $\mbox{ALG}$ could
be used to solve NP-hard problems, thus implying P=NP.

Let ${\cal {Q}}$ be a (carefully chosen, a point that we will return
to later) class of sparse graphs on which the problem of 3-coloring
is NP-hard. As in the proof of Theorem~\ref{thm:A/A}, we may assume
that if a graph in ${\cal {Q}}$ is 3-colorable, all color classes
are of the same size.

Given a graph $Q\in{\cal Q}$ on $n_{1}=n^{\epsilon}$ vertices (for
some small $\epsilon>0$ to be determined later) for which one wishes
to determine 3-colorability, do the following. Construct a random
graph $Z$ on $n_{2}=n-n^{\epsilon}$ vertices, distributed like $G_{n_{2},p}$.
Plant a random balanced 3-coloring in $Z$, obtaining a graph
that we call $Z_{3}$. Now give the graph $G$ that is a disjoint
union of $Q$ and $Z_{3}$ as input to $\mbox{ALG}$. If $\mbox{ALG}$ finds a 3-coloring
declare $Q$ to be 3-colorable, and else declare $Q$ as not having
a 3-coloring.

If $Q$ is not 3-colorable, then clearly $\mbox{ALG}$ cannot 3-color $G$.
What we need to prove is that if $Q$ is 3-colorable, then $\mbox{ALG}$ will
indeed find a 3-coloring of $G$. For this we need to show that the
distribution over graphs $G$ constructed in the above manner (we
speak of distributions because $Z$ is a random) is the same (up to
some small statistical distance) as a distribution that can be generated
by an adversary in the $H_{R}/P_{A}$ model (otherwise we do not know
what $\mbox{ALG}$ would answer given $G$). The difficulty is that the adversary
does not control the host graph (which is random) and its only power
is in choosing the planted coloring. But still, given a random $G_{n,p}$
host graph $H$, we propose the following three step procedure for
the (exponential time) adversary.
\begin{enumerate}
\item If $Q$ is not a vertex induced subgraph of $H$ then {\em fail},
and plant a random 3-coloring in $H$.
\item Else, pick a random vertex-induced copy of $Q$ in $H$ (note that
$H$ could contain more than one copy of $Q$). Let $Z'$ denote the
graph induced on the remaining part of $H$. Let $\chi$ be a balanced
3-coloring of $Q$. If there are two vertices $u,v\in Q$ with $\chi(u)\not=\chi(v)$
which have a common neighbor $w\in Z'$, then {\em fail}, and plant
a random 3-coloring in $H$.
\item Else, extend $\chi$ to a balanced planted 3-coloring on the whole
of $H$, while taking care that if a vertex $w\in Z'$ has a neighbor
$v\in Q$, then $\chi(w)=\chi(v)$. After dropping the monochromatic
edges, one gets a graph $G'$ composed of two components: $Q$, and
another component that we call $Z_{3}'$.
\end{enumerate}
What needs to be shown is that the probability (over choice of $H\in_{R}G_{n,p}$)
of failing in either step~1 or~2 is small, and that conditioned
on not failing, $G'$ has a distribution similar to that of $G$ (equivalently,
$Z_{3}'$ has a distribution similar to $Z_{3}$).

For the first step not to fail, one needs graphs $Q\in{\cal Q}$ to
occur in a random $G_{n,p}$ graph by chance rather than by design.
This requires average degree $d>n^{1/3}$, as shown by the following
proposition (the proof can be found in Section~\ref{sec:Proof-of-Proposition sparce copy}).
\begin{prop}
\label{pro:sparseQ} Let ${\cal Q}$ be an arbitrary class of graphs.
Then either there is a polynomial time algorithm for solving 3-colorability
on every graph in ${\cal Q}$, or ${\cal Q}$ contains graphs that
are unlikely to appear as subgraphs of a random graph from $G_{n,p}$,
if $p\le n^{-\frac{2}{3}}$ (namely, if the average degree
is $d\le n^{\frac{1}{3}}$).
\end{prop}
Moreover, it is not hard to show that if the average degree is too
large, namely, $d>n^{\frac{1}{2}}$, the second step is likely
to fail (there are likely to be pairs of vertices in $Q$ that have
common neighbors in $Z'$). Hence we restrict attention to average
degrees satisfying $n^{\frac{1}{3}}<d<n^{\frac{1}{2}-\epsilon}$.
Consequently, the class of 4-regular graphs cannot serve as the NP-hard
class ${\cal Q}$ (in contrast to the proof of Theorem~\ref{thm:A/A}),
because any particular 4-regular graph is expected not to be a subgraph
of a random graph of average degree below $\sqrt{n}$. Given the above,
we choose ${\cal Q}$ to be the class of {\em balanced} graphs
of average degree $3.75$. (The choice of $3.75$ is for convenience. With extra work in the proof of Lemma~\ref{lem:3.75}, one can replace $3.75$ by any constant larger than $10/3$ and consequently decrease the value of $\delta_0$ in Lemma~\ref{lem:Q} to any constant above $2/5$.)
\begin{defn}
\label{def:balancedgraph} A graph $Q$ is {\em balanced} if no
subgraph of $Q$ has average degree larger than the average degree
of $Q$.
\end{defn}
This choice of ${\cal Q}$ is justified by the combination of Lemmas~\ref{lem:3.75}
and~\ref{lem:Q}.
\begin{lem}
\label{lem:3.75} The problem of 3-coloring is NP-hard on the class
of balanced graphs of average degree $3.75$.
\end{lem}

\begin{lem}
\label{lem:Q} Let $Q$ be an arbitrary balanced graph on $n^{\epsilon}$
vertices and of average degree $3.75$. 
Then a random graph of degree
$d>n^{\delta_0}$ is likely to contain $Q$ as a vertex induced subgraph, for $\delta_0 = \frac{7 + 16\epsilon}{15}$.
\end{lem}

The above implies that if ${\cal Q}$ is the class of balanced graphs
of average degree $3.75$ and if $n^{\delta_0}<d<n^{\frac{1}{2}-\epsilon}$ (this interval is nonempty for $\epsilon < \frac{1}{62}$),
then the adversary is likely not to fail in steps~1 and~2. What
remains to prove is that the distribution over $Z_{3}'$ is statistically
close to the distribution over $Z_{3}$. (It may seem strange that this is needed, given that $Q$ is isolated from $Z_3/Z_3'$.
However, as $Z_3'$ is constructed by a procedure that depends on $Q$ and is not known to run in polynomial time, we need to argue that
it does not contain information that can be used by a 3-coloring algorithm for $Q$. Being statistically close to the polynomial time constructible $Z_3$ serves this purpose.)
Proving this claim is nontrivial. Our proof for this claim is partly inspired
by work of~\cite{Juels98hidingcliques} that showed that the distribution
of graphs from $G_{n,\frac{1}{2}}$ in which one plants a random clique
of size $\frac{3}{2}\log n$ is statistically close to the $G_{n,\frac{1}{2}}$
distribution (with no planting).

Full proofs for this section appear in Section~\ref{sec:G Hardness of 3-coloring random graphs with adversarial planted 3-coloring}
(see Theorem~\ref{thm:Hardnes adv planting on gnp full}).

\subsection{Some open questions}

An intriguing question left open by Theorem~\ref{thm:R/A}b is the
following:

{\em Is there a polynomial time 3-coloring algorithm for the $H_{R}/P_{A}$
model at average degrees significantly below $n^{\frac{1}{3}}$?}

We believe that the hosted coloring framework (and similar frameworks
for other NP-hard problems) is a fertile ground for further research.
A fundamental question is whether the choice of host graph matters
at all. Specifically:

{\em Does Theorem~\ref{thm:A/R} (existence of 3-coloring algorithms
for $H_{A}/P_{R}$) continue to hold if the class of host graphs is
that of all regular graphs (with no restriction on $\lambda$)?}

The answer to the above question is positive for host graphs of minimum degree linear in $n$ (left as exercise to the reader), and it will be very surprising if the answer is positive in general.

\subsection*{Acknowledgements}
Work supported in part by the Israel Science Foundation (grant No.
621/12) and by the I-CORE Program of the Planning and Budgeting Committee
and the Israel Science Foundation (grant No. 4/11).

\bibliographystyle{plain}
\bibliography{bib}

\appendix

\section{\label{sec:Organization of the technical sections}Organization of the technical sections }

The technical sections are organized in an order somewhat different than that of the main text. Specifically, we consider the models $H_A/P_R$, $H_A/P_A$ and $H_R/P_A$ one by one (the model $H_R/P_R$ need not be considered because it is handled in~\cite{alon1997spectral}). Consequently, different parts of the proof of Theorem~\ref{thm:partial} appear in different sections.


Section~\ref{app:AR} contains all proofs for the model with an adversarial expander host graph and a random planted coloring, namely, $H_A/P_R$.
It is divided into subsections as follows.
Section~\ref{sec:The eigenvalues of an expander with a planted coloring} analyzes the spectrum of an expander graph with a random
planted coloring (this relates to Lemma~\ref{lem:step1} in the $H_{A}/P_{R}$ model). Specifically,
we show a generalization to Proposition~2.1 in \cite{alon1997spectral},
see Theorem~\ref{thm:Main}.
Section~\ref{sec:Coloring expander graphs with random color planting}
presents in a unified manner the algorithm that 3-color graphs in the $H_{A}/P_{R}$ model, combining into one algorithm the parts of Theorem~\ref{thm:partial} related to the $H_A/P_R$ model and Theorem~\ref{thm:A/R}.
The proofs related to this algorithm are presented in Section~\ref{sub:Proof-Of-Theorem MAin Alg}.
Our first
step is to show (using Theorem~\ref{thm:Main}) that a spectral clustering
algorithm gives a good approximation to the planted coloring. We provide
a new spectral clustering algorithm that is based on sampling within the lowest eigenvectors, leading to improved efficiency
(in the case of random $k$-color planting,
the algorithm's expected running time is roughly $e^k n$ as
opposed to $n^{k}$ in \cite{alon1997spectral}) -- see
Lemma~\ref{lem:Eigenvector approximation}. In the rest of the section
we show how to attain a full legal coloring. Our proofs of Lemmas~\ref{lem:step2} and~\ref{lem:step3} are formulated in terms of a parameter $sb$ that upper bounds the number of vertices with ``bad statistics". The actual value of this parameter depends on whether
the planted coloring is random (as in Lemma~\ref{lem:SB bound the random planting Case}) or adversarial, but other than that the proofs do not assume that the planting is random. Instead they use the expander mixing
lemma. The place where we do use the randomness of the planted coloring
is in Lemma~\ref{lem:Safe Recoloring} that shows that after the safe
recoloring stage the remaining free vertices break up into connected
components of size $O\left(\frac{\lambda^{2}}{d^{2}}\log n\right)$.
(This bound is smaller and hence better than the bound of $\log_{d}n$ in~\cite{alon1997spectral}).
The analysis of Lemma~\ref{lem:Safe Recoloring} is tailored to the
randomness of the planted coloring and to the recoloring step that we
added.

Section~\ref{app:AA} contains all proofs for the model with an adversarial expander host graph and an adversarial planted balanced coloring, namely, $H_A/P_A$.
Section~\ref{sec:Coloring Expander Graphs With Adversarial Planting} shows how to obtain a $b$-partial coloring for graphs in the $H_{A}/P_{A}$ model (Theorem~\ref{thm:partial}). For $H_A/P_A$, Lemma~\ref{lem: eigen values after arbitrary planting} replaces Theorem~\ref{thm:Main} in the spectral analysis. Section~\ref{sec:On the hardness of coloring expander graphs with adversarial planting}
proves Theorem~\ref{thm:A/A}
(hardness result for the $H_{A}/P_{A}$ model).

Section~\ref{app:RA} contains all proofs for the model with a random host graph and an adversarial planted balanced coloring, namely, $H_R/P_A$.
Section~\ref{sec:F 3-coloring random graphs with adversarial color planting}
proves Theorem~\ref{thm:R/A}a (a 3-coloring algorithm for $H_R/P_A$). Section~\ref{sec:G Hardness of 3-coloring random graphs with adversarial planted 3-coloring}
proves Theorem~\ref{thm:R/A}b (hardness for $H_R/P_A$).

Section~\ref{sec:k>3} discusses extensions of our results to $k > 3$. The main new proof there is in Section~\ref{sub:G.2 Hardness result for random graphs with adversarial 4-color planting} which shows a hardness result for $k$-coloring in the $H_R/P_A$ model (for
$k>3$). The range of degrees for this proof is not upper bounded (unlike the case of 3-coloring in Theorem~\ref{thm:R/A}). This implies that
Theorem~\ref{thm:R/A}a does not extend to $k\ge4$,
unless $P=NP$.

Section~\ref{sec:counterintuitive} demonstrates (within our context of planted models) that even if a graph $G$ contains a vertex induced subgraph that is difficult to 3-color, this by itself does not imply that it is difficult to 3-color $G$. It is useful to bear this fact in mind when trying to prove NP-hardness results within our framework.

Section~\ref{sec:Useful-Facts} contains some useful facts (Chernoff bounds, the expander mixing lemma, and more) that are used throughout this manuscript.

\subsection{Extended notations and definitions}
\label{sec:notation}

Given a graph $G$ we denote its adjacency matrix by $A_{G}$. We
denote $A_{G}$'s normalized eigenvectors by $e_{1}\left(G\right),e_{2}\left(G\right),..,e_{n}\left(G\right)$
and the corresponding real eigenvalues by $\lambda_{1}\left(G\right)\ge\lambda_{2}\left(G\right)\ge...\ge\lambda_{n}\left(G\right)$.
We denote by $V\left(G\right)$ the vertex set of the graph and by
$E\left(G\right)$ the edge set of the graph. We denote by $d_{v}$
the degree of $v\in V\left(G\right)$. For any $H\subseteq V\left(G\right)$
let $G_{H}$ be the induced sub-graph of $G$ on the vertices of $H$.
Given a vertex $v\in V\left(G\right)$ we denote by $N\left(v\right)$
the neighborhood of $v$, excluding $v$. For a set $S\subseteq V$,
$N\left(S\right)$ denotes the neighborhood of $S$, i.e. , $N\left(S\right)=\cup_{s\in S}N\left(s\right)$
(note that with this definition we can have $S\cap N\left(S\right)\neq\emptyset$).
For $S,T\subseteq V\left(G\right)$, $E_{G}\left(S,T\right)$ denotes
the number of edges between $S$ and $T$ in $G$. If $S,T$ are not
disjoint then the edges in the induced sub-graph of $S\cap T$ are
counted twice.

Given a vector $\vec{v}\in\mathbb{R}^{n}$ we denote by $\left(\vec{v}\right)_{i}$
the $i$-th coordinate of $\vec{v}$.

\begin{defn}
{[}Graph Sparsity{]}. \label{def:=00005BGraph-Sparsity=00005D}Let
$G=\left(V,E\right)$ be a $d$-regular graph, $S\subseteq V$. The
sparsity of $S$ is
\[
\phi\left(S\right):=\frac{E_{G}\left(S,V-S\right)}{\frac{d}{n}\left|S\right|\left|V-S\right|},
\]
and the sparsity of the graph $G$ is
\[
\phi\left(G\right):=min_{S\subseteq V}\phi\left(S\right)
\]

\end{defn}

\begin{defn}
{[}Vertex expander{]}.\label{def:=00005BVertex-expander=00005D.}
A graph $G$ is an $\left(s,\alpha\right)$ vertex expander if for
every $S\subseteq V\left(G\right)$ of size at most $s$ it holds
that the neighborhood of $S$ is of size at least $\alpha\left|S\right|$.
\end{defn}

\begin{defn}
{[}$\lambda-expander${]}. \label{def:expander}A graph $G$ is a
$\lambda$-expander if $\max\left(\lambda_{2}\left(G\right),\left|\lambda_{n}\left(G\right)\right|\right) \le \lambda$.
\end{defn}

\begin{defn}
Let $P_{3}\left(G\right):=\left(V_{1}\cup V_{2}\cup V_{3},E'\right)$
be the following random graph:
\begin{enumerate}
\item Each vertex of $V$ is in each $V_{i}$ with probability $\frac{1}{3}$,
independently over all the vertices.
\item Each edge of $\left(u,v\right)\in E$ is in $E'$ iff $u\in V_{i}$
and $v\in V_{j}$ for $i\ne j$.
\end{enumerate}
\end{defn}

\begin{defn}
For $v\in V_{i}$ we say $v$ is colored with $i$. We call this coloring
of $P_{3}\left(G\right)$ the planted coloring.
\end{defn}

\begin{defn}
Given $P_{3}\left(G\right)$ and $i \in \{1,2,3\}$ we denote by $G_{i}$ the subgraph of
$G$ induced on the vertex set $V_{i}$.
\end{defn}

\begin{defn}
\label{Def:x,y}Define the following vectors in $\mathbb{R}^{n}$ (where $i$ ranges from~1 to~$n$).

$\left(\vec{1}_{n}\right)_{i}=1$, $\left(\vec{p_1}\right)_{i}:=\begin{cases}
1 & v_{i}\in V_1\\
0 & v_{i}\notin V_1
\end{cases}$, $\left(\vec{p_2}\right)_i:=\begin{cases}
1 & v_{i}\in V_2\\
0 & v_{i}\notin V_2
\end{cases}$, $\left(\vec{p_3}\right)_i:=\begin{cases}
1 & v_{i}\in V_3\\
0 & v_{i}\notin V_3
\end{cases}$,

$\left(\vec{x}\right)_{i}:=\begin{cases}
2 & v_{i}\in V_{1}\\
-1 & v_{i}\in V_{2}\\
-1 & v_{i}\in V_{3}
\end{cases}$ and $\left(\vec{y}\right)_{i}:=\begin{cases}
0 & v_{i}\in V_{1}\\
1 & v_{i}\in V_{2}\\
-1 & v_{i}\in V_{3}
\end{cases}$.

We denote by $\bar{v}$ a normalized vector, i.e., $\bar{x}:=\vec{x}/\left|\vec{x}\right|_{2}$.
\end{defn}

\section{Adversarial host and random planting, $H_A/P_R$}
\label{app:AR}

Section~\ref{app:AR} contains the proofs for the model with an adversarial expander host graph and a random planted coloring, namely, $H_A/P_R$.
Section~\ref{sec:The eigenvalues of an expander with a planted coloring} analyzes the spectrum of an expander graph with a random
planted coloring (this relates to Lemma~\ref{lem:step1} in the $H_{A}/P_{R}$ model).
Section~\ref{sec:Coloring expander graphs with random color planting}
presents in a unified manner the algorithm that 3-colors graphs in the $H_{A}/P_{R}$ model, combining into one algorithm the parts of Theorem~\ref{thm:partial} related to the $H_A/P_R$ model and Theorem~\ref{thm:A/R}.
The proofs related to this algorithm are presented in Section~\ref{sub:Proof-Of-Theorem MAin Alg}.

\subsection{\label{sec:The eigenvalues of an expander with a planted coloring}The
eigenvalues of an expander with a random planted coloring}

\begin{thm}
\label{thm:Main}Let $G$ be a $d$ regular $\lambda$-expander, where
$\lambda\le\left(\frac{1}{3}-\epsilon\right)d$ (for $0<\epsilon<\frac{1}{3}$).
\textup{Let $G'=P_{3}\left(G\right)$ (the graph $G$ after a random
3-color planting).  Then with high probability the following holds.}

\begin{enumerate}
\item The eigenvalues of $G'$ have the following spectrum.

\begin{enumerate}
\item $\lambda_{1}\left(G'\right)\ge\left(1-2^{-\Omega(d)}\right)\frac{2}{3}d$.
\item $\lambda_{n}\left(G'\right)\le\lambda_{n-1}\left(G'\right)\le-\frac{1}{3}d\left(1-\frac{3}{\sqrt{d}}\right)$.
\item $|\lambda_{i}\left(G'\right)|\le2\lambda+O\left(\sqrt{d}\right)$
for all $2\le i\le n-2$.
\end{enumerate}
\item The following vectors exist.

\begin{enumerate}
\item $\vec{\epsilon}_{\bar{x}}$ such that $\left\Vert \vec{\epsilon}_{\bar{x}}\right\Vert _{2}=O\left(\frac{1}{\sqrt{d}}\right)$
and $\bar{x}+\vec{\epsilon}_{\bar{x}}\in\text{span}\left(\left\{ e_{i}\left(G'\right)\right\} _{i\in\left\{ n-1,n\right\} }\right)$.
\item $\vec{\epsilon}_{\bar{y}}$ such that $\left\Vert \vec{\epsilon}_{\bar{y}}\right\Vert _{2}=O\left(\frac{1}{\sqrt{d}}\right)$
and $\bar{y}+\vec{\epsilon}_{\bar{y}}\in\text{span}\left(\left\{ e_{i}\left(G'\right)\right\} _{i\in\left\{ n-1,n\right\} }\right)$.
\end{enumerate}
\end{enumerate}
\end{thm}

In the rest of this section we prove Theorem~\ref{thm:Main}.

\begin{lem}
\label{lem:Existance of almost eiganvectors}Let $G'$ be as in Theorem~\ref{thm:Main}.
With high probability the vectors $\bar{1},\bar{x},\bar{y}$, see
definition~\ref{Def:x,y}, satisfy $A_{G'}\bar{x}=\left(-\frac{1}{3}d\right)\bar{x}+\vec{\delta}_{\bar{x}}$,
$A_{G'}\bar{y}=\left(-\frac{1}{3}d\right)\bar{y}+\vec{\delta}_{\bar{y}}$
and $A_{G'}\bar{1}=\left(\frac{2}{3}d\right)\bar{1}+\vec{\delta}_{\bar{1}}$.

Here $\vec{\delta}_{\bar{x}},\vec{\delta}_{\bar{y}},\vec{\delta}_{\bar{1}}$
are vectors with $\ell_{2}$ norm of at most $\sqrt{d}$.\end{lem}
\begin{proof}
We prove the lemma for $\bar{1},\vec{\delta}_{\bar{1}}$ (the other cases have
a similar proof). We assume $d=o\left(\sqrt{n}\right)$ since for large
values of $d$ one can show a different proof using the union bound,
details omitted.

Consider the vector $A_{G'}\vec{1}-\left(\frac{2}{3}d\right)\vec{1}$.
We give an upper bound (that holds with high probability) on the sum of squares of its coordinates.

Let $i\in\left[n\right]$. As $\mathbb{E}\left[\left(\left(A_{G'}\vec{1}\right)_{i}-\frac{2}{3}d\right)^{2}\right]$
can be seen as the variance of $d$ independent Bernoulli variables,
(each with variance $\frac{2}{3}\left(1-\frac{2}{3}\right)$), it
holds that
\[
\mathbb{E}\left[\sum_{i}\left(\left(A_{G'}\vec{1}\right)_{i}-\frac{2}{3}d\right)^{2}\right]=n\mathbb{E}\left[\left(\left(A_{G'}\vec{1}\right)_{1}-\frac{2}{3}d\right)^{2}\right]=nd\frac{2}{9}.
\]

Let $x_{i}$ be a random variable that indicates the color of a the
$i$-th vertex of the graph. Set $f=\sum_{i} \left( \left(A_{G'}\vec{1}\right)_{i}-\frac{2}{3}d\right)^{2}$
to be a function of $x_{i}$'s. Consider the terms of $f$ that are effected by a particular random variable $x_{i}$. One such term is $\left(\left(A_{G'}\vec{1}\right)_{i}-\frac{2}{3}d\right)^{2}$. Its value is bounded between~0 and $\frac{4}{9}d^2$, and hence $x_i$ can effect it by at most $\frac{4}{9}d^2$. Other terms that are affected by $x_i$ are $\left(\left(A_{G'}\vec{1}\right)_{j}-\frac{2}{3}d\right)^{2}$, where $j$ ranges over neighbors of $i$ in $G$. As $i$ has degree $d$ in $G$, there are $d$ such terms. For each neighbor $j$, the value $\left(A_{G'}\vec{1}\right)_{j}-\frac{2}{3}d$ (which lies between $\frac{d}{3}$ and $-\frac{2d}{3}$) changes by at most~1 by $x_i$, implying that $\left(\left(A_{G'}\vec{1}\right)_{j}-\frac{2}{3}d\right)^{2}$ changes by at most $\frac{2d}{3}$. Overall, the effect of $x_i$ on $\sum_{j\in N(i)}\left(\left(A_{G'}\vec{1}\right)_{j}-\frac{2}{3}d\right)^{2}$ is at most $\frac{2d^2}{3}$. As $x_i$ does not effect any other term in $f$, its total effect on $f$ is at most $\frac{2d^2}{3} + \frac{4d^2}{9} = \frac{10d^2}{9} < 2d^2$.

 By the above we can apply McDiarmid's Inequality (Theorem~\ref{thm:=00005BMcDiarmid=00005D})
with $c=2d^{2}$ to deduce
\[
\Pr\left[\left|\sum_{i}\left(\left(A_{G'}\vec{1}\right)_{i}-\frac{2}{3}d\right)^{2}-nd\frac{2}{9}\right|\ge c_{1}nd\right]\le2\exp\left(-\Omega\left(\frac{n}{d^{2}}\right)\right)\,,
\]
where $c_{1}$ is a constant, say $\frac{1}{9}$. Overall , with high
probability, $\left\Vert A_{G'}\vec{1}-\left(\frac{2}{3}d\right)\vec{1}\right\Vert _{2}\le\sqrt{nd}$
and thus, by normalization, $\left\Vert A_{G'}\bar{1}-\left(\frac{2}{3}d\right)\bar{1}\right\Vert _{2}\le\sqrt{d}$.

\end{proof}
\begin{lem}
\label{lem:Almost Eigan Vectors implies more}Let $A$ be a symmetric
matrix of order $n$ with eigenvalues $\lambda_{1},\ldots,\lambda_{n}$
and associated eigenvectors $v_{1}\ldots,v_{n}$. For a unit vector
$\bar{u}$ and scalar $\lambda$, suppose that $A\bar{u}=\lambda\bar{u}+\vec{\epsilon}$,
where the $\ell_{2}$ norm of $\vec{\epsilon}$ is at most $\epsilon$.

Let $S_{\delta}\subset\{1,\ldots,n\}$ be the set of those indices
$i$ for which $|\lambda_{i}-\lambda|\le\delta$. Then there exists
a vector $\vec{\epsilon}_{\bar{u}}$ such that $\bar{u}+\vec{\epsilon}_{\bar{u}}$
is a vector spanned by the eigenvectors associated with $S_{\delta}$
and $\left\Vert \vec{\epsilon}_{\bar{u}}\right\Vert _{2}\le\frac{\epsilon}{\delta}$. \end{lem}
\begin{proof}
Write $\bar{u}=\sum_{i=1}^{n}\alpha_{i}v_{i}$, with $\sum(\alpha_{i})^{2}=1$.
Then
\[
A\bar{u}=\sum\lambda_{i}\alpha_{i}v_{i}=\lambda\bar{u}+\sum(\lambda_{i}-\lambda)\alpha_{i}v_{i}.
\]
Hence $\vec{\epsilon}=\sum(\lambda_{i}-\lambda)\alpha_{i}v_{i}$,
implying $\sum(\lambda_{i}-\lambda)^{2}(\alpha_{i})^{2}\le\epsilon^{2}$.
By averaging, it follows that $\sum_{i\not\in S_{\delta}}(\alpha_{i})^{2}<\frac{\epsilon^{2}}{\delta^{2}}$.
Let $\vec{\epsilon}_{\bar{u}}:=-\sum_{i\not\in S_{\delta}}\alpha_{i}v_{i}$.
\end{proof}

\begin{lem}
\label{lem:distance to S_delta}Let $G'$ be as in Theorem~\ref{thm:Main},
and let the unite vectors $\bar{1},\bar{x},\bar{y}$ be as in definition~\ref{Def:x,y}.
Define $S_{\delta}\left(x\right)\subset\{1,\ldots,n\}$ to be the
set of those indices $i$ for which $|\lambda\left(G'\right)_{i}-\left(x\right)|\le\delta$.
The following vectors exist with high probability.
\begin{itemize}
\item $\vec{\epsilon}_{\bar{1}}$ such that $\left\Vert \vec{\epsilon}_{\bar{1}}\right\Vert _{2}\le\frac{\sqrt{d}}{\delta}$
and $\bar{1}+\vec{\epsilon}_{\bar{1}}\in\text{span}\left(\left\{ e_{i}\left(G'\right)\right\} _{i\in S_{\delta}\left(\frac{2}{3}d\right)}\right)$.
\item $\vec{\epsilon}_{\bar{x}}$ such that $\left\Vert \vec{\epsilon}_{\bar{x}}\right\Vert _{2}\le\frac{\sqrt{d}}{\delta}$
and $\bar{x}+\vec{\epsilon}_{\bar{x}}\in\text{span}\left(\left\{ e_{i}\left(G'\right)\right\} _{i\in S_{\delta}\left(-\frac{1}{3}d\right)}\right)$.
\item $\vec{\epsilon}_{\bar{y}}$ such that $\left\Vert \vec{\epsilon}_{\bar{y}}\right\Vert _{2}\le\frac{\sqrt{d}}{\delta}$
and $\bar{y}+\vec{\epsilon}_{\bar{y}}\in\text{span}\left(\left\{ e_{i}\left(G'\right)\right\} _{i\in S_{\delta}\left(-\frac{1}{3}d\right)}\right)$.
\end{itemize}
\end{lem}
\begin{proof}
By Lemma~\ref{lem:Existance of almost eiganvectors} and Lemma~\ref{lem:Almost Eigan Vectors implies more}
the statement of the lemma holds with high probability.
\end{proof}

\begin{lem}
\label{lem:Random Induced Graph is almost regular.}Let $G$ be a
$d$-regular $\lambda$-expander. Let $G_{1}$ be a random induced
graph of $G$ where each vertex of $G$ is in $G_{1}$ independently
with probability $\frac{1}{3}$. With high probability the vector
$\bar{1}_{\left|V\left(G_{1}\right)\right|}$, see definition~\ref{Def:x,y},
satisfies $A_{G_{1}}\bar{1}_{\left|V\left(G_{1}\right)\right|}=\left(\frac{1}{3}d\right)\bar{1}_{\left|V\left(G_{1}\right)\right|}+\vec{\delta}_{\bar{1}}$.

Here $\vec{\delta}_{\bar{1}}$ is a vector with $\ell_{2}$ norm of
at most $\sqrt{d}$.\end{lem}
\begin{proof}
The proof is similar to the proof of Lemma~\ref{lem:Existance of almost eiganvectors}.\end{proof}
\begin{lem}
\label{lem:Random Induced Graph is expander}Let $G$ be a $d$-regular
$\lambda$-expander. Let $G_{1}$ be a random induced graph of $G$
where each vertex of $V\left(G\right)$ is in $V\left(G_{1}\right)$
independently with probability $\frac{1}{3}$. It holds that $G_{1}$
is a $\lambda$-expander.\end{lem}
\begin{proof}
The proof follows by Cauchy interlacing theorem, (Theorem~\ref{thm:=00005BCauchy-interlacing-theorem=00005D.}).\end{proof}
\begin{lem}
\label{lem:distance to S_delta in G_1}Let $G'$ be as in Theorem~\ref{thm:Main},
the vector $\bar{p}_{1}$ be as in Definition~\ref{Def:x,y}, and
$\delta\le\frac{d}{3}-\sqrt{d}-\lambda$ . The following vector exists
with high probability.
\begin{itemize}
\item \textup{$\vec{\epsilon}_{\bar{p}_{1}}$}such that $\left\Vert \vec{\epsilon}_{\bar{p}_{1}}\right\Vert _{2}\le\frac{\sqrt{d}}{\delta}$
and $\bar{p}_{1}+\vec{\epsilon}_{\bar{p}_{1}}\in\text{span}\left(\left\{ e_{1}\left(G_{1}\right)\right\} \right)$.
\end{itemize}
\end{lem}
\begin{proof}
Let $S_{\delta}\left(G_{1}\right)$ be the set of those indices $i$
for which $|\lambda\left(G_{1}\right)_{i}-\frac{1}{3}d|\le\delta$.
By Lemma~\ref{lem:Random Induced Graph is expander} it follows that
$S_{\delta}\left(G_{1}\right)=\left\{ 1\right\} $. By Lemma~\ref{lem:Random Induced Graph is almost regular.}
and Lemma~\ref{lem:Almost Eigan Vectors implies more} the following
vector exists with high probability. $\vec{\epsilon}_{\bar{p}_{1}}$such
that $\left\Vert \vec{\epsilon}_{\bar{p}_{1}}\right\Vert _{2}\le\frac{\sqrt{d}}{\delta}$
and
\[
\bar{p}_{1}+\vec{\epsilon}_{\bar{p}_{1}}\in\text{span}\left(\left\{ e_{i}\left(G_{1}\right)\right\} _{i\in S_{\delta}\left(G_{1}\right)}\right)=\text{span}\left(\left\{ e_{1}\left(G_{1}\right)\right\} \right)\,.
\]

\end{proof}

Now we prove Theorem~\ref{thm:Main}.
\begin{proof}
{[}Theorem~\ref{thm:Main}{]}. Let $\delta=\frac{\epsilon}{2}d$
and $\vec{\epsilon}_{\bar{1}},\vec{\epsilon}_{\bar{x}},\vec{\epsilon}_{\bar{y}},\vec{\epsilon}_{\bar{p}_{1}}$
be the vectors that were defined in Lemma~\ref{lem:distance to S_delta}
and Lemma~\ref{lem:distance to S_delta in G_1}. Define $S_{\delta}\subset\{1,\ldots,n\}$
to be the set of those indices $i$ for which either $|\lambda\left(G'\right)_{i}-\left(-\frac{1}{3}d\right)|\le\delta$
or $|\lambda\left(G'\right)_{i}-\left(\frac{2}{3}d\right)|\le\delta$.
We conclude the proof by the following two claims.
\begin{claim}
\label{Claim:Bounding eiganvalues}For $\left\{ i\right\} \notin S_{\delta}$
it holds that $|\lambda\left(G'\right)_{i}|\le2\lambda+O\left(\frac{d^{1.5}}{\delta}\right)$.\end{claim}
\begin{proof}
Let $A_{\ensuremath{G_{j}}}'$ be the adjacency matrix of the induced
sub-graph on $j$-colored vertices, i.e.,
\[
\left(A_{\ensuremath{G_{j}}}'\right)_{i,k}=\begin{cases}
\left(A_{\ensuremath{G_{j}}}\right)_{i,k} & v_{i}\in V(G_{j}),v_{k}\in V(G_{j})\\
0 & \text{Otherwise.}
\end{cases}.
\]
Let
\[
\left(e_{i}\left(G_{j}\right)'\right)_{k}=\begin{cases}
\left(e_{i}\left(G_{j}\right)\right)_{k} & v_{k}\in V\left(V\left(G_{j}\right)\right)\\
0 & \text{Otherwise.}
\end{cases}.
\]
For brevity, the last two notation are used without the apostrophe.
Note that $A_{G'}=A_{G}-\left(\Sigma_{j=1}^{3}A_{G_{j}}\right)$.

~

Let $i\notin S_{\delta}$, by the triangle inequality it holds that,

\begin{equation}
\left\Vert A_{G'}e_{i}\left(G'\right)\right\Vert _{2}\le\left\Vert A_{G}e_{i}\left(G'\right)\right\Vert _{2}+\left\Vert \sum_{j=1}^{3}A_{G_{j}}e_{i}\left(G'\right)\right\Vert _{2}.\label{eq:1}
\end{equation}

First we show $\left\Vert A_{G}e_{i}\left(G'\right)\right\Vert _{2}\le\lambda+O\left(\frac{d^{1.5}}{\delta}\right)$.
\begin{align*}
\left\Vert A_{G}e_{i}\left(G'\right)\right\Vert _{2}^{2} & =\Sigma_{j=1}^{n}\lambda_{j}^{2}\left(G\right)\left\langle e_{i}\left(G'\right),e_{j}\left(G\right)\right\rangle ^{2}\\
 & =\lambda_{1}^{2}\left(G\right)\left\langle e_{i}\left(G'\right),e_{1}\left(G\right)\right\rangle ^{2}+\Sigma_{j=2}^{n}\lambda_{j}^{2}\left(G\right)\left\langle e_{i}\left(G'\right),e_{j}\left(G\right)\right\rangle ^{2}\\
 & \le\lambda_{1}^{2}\left(G\right)\left\langle e_{i}\left(G'\right),e_{1}\left(G\right)\right\rangle ^{2}+\lambda^{2}.
\end{align*}

The last inequality follows since $G$ is a $\lambda$-expander.
\begin{align}
 & =d^{2}\left\langle e_{i}\left(G'\right),\bar{1}\right\rangle ^{2}+\lambda^{2}\nonumber \\
 & =d^{2}\left\langle e_{i}\left(G'\right),\bar{1}+\vec{\epsilon}_{\bar{1}}-\vec{\epsilon}_{\bar{1}}\right\rangle ^{2}+\lambda^{2}\nonumber \\
 & =d^{2}\left\langle e_{i}\left(G'\right),\vec{\epsilon}_{\bar{1}}\right\rangle ^{2}+\lambda^{2}\label{eq:2}
\end{align}

The last equality holds since $\bar{1}+\vec{\epsilon}_{\bar{1}}$
and $e_{i}\left(G'\right)$ are orthogonal.
\[
\le d^{2}\left\Vert \vec{\epsilon}_{\bar{1}}\right\Vert _{2}^{2}+\lambda^{2}
\]

The last inequality follows by Cauchy\textendash Schwarz inequality
and now we use the assumption on $\left\Vert \vec{\epsilon}_{\bar{1}}\right\Vert _{2}$.

\[
\le d^{2}\left(\frac{\sqrt{d}}{\delta}\right)^{2}+\lambda^{2}.
\]

We left to show that with high probability $\left\Vert \sum_{j=1}^{3}A_{G_{j}}e_{i}\left(G'\right)\right\Vert _{2}\le\lambda+O\left(\frac{d^{1.5}}{\delta}\right)$.
Denote by $n_{j}$ the $G_{j}$'s number of vertices. Note that $\sum_{j=1}^{3}A_{G_{j}}$
is the adjacency matrix of a graph with at least three connected components,
($\left\{ G_{j}\right\} _{j\in\left\{ 1,2,3\right\} }$), thus $\left\{ e_{i}\left(G_{j}\right)|1\le i\le n_{j},j\in\left\{ 1,2,3\right\} \right\} $
are the eigenvectors of $\sum_{j=1}^{3}A_{G_{j}}$, (up to appropriate
padding with zeros).

\begin{align*}
\left\Vert \sum_{j=1}^{3}A_{G_{j}}e_{i}\left(G'\right)\right\Vert _{2}^{2} & =\sum_{j=1}^{3}\sum_{k=1}^{n_{j}}\lambda_{k}^{2}\left(G_{j}\right)\left\langle e_{i}\left(G'\right),e_{k}\left(G_{j}\right)\right\rangle ^{2}\\
 & =\sum_{j=1}^{3}\lambda_{1}^{2}\left(G_{j}\right)\left\langle e_{i}\left(G'\right),e_{1}\left(G_{j}\right)\right\rangle ^{2}+\sum_{j=1}^{3}\sum_{j=2}^{n_{j}}\lambda_{j}^{2}\left(G_{1}\right)\left\langle e_{i}\left(G'\right),e_{j}\left(G_{1}\right)\right\rangle ^{2}\\
 & \le d^{2}\sum_{j=1}^{3}\left\langle e_{i}\left(G'\right),e_{1}\left(G_{j}\right)\right\rangle ^{2}+\lambda^{2}
\end{align*}

The last inequality follows since by Lemma~\ref{lem:Random Induced Graph is expander}
it follows that $\left\{ G_{j}\right\} _{j\in\left\{ 1,2,3\right\} }$
are $\lambda$-expanders with degree at most $d$.

\begin{align*}
=d^{2}\sum_{j=1}^{3}\left\langle e_{i}\left(G'\right),\bar{p}_{j}+\vec{\epsilon}_{\bar{p}_{j}}\right\rangle ^{2}+\lambda^{2} & =d^{2}\sum_{j=1}^{3}\left(\left\langle e_{i}\left(G'\right),\vec{\epsilon}_{\bar{p}_{j}}\right\rangle +\left\langle e_{i}\left(G'\right),\bar{p}_{j}\right\rangle \right)^{2}+\lambda^{2}\\
 & \le d^{2}\sum_{j=1}^{3}\left(\left\Vert \vec{\epsilon}_{\bar{p}_{j}}\right\Vert _{2}+\left\langle e_{i}\left(G'\right),\bar{p}_{j}\right\rangle \right)^{2}+\lambda^{2}
\end{align*}

The last inequality follows by Cauchy\textendash Schwarz inequality.

\[
\le d^{2}\sum_{j=1}^{3}\left(\frac{\sqrt{d}}{\delta}+\left\langle e_{i}\left(G'\right),\bar{p}_{j}\right\rangle \right)^{2}+\lambda^{2}
\]

The last inequality follows by the assumption on $\left|\vec{\epsilon}_{\bar{p}_{1}}\right|_{2}$.

\[
=d^{2}\sum_{j=1}^{3}\left(\frac{\sqrt{d}}{\delta}+\left\langle e_{i}\left(G'\right),c_{1}^{j}\bar{1}+c_{2}^{j}\bar{x}+c_{3}^{j}\bar{y}\right\rangle \right)^{2}+\lambda^{2}
\]

With high probability$\left\{ c_{k}^{j}\right\} _{k,j\in\left\{ 1,2,3\right\} }$
are universal constants. To see that such constants exist with high
probability note that for example $\vec{p_{1}}=\frac{1}{3}\vec{1}+\frac{1}{3}\vec{x}$,
(In general $\vec{1},\vec{x},\vec{y}$ spans $\vec{p}_{1},\vec{p}_{2},\vec{p}_{3}$
using universal constants), thus by normalization $\bar{p}_{1}=\frac{1}{3\sqrt{\left|V_{1}\right|}}\vec{1}+\frac{1}{\sqrt{\left|V_{1}\right|}}\vec{x}=\frac{\sqrt{n}}{3\sqrt{\left|V_{1}\right|}}\bar{1}+\frac{\left\Vert \vec{x}\right\Vert _{2}}{\sqrt{\left|V_{1}\right|}}\bar{x}$
and with high probability $\frac{\sqrt{n}}{3\sqrt{\left|V_{1}\right|}},\frac{\left\Vert \vec{x}\right\Vert _{2}}{\sqrt{\left|V_{1}\right|}}$
are constants.

\begin{align}
 & =d^{2}\sum_{j=1}^{3}\left(\frac{\sqrt{d}}{\delta}+\left\langle e_{i}\left(G'\right),c_{1}^{j}\left(\bar{1}+\vec{\epsilon}_{\bar{1}}-\vec{\epsilon}_{\bar{1}}\right)+c_{2}^{j}\left(\bar{x}+\vec{\epsilon}_{\bar{x}}-\vec{\epsilon}_{\bar{x}}\right)+c_{3}^{j}\left(\bar{y}+\vec{\epsilon}_{\bar{y}}-\vec{\epsilon}_{\bar{y}}\right)\right\rangle \right)^{2}+\lambda^{2}\nonumber \\
 & =d^{2}\sum_{j=1}^{3}\left(\frac{\sqrt{d}}{\delta}+\left\langle e_{i}\left(G'\right),c_{1}^{j}\left(-\vec{\epsilon}_{\bar{1}}\right)+c_{2}^{j}\left(-\vec{\epsilon}_{\bar{x}}\right)\right\rangle +c_{3}^{j}\left(-\vec{\epsilon}_{\bar{y}}\right)\right)^{2}+\lambda^{2}\label{eq:3}
\end{align}

The last equality holds since $\bar{1}+\vec{\epsilon}_{\bar{1}},\bar{x}+\vec{\epsilon}_{\bar{x}},\bar{y}+\vec{\epsilon}_{\bar{y}}$
are orthogonal to $e_{i}\left(G'\right)$.
\[
\le d^{2}\sum_{j=1}^{3}\left(\frac{\sqrt{d}}{\delta}+\left(\sum_{k=1}^{3}\left|c_{k}^{j}\right|\right)\left(\frac{\sqrt{d}}{\delta}\right)\right)^{2}+\lambda^{2}=O\left(\frac{d^{3}}{\delta^{2}}\right)+\lambda^{2}
\]

The last inequality follows by Cauchy\textendash Schwarz inequality
and the assumption on $\left\Vert \vec{\epsilon}_{\bar{1}}\right\Vert _{2},\left\Vert \vec{\epsilon}_{\bar{x}}\right\Vert _{2},\left\Vert \vec{\epsilon}_{\bar{y}}\right\Vert _{2}$.\end{proof}
\begin{claim}
\label{claim:.s_delta=00003D3}$\left|S_{\delta}\right|=3$.\end{claim}
\begin{proof}
Clearly $\bar{1}+\vec{\epsilon}_{\bar{1}},\bar{x}+\vec{\epsilon}_{\bar{x}},\bar{y}+\vec{\epsilon}_{\bar{y}}\in\text{span}\left(\left\{ e_{i}\left(G'\right)\right\} _{i\in S_{\delta}}\right)$
and since they are independent, (the vectors $\bar{1},\bar{x},\bar{y}$
are nearly orthogonal, adding to them the vectors $\vec{\epsilon}_{\bar{1}},\vec{\epsilon}_{\bar{x}},\vec{\epsilon}_{\bar{y}}$,
each of small $\ell_{2}$ norm, doesn't change this.), it follows
that $\left|S_{\delta}\right|\ge3$.

Now we show $\left|S_{\delta}\right|=3$. Let $W=\text{span}\left\{ \bar{1}+\vec{\epsilon}_{\bar{1}},\bar{x}+\vec{\epsilon}_{\bar{x}},\bar{y}+\vec{\epsilon}_{\bar{y}}\right\} $
and let
\[
W^{\perp}:=\left\{ v\in\text{span}\left(\left\{ e_{i}\left(G'\right)\right\} _{i\in S_{\delta}}\right)|\forall w\in W\,,\,\left\langle w,v\right\rangle =0\right\} .
\]
 Assume, towards a contradiction, that exists $0\neq v\in W^{\perp}$.
Note that we can write
\[
v=\sum_{i\in S_{\delta}}\alpha_{i}e_{i}\left(G'\right)\,,
\]
 where $\sum_{i\in S_{\delta}}\alpha_{i}^{2}=1$. Now we give a lower
bound on $\left\Vert A_{G'}v\right\Vert _{2}$.
\begin{align*}
\left\Vert A_{G'}v\right\Vert _{2}^{2} & =\left\langle A_{G'}v,A_{G'}v\right\rangle \\
 & =\left\langle A_{G'}\left(\sum_{i\in S_{\delta}}\alpha_{i}e_{i}\left(G'\right)\right),A_{G'}\left(\sum_{i\in S_{\delta}}\alpha_{i}e_{i}\left(G'\right)\right)\right\rangle \\
 & =\left\langle \sum_{i\in S_{\delta}}\lambda_{i}\left(G'\right)\alpha_{i}e_{i}\left(G'\right),\sum_{i\in S_{\delta}}\lambda_{i}\left(G'\right)\alpha_{i}e_{i}\left(G'\right)\right\rangle \\
 & =\sum_{i\in S_{\delta}}\lambda_{i}^{2}\left(G'\right)\alpha_{i}^{2}\\
 & \ge\left(\left(\frac{1}{3}-\epsilon\right)d\right)^{2}\sum_{i\in S_{\delta}}\alpha_{i}^{2}=\left(\left(\frac{1}{3}-\frac{\epsilon}{2}\right)d\right)^{2}\,.
\end{align*}

The last inequality uses the fact $\delta=\frac{\epsilon}{2}d$. By
the triangle inequality at least one of the following is $\left(\frac{1}{3}-\frac{\epsilon}{2}\right)d$:

$\left\Vert A_{G}v\right\Vert _{2}$ and $\left\Vert \sum_{j=1}^{3}A_{G_{j}}v\right\Vert _{2}$,
see Equation~\ref{eq:1}. But applying the proof Claim~\ref{Claim:Bounding eiganvalues}
shows that $\left\Vert A_{G}v\right\Vert _{2}=\lambda+O\left(\frac{d^{1.5}}{\delta}\right)$
and $\left\Vert \sum_{j=1}^{3}A_{G_{j}}v\right\Vert _{2}=\lambda+O\left(\frac{d^{1.5}}{\delta}\right)$,
(The only requirement from the vector $e_{i}\left(G'\right)$ is to
be orthogonal to $\bar{1}+\vec{\epsilon}_{\bar{1}},\bar{x}+\vec{\epsilon}_{\bar{x}},\bar{y}+\vec{\epsilon}_{\bar{y}}$,
see Equation~\ref{eq:2} and Equation~\ref{eq:3}. $v\in W^{\perp}$
satisfies these requirements by definition). As $\lambda<\left(\frac{1}{3}-\epsilon\right)d$
and since $\delta=\frac{\epsilon}{2}d$ this is a contradiction (for
large enough $d$). Since $W^{\perp}=\phi$ Claim~\ref{claim:.s_delta=00003D3}
follows.
\end{proof}
By Claim~\ref{Claim:Bounding eiganvalues} and Claim~\ref{claim:.s_delta=00003D3}
items~1.(a) and 1.(c) of Theorem~\ref{thm:Main} hold and $-\left(1+\epsilon\right)\frac{1}{3}d\le\lambda_{n}\left(G'\right)\le\lambda_{n-1}\left(G'\right)\le-\left(1-\epsilon\right)\frac{1}{3}d$.
Items~2(a),2(b) and 2(c) hold with respect to the vectors $\vec{\epsilon}_{\bar{1}},\vec{\epsilon}_{\bar{x}},\vec{\epsilon}_{\bar{y}}$
defined above. To establish Item~1.(b) of Theorem~\ref{thm:Main}
note that by Lemma~\ref{lem:Existance of almost eiganvectors}
and by the Courant-Fischer theorem the following holds with high probability
\begin{align*}
\lambda_{n}\left(G'\right)\le\lambda_{n-1}\left(G'\right) & =\min_{\begin{array}{c}
S\in\mathbb{R}^{n}\\
\dim\left(S\right)=2
\end{array}}\max_{v\in S}\frac{v^{t}Av}{v^{t}v}\\
 & \le\max_{v\in\text{span}\left\{ \vec{x},\vec{y}\right\} }\frac{v^{t}Av}{v^{t}v}\\
 & \le-\frac{1}{3}d+\sqrt{d}\,.
\end{align*}

\end{proof}

\subsection{\label{sec:Coloring expander graphs with random color planting}Coloring
expander graphs with a random planted coloring}

Let $G'=\left(V,E\right)$ be a graph as in Theorem~\ref{thm:Main}
($G'$ is a $d$-regular $\lambda$-expander graph after a random
3-color planting). Theorem~\ref{thm:Main} gives some of the mathematical background that is needed in order to analyze our algorithm for 3-coloring $G'$. The following theorem restates in a combined form the parts of Theorem~\ref{thm:partial} that relate to the $H_A/P_R$ model together with Theorem~\ref{thm:A/R}.
\begin{thm}
\label{thm:Main-algorithmicly}Let
$G'$ be as above and assume $\lambda\le\frac{d}{24}$ and $d\ge d_{min}$.
Here $d_{min}$ is a large enough constant. Algorithm~\ref{alg:3-Coloring},
see below, colors $G'$, with high probability over choice of the random planted 3-coloring.
\end{thm}
First we define the following.
\begin{defn}
Given any coloring of $G'$, say $C$, denote by $col_{C}\left(v\right)$
the color of vertex $v$ according to $C$. Denote by $P$ the planted
coloring of $G'$. We say a coloring is partial if some of the vertices
are not colored.
\end{defn}

\begin{defn}
An $f$-approximated coloring is a $3$-coloring of $G'$, possibly
not legal, that agrees with the planted coloring, $col_{P}$, of $G'$
on at least $n-f$ vertices.
\end{defn}

Throughout this section, let $\epsilon=0.01$. Consider Algorithm~\ref{alg:3-Coloring}.

\begin{algorithm}[H]
\begin{enumerate}
\item Apply Algorithm~\ref{alg:Spectral Clustering} ({\em spectral clustering})
on $G'$, obtaining a coloring $C_{1}$.
\item Apply Algorithm~\ref{alg:iterative recoloring} ({\em iterative
recoloring}) on $(G',C_{1})$, obtaining a coloring $C_{2}$.
\item Apply Algorithm~\ref{alg:Cautious Uncoloring} ({\em cautious uncoloring})
on $(G',C_{2})$, obtaining a partial coloring $C_{3}$.
\item Apply Algorithm~\ref{alg:Safe-Recoloring} ({\em safe recoloring})
on $(G',C_{3})$, obtaining a partial coloring $C_{4}$.
\item Apply Algorithm~\ref{alg:Brute-Force} ({\em brute force}) on
$(G',C_{4})$, obtaining a coloring $C_{5}$.
\end{enumerate}
\protect\caption{$3$-Coloring\label{alg:3-Coloring}}
\end{algorithm}

Algorithm 3-Coloring is patterned after the algorithm of Alon and Kahale
for the random planted model. The main differences between the two
{\em algorithms} are as follows. (There are also differences in
the analysis) Similar to \cite{alon1997spectral}, Step~1 (spectral
3-clustering) starts with computing the eigenvectors corresponding
to the two most negative eigenvalues. However, it differs from \cite{alon1997spectral}
in the way $C_{1}$ is derived from these eigenvector. Our approach
is computationally more efficient, an aspect that is only of minor
significance in the context of 3-coloring, but may become significant
for $k$-coloring with $k>3$. See more details in Section~\ref{sub:Proof-Of-Theorem MAin Alg}.
Another difference is the introduction of Step~4 in our algorithm,
that is not present in \cite{alon1997spectral}. The purpose of introducing
this additional step is so as to be able to prove the last statement
of Lemma~\ref{lem:Safe Recoloring}. In \cite{alon1997spectral}
such a statement could be proved already after Step~3.

\subsection{\label{sub:Proof-Of-Theorem MAin Alg}Proof Of Theorem~\ref{thm:Main-algorithmicly}}

The following lemmas outline the desired outcome of the respective
steps of Algorithm~\ref{alg:3-Coloring} and the proof of Theorem~\ref{thm:Main-algorithmicly}
follows by applying them.

At a high level, one may think of most of our proofs as adaptations
(some immediate, some requiring work) of proofs of corresponding statements
made in \cite{alon1997spectral}. The exception to this is the proof
of the last statement in Lemma~4, which involves ideas not present
in \cite{alon1997spectral}.
\begin{defn}
\label{def: Statistically Bad}
The set $SB\subseteq V$, statistically
bad, is the the following vertex set. $v\in SB$ if for some color class other than $col_{P}\left(v\right)$, $v$ has either has more
than $\left(\frac{1}{3}+\epsilon\right)d$ neighbors or less than
$\left(\frac{1}{3}-\epsilon\right)d$ neighbors in that color class. \end{defn}
\begin{lem}
\label{lem:SB bound the random planting Case}It holds that $\left|SB\right|\le n2^{-\Omega\left(d\right)}$
with high probability over the random planting process.\end{lem}
\begin{proof}
(Of Lemma~\ref{lem:SB bound the random planting Case}). In expectation
$v$ has $\frac{1}{3}d$ neighbors in $G'$ with original color $k\in\left\{ 1,2,3\right\} \setminus\left\{ col_{P}\left(v\right)\right\} $.
 By Chernoff bound, Theorem~\ref{thm:=00005BChernoff=00005D}, the
probability that a random vertex $v\in G'$ is in $SB$ is at most
$2^{-\Omega\left(d\right)}$, thus $\mathbb{E}\left[SB\right]=2^{-\Omega\left(d\right)}n$.
If $d\ge\eta_{1}\log n$, for a large enough constant $\eta_{1}$,
then, by using the union bound, the probability that $\left|SB\right|\ge1$
is bounded by $2^{-\Omega\left(d\right)}n=n^{-\Omega\left(1\right)}$.
Assume $d\le\eta_{2}\log n$, where $\eta_{2}$ is a constant. $\left|SB\right|$
is a function of $n$ independent random variables that were used
to create $G'$. By changing one of the above random variables $\left|SB\right|$
changes by at most $d$. By applying McDiarmid's theorem, Theorem~\ref{thm:=00005BMcDiarmid=00005D},
it follows that for a small enough constant $\eta_{2}$ the following
holds,
\begin{align*}
\Pr\left[\left|\left|SB\right|\text{\textminus}\mathbb{E}\left[\left|SB\right|\right]\right|\ge d\mathbb{E}\left[\left|SB\right|\right]\right] & \le2\exp\left(-\frac{2\left(d\mathbb{E}\left[\left|SB\right|\right]\right)^{2}}{2nd^{2}}\right)\\
 & =2^{-\frac{n}{2^{\Omega\left(d\right)}}}=2^{-n^{\Omega\left(1\right)}}\,.
\end{align*}

For (constant) large enough values of $d$ and a suitable choice of
$\epsilon$ we can assure $\eta_{1}\le\eta_{2}$ and the proof follows.
\end{proof}

The following lemma is the first step towards coloring $G'$.
\begin{lem}
\label{lemma:log(n)-Algorithmicly}Let $G'$ as above and assume $\lambda\le\frac{1}{24}d$.
Algorithm~\ref{alg:iterative recoloring}, see below, outputs an
$O\left(\left|SB\right|\right)$-approximated coloring with high probability
over the colors of the vertices in the random planting process.
\end{lem}
Note that if $d=\eta\log n$, (for large enough constant $\eta$),
then Lemma~\ref{lemma:log(n)-Algorithmicly} implies that Algorithm~\ref{alg:iterative recoloring}
recovers the original coloring of $G'$. To prove Lemma~\ref{lemma:log(n)-Algorithmicly}
we first show Lemma~\ref{lem:Eigenvector approximation} and Lemma~\ref{lem:approx from f<1/d to f(lambda/d)^2}
presented below.

The next lemma shows, using Theorem~\ref{thm:Main}, that we can
derive an approximated coloring from the eigenvectors $e_{n-1}\left(G'\right),e_{n}\left(G'\right)$
(see Algorithm~\ref{alg:Spectral Clustering} bellow).

\begin{algorithm}[H]
Input: A graph $G'$ and a positive constant $c\le\frac{1}{2}$.
\begin{enumerate}
\item Calculate $e_{n-1}:=e_{n-1}\left(G'\right)$ and $e_{n}:=e_{n}\left(G'\right)$.
\item For every triplet of vertices $v_{1},v_{2},v_{3}\in V\left(G'\right)$
do

\begin{enumerate}
\item Put a vertex $u\in V\left(G'\right)$ in $S_{i}$ if  it
holds that
\[
\left(\left(e_{n-1}\right)_{u}-\left(e_{n-1}\right)_{v_{i}}\right)^{2}+\left(\left(e_{n}\right)_{u}-\left(e_{n}\right)_{v_{i}}\right)^{2}<\frac{1}{70n}\,.
\]
 In case that two deferent sets, $S_{i}$ and $S_{j}$, satisfy the above condition with respect to    $u$, go to Step~2.
\item If for every $i=1,2,3$ it holds that $\left|S_{i}\right|\ge\left(\frac{1}{3}-\frac{1}{d^{2c}}\right)n$
and that $\sum_{i=1}^{3}|S_i|\ge n- \theta(\frac{n}{d})$
then output any coloring $C$ of $G'$ that satisfies $col_{C}\left(u\right)=i$
if and only if $u\in S_{i}$ and stop.
\end{enumerate}
\end{enumerate}
\protect\caption{\label{alg:Spectral Clustering}Spectral Clustering}
\end{algorithm}

\begin{lem}
\label{lem:Eigenvector approximation}(Restatement of Lemma~\ref{lem:step1}).
Let $G'$ be as above and  $c$ be a positive constant. Algorithm~\ref{alg:Spectral Clustering}
runs in polynomial time and if the following
vectors exist (recall Definition~\ref{Def:x,y}). $\vec{\epsilon}_{\bar{x}}$
such that $\left\Vert \vec{\epsilon}_{\bar{x}}\right\Vert _{2}=O\left(d^{-c}\right)$
and $\bar{x}+\vec{\epsilon}_{\bar{x}}\in\text{span}\left(\left\{ e_{i}\left(G'\right)\right\} _{i\in\left\{ n-1,n\right\} }\right)$.
$\vec{\epsilon}_{\bar{y}}$ such that $\left\Vert \vec{\epsilon}_{\bar{y}}\right\Vert _{2}=O\left(d^{-c}\right)$
and $\bar{y}+\vec{\epsilon}_{\bar{y}}\in\text{span}\left(\left\{ e_{i}\left(G'\right)\right\} _{i\in\left\{ n-1,n\right\} }\right)$. Then Algorithm~\ref{alg:Spectral Clustering} outputs an $O\left(nd^{-2c}\right)$-approximated
coloring.\end{lem}
By Theorem~\ref{thm:Main} we have $\left\Vert \vec{\epsilon}_{\bar{x}}\right\Vert _{2}=O\left(\frac{1}{\sqrt{d}}\right)$
and respectively for $y$ so we can run Algorithm~\ref{alg:Spectral Clustering}
with $c=\frac{1}{2}$.
\begin{proof}
{[}Lemma~\ref{lem:Eigenvector approximation}{]}. We note that \cite{alon1997spectral}
present a different algorithm to get an approximate coloring from
the eigenvectors. Our algorithm generalizes more easily to the case
of planting $k$ colors. The running time of Step~1 is roughly linear (it is enough to get an approximation of $e_{n-1},e_{n}$). The running time of each iteration in Step~2 is $O\left(n\right)$.
The number of iteration in Step~2 is $O\left(n^{3}\right)$ (for
a planted coloring with $k$ colors it is $O\left(n^{k}\right)$ )
but choosing at each iteration a random triplet reduces the (expected)
number of iterations to be constant (for a planted coloring with $k$
colors it is expected to be less than $\frac{k^{k}}{k!}$, see details below). Thus for a planted
coloring with a constant number of vertices the running time of Algorithm~\ref{alg:Spectral Clustering}
is linear.

Assume, for simplicity, that the planted color classes are balanced.
Define the following subset $Good\subseteq V\left(G'\right)$. A vertex
$v\in V\left(G'\right)$ is in $Good$ if both $\left(\left(\vec{\epsilon}_{\bar{x}}\right)_{v}\right)^{2}\le\frac{1}{an}$
and $\left(\left(\vec{\epsilon}_{\bar{y}}\right)_{v}\right)^{2}\le\frac{1}{an}$,
for $a=70$. By averaging argument it holds that $\left|Good\right|\ge\left(1-\frac{2a}{d^{2c}}\right)n$.
Note that as $\bar{x},\bar{y}$ are almost orthogonal then $\bar{x}+\vec{\epsilon}_{\bar{x}},\bar{y}+\vec{\epsilon}_{\bar{y}}$
tend to be almost orthogonal as $d$ gets larger. Formally, by the
triangle's inequality and the Cauchy-Schwartz inequality it follows
that,
\begin{align}
\left|\left\langle \bar{x}+\vec{\epsilon}_{\bar{x}},\bar{y}+\vec{\epsilon}_{\bar{y}}\right\rangle \right| & \le\left|\left\langle \bar{x},\bar{y}\right\rangle \right|+\left|\vec{\epsilon}_{\bar{x}}\right|+\left|\vec{\epsilon}_{\bar{y}}\right|+\left|\vec{\epsilon}_{\bar{x}}\right|\left|\vec{\epsilon}_{\bar{y}}\right|\nonumber \\
 & =O\left(d^{-c}\right)\,.\label{eq:x+eps,y+eps are almost orthogonal.}
\end{align}

Define the vector $\tilde{x}$ to be the second Gram\textendash Schmidt
orthonormalization vector with respect to $\bar{y}+\vec{\epsilon}_{\bar{y}},\bar{x}+\vec{\epsilon}_{\bar{x}}$,
i.e.,
\[
\tilde{x}=\frac{\bar{x}+\vec{\epsilon}_{\bar{x}}-\left\langle \bar{x}+\vec{\epsilon}_{\bar{x}},\bar{y}+\vec{\epsilon}_{\bar{y}}\right\rangle \left(\bar{y}+\vec{\epsilon}_{\bar{y}}\right)}{\left\Vert \bar{x}+\vec{\epsilon}_{\bar{x}}-\left\langle \bar{x}+\vec{\epsilon}_{\bar{x}},\bar{y}+\vec{\epsilon}_{\bar{y}}\right\rangle \left(\bar{y}+\vec{\epsilon}_{\bar{y}}\right)\right\Vert _{2}}\,.
\]
 By Equation~\ref{eq:x+eps,y+eps are almost orthogonal.} it follows
that
\begin{align*}
\left\Vert \bar{x}+\vec{\epsilon}_{\bar{x}}-\tilde{x}\right\Vert _{2} & \le O\left(d^{-c}\right)\,,
\end{align*}
Assume for simplicity that $\left\Vert \bar{x}+\vec{\epsilon}_{\bar{x}}-\tilde{x}\right\Vert _{2}\le d^{-c}$.
Since
\begin{align*}
\left\Vert \tilde{x}\right\Vert _{2} & \ge\left\Vert \bar{x}+\vec{\epsilon}_{\bar{x}}\right\Vert -\left\Vert \bar{x}+\vec{\epsilon}_{\bar{x}}-\tilde{x}\right\Vert _{2}\\
 & \ge1-d^{-c}>0,
\end{align*}
 then $\tilde{x}\neq0$ and by definition $\left\langle \tilde{x},\bar{y}+\vec{\epsilon}_{\bar{y}}\right\rangle =0$.
Hence
\begin{equation}
\text{span}\left(\tilde{x},\bar{y}+\vec{\epsilon}_{\bar{y}}\right)=\text{span}\left(e_{n-1},e_{n}\right)\,.\label{eq:span=00003Dspan}
\end{equation}
Define the following subset $Good_{2}\subseteq V\left(G'\right)$.
A vertex $v\in V\left(G'\right)$ is in $Good_{2}$ if both $v\in Good$
and $\left(\left(\bar{x}+\vec{\epsilon}_{\bar{x}}-\tilde{x}\right)_{v}\right)^{2}\le\frac{1}{an}$.
By averaging argument it holds that $\left|Good_{2}\right|\ge\left(1-\frac{3a}{d^{2c}}\right)n$.
If the condition in Step~2.a of Algorithm~\ref{alg:Spectral Clustering}
holds for $u,w\in V\left(G'\right)$ then for any $x,y\in\mathbb{R}$
that satisfy $x^{2}+y^{2}=1$ the condition
\[
\left|x\left(\left(e_{n-1}\right)_{u}-\left(e_{n-1}\right)_{w}\right)+y\left(\left(e_{n}\right)_{u}-\left(e_{n}\right)_{w}\right)\right|\le\frac{1}{\sqrt{an}}
\]
holds as well, and By Equation~\ref{eq:span=00003Dspan}
\begin{equation}
\left|x\left(\left(\tilde{x}\right)_{u}-\left(\tilde{x}\right)_{w}\right)+y\left(\left(\bar{y}+\vec{\epsilon}_{\bar{y}}\right)_{u}-\left(\bar{y}+\vec{\epsilon}_{\bar{y}}\right)_{w}\right)\right|\le\frac{1}{\sqrt{an}}\,,\label{eq:Rotation under x,y}
\end{equation}
holds for every $x,y$ such that $x^{2}+y^{2}=1$ as well. If $v_{1}$
and $u$ are both in $Good_{2}$ and they both have the same planted
color, by Equation~\ref{eq:Rotation under x,y} it follows that $v_{1},u\in S_{1}$.
This is since for every $x,y$
\begin{align*}
\left|x\left(\left(\tilde{x}\right)_{v_{1}}-\left(\tilde{x}\right)_{u}\right)+y\left(\left(\bar{y}+\vec{\epsilon}_{\bar{y}}\right)_{v_{1}}-\left(\bar{y}+\vec{\epsilon}_{\bar{y}}\right)_{u}\right)\right|\\
\le\left|x\left(\left(\tilde{x}\right)_{v_{1}}-\left(\tilde{x}\right)_{u}\right)\right| & +\left|y\left(\left(\bar{y}+\vec{\epsilon}_{\bar{y}}\right)_{v_{1}}-\left(\bar{y}+\vec{\epsilon}_{\bar{y}}\right)_{u}\right)\right|\\
\le\left|\left(\left(\bar{x}+\vec{\epsilon}_{\bar{x}}\right)_{v_{1}}-\left(\bar{x}+\vec{\epsilon}_{\bar{x}}\right)_{u}\right)\right|+2\sqrt{\frac{1}{an}} & +\left|\left(\left(\bar{y}+\vec{\epsilon}_{\bar{y}}\right)_{v_{1}}-\left(\bar{y}+\vec{\epsilon}_{\bar{y}}\right)_{u}\right)\right|\\
\le\left|\left(\left(\vec{\epsilon}_{\bar{x}}\right)_{v_{1}}-\left(\vec{\epsilon}_{\bar{x}}\right)_{u}\right)\right| & +\left|\left(\left(\vec{\epsilon}_{\bar{y}}\right)_{v_{1}}-\left(\vec{\epsilon}_{\bar{y}}\right)_{u}\right)\right|\le6\sqrt{\frac{1}{an}}\,.
\end{align*}
If $v_{1}$ and $u$ are both in $Good$ and $u$ has a different
planted color than $v$ then by Equation~\ref{eq:Rotation under x,y}
it follows that $u\notin S_{1}$. To see this set $x=0$ and $y=1$,
it follows that
\begin{align*}
\left|\left(\left(\bar{y}+\vec{\epsilon}_{\bar{y}}\right)_{v_{1}}-\left(\bar{y}+\vec{\epsilon}_{\bar{y}}\right)_{u}\right)\right| & \ge\left|\left(\left(\bar{y}\right)_{v_{1}}-\left(\bar{y}\right)_{u}\right)\right|-2\left|\vec{\epsilon}_{\bar{x}}\right|\\
 & \ge\frac{1}{\sqrt{n}}-2\sqrt{\frac{1}{an}}.
\end{align*}
If $v_{1},v_{2}$ and $v_{3}$ are all in $Good$ and have different
planted coloring then by the above Step~2.b of Algorithm~\ref{alg:Spectral Clustering}
will output an $O\left(\frac{n}{d^{2c}}\right)$-approximated coloring.
Note that a random set of three vertices from $V\left(G'\right)$
has a constant probability to satisfy this (specifically by assuming $d$ is much larger than $k$ it is roughly $p=\frac{k!}{k^k}$) so by choosing $v_{1},v_{2}$
and $v_{3}$ at random the expected number of the iterations that
Algorithm~\ref{lem:Eigenvector approximation} performs is constant (more precisely  it is $1/p$).
By similar arguments it is straight forward to prove that any coloring
that Algorithm~\ref{alg:Spectral Clustering} outputs is an $O\left(\frac{n}{d^{2c}}\right)$-approximated
coloring.

\end{proof}
Recall Definition~\ref{def: Statistically Bad}. The following lemma
shows that Algorithm~\ref{alg:One Step Refine} refines any $f=O\left(\frac{n}{24}\right)$-approximated
coloring.

\begin{algorithm}[H]
Input: A graph $G'$ and a coloring $C$.
\begin{enumerate}
\item For each vertex $v$, set $col_{C_{2}}\left(v\right)$ to be the minority
color in the multiset $\left\{ col_{C}\left(u\right)|u\in N\left(v\right)\right\} $
(break ties arbitrarily).
\end{enumerate}
\protect\caption{\label{alg:One Step Refine}One Step Refinement}
\end{algorithm}

\begin{lem}
\label{lem:approx from f<1/d to f(lambda/d)^2} Algorithm~\ref{alg:One Step Refine}
runs in polynomial time ($O\left(nd\right)$). Let $G'$ be as above
and let $f$ be such that $C$ is an $f$-approximated coloring for
$G'$. If $f\le\frac{n}{24}$ then Algorithm~\ref{alg:One Step Refine}
outputs an $\left(\left(\frac{24\lambda}{d}\right)^{2}f+\left|SB\right|\right)$-approximated
coloring for $G'$.\end{lem}
\begin{proof}
{[}Lemma~\ref{lem:approx from f<1/d to f(lambda/d)^2}{]}. Let $W_{C}$
be the set of vertices colored with a different color than the original
coloring in step 1 of the algorithm. Let $W_{C_{2}}$ be the set of
vertices colored with a different color than the original in $C_{2}$.
Let $k=\left|W_{C_{2}}\setminus SB\right|$. The key point, in bounding
$k$, is that any $v\in W_{C_{2}}\setminus SB$ has at least $\left(\frac{1}{6}-\epsilon\right)d$
neighbors in $W_{C}$ as otherwise $v$ was colored correctly.

So on the one hand, by using the Expander Mixing Lemma, (Lemma~\ref{lem:=00005BMixing Lemma=00005D}),
the following holds
\begin{align*}
E_{G'}\left(W_{C},W_{C_{2}}\setminus SB\right) & \le E_{G}\left(W_{C},W_{C_{2}}\setminus SB\right)\\
 & \le\lambda\sqrt{\left|W_{C}\right|\left|W_{C_{2}}\setminus SB\right|}+\frac{d}{n}\left|W_{C}\right|\left|W_{C_{2}}\setminus SB\right|\\
 & \le\lambda\sqrt{fk}+\frac{d}{n}fk\,.
\end{align*}
On the other hand
\begin{align*}
E_{G'}\left(W_{C},W_{C_{2}}\setminus SB\right) & \ge\frac{1}{2}k\left(\frac{1}{6}-\epsilon\right)d\,.
\end{align*}
 It follows that if $k>0$ and $f\le\frac{n}{24}$ then
\[
k\le\left(\frac{24\lambda}{d}\right)^{2}f\,.
\]

Thus if $\lambda<\frac{1}{24}d$ then we can improve the approximation.
\end{proof}
To prove Lemma~\ref{lemma:log(n)-Algorithmicly} consider Algorithm~\ref{alg:iterative recoloring}.

{[}Lemma~\ref{lemma:log(n)-Algorithmicly}{]}.

\begin{algorithm}[H]
Input: A graph $G'$ and a coloring $F_{1}$.
\begin{enumerate}
\item For $i$ taking values from $2$ to $k=\Omega\left(d\right)$ do.

\begin{enumerate}
\item Apply Algorithm~\ref{alg:One Step Refine} where the specified coloring
is $F_{i-1}$ to get $f_{i}$-approximated coloring, namely $F_{i}$.
\end{enumerate}
\end{enumerate}
\protect\caption{\label{alg:iterative recoloring}Iterative Recoloring}
\end{algorithm}

\begin{lem}
\label{lem:Iterative Refinement}(Restatement of Lemma~\ref{lem:step2}).
Algorithm~\ref{alg:iterative recoloring} runs in polynomial time
(the running time is $O\left(d^{2}n\right)$). Let $G'$ be as above.
If $F_{1}$ is a $\left(\frac{n}{24}\right)$-approximated coloring
for $G'$ then Algorithm~\ref{alg:iterative recoloring} outputs
an $O\left(\left|SB\right|\right)$-approximated coloring for $G'$.\end{lem}
\begin{proof}
(Of Lemma~\ref{lem:Iterative Refinement}). By Lemma~\ref{lem:approx from f<1/d to f(lambda/d)^2},
it holds that $f_{i}\le\left(\frac{24\lambda}{d}\right)^{2}f_{i-1}+\left|SB\right|$
for all $2\le i\le k$. It follows that $f_{k}\le\left(\frac{24\lambda}{d}\right)^{2k}\frac{n}{d}+O\left(\left|SB\right|\right)$,
substituting $k=\Omega\left(d\right)$%
\footnote{If $\lambda=o\left(d\right)$ then less iterations are needed.%
} in the last inequality resolves the proof.
\end{proof}
The following lemma shows a structural property of expander graphs
and is used in the proof of Lemma~\ref{lem: Uncoloring leaves us with a large set colored correctly.}.
\begin{lem}
\label{lem:An expander after set being removed has a large core.}Let
$G$ be a $d$-regular $\lambda$-expander graph with $n$ vertices
and consider an arbitrary $S\subset V\left(G\right)$. If $\lambda\le\frac{\epsilon}{16}d$,
$\left|S\right|<\frac{\epsilon}{4}\, n$ and $d$ is large enough
then there exists $CC\subseteq V\left(G\right)\setminus S$ such that
$G_{CC}$ is a graph with a minimum degree of $\left(1-\epsilon\right)d$
and $\left|CC\right|\ge n-2\left|S\right|$.
\end{lem}
Before we give the proof note that if $G$ was the complete graph,
as a toy example, then we can take $CC=V\left(G\right)\setminus S$.
This lemma can be seen as a relaxed version for this property that
holds in expander graphs.
\begin{proof}
(Of Lemma~\ref{lem:An expander after set being removed has a large core.}).
Consider the following iterative procedure. Repeatedly remove a vertex
$v_{t}$ from $V\left(G\right)\setminus\left(S\cup\left(\bigcup_{i=1}^{t-1}v_{i}\right)\right)$
if $v_{t}$ has more than $\epsilon d$ neighbors in $S_{t}=S\cup\left(\bigcup_{i=1}^{t-1}v_{i}\right)$.
Denote by $S'$ the union of $S$ and the removed vertices. Assume
towards a contradiction that this process stops after more than $t=\left|S\right|$
steps. The graph $G_{S_{t}}$ has an average degree of at least $\frac{\epsilon}{2}d$.
By the expander mixing lemma it holds that
\begin{align*}
\frac{\epsilon}{4}d\left|S_{t}\right| & \le E_{G}\left(S_{t},S_{t}\right)\\
 & \le\frac{d}{n}\left|S_{t}\right|^{2}+\lambda\left|S_{t}\right|
\end{align*}
It follows that $\frac{\epsilon}{8}n\le\frac{\left(\frac{\epsilon}{4}d-\lambda\right)}{d}n\le\left|S_{t}\right|$,
but $\left|S_{t}\right|<\frac{\epsilon}{8}n$, which is a contradiction
for large enough $d$ (as $\lambda\le\frac{\epsilon}{8}d$). Clearly
every vertex in $G_{V\left(G\right)\setminus S'}$ has degree at least
$\left(1-\epsilon\right)d$ and the proof follows.
\end{proof}

To prove Theorem~\ref{thm:Main-algorithmicly} consider Algorithm~\ref{alg:Cautious Uncoloring}.

\begin{algorithm}[H]
Input: A graph $G'$ and a coloring $C$.
\begin{enumerate}
\item Repeatedly uncolor the following vertices $v\in G'$ (denote the set
of uncolored vertices by $H_{1}$):

\begin{enumerate}
\item Vertices with less than $\left(\frac{2}{3}-2\epsilon\right)d$ neighbors.
\item Vertices with less than $\left(\frac{1}{6}\right)d$ neighbors of
color $l\in\left\{ 1,2,3\right\} \setminus col_{C}\left(v\right)$.
\end{enumerate}
\end{enumerate}
\protect\caption{\label{alg:Cautious Uncoloring}Cautious Uncoloring}
\end{algorithm}

\begin{lem}
\label{lem: Uncoloring leaves us with a large set colored correctly.}(Restatement
of Lemma~\ref{lem:step3}). Algorithm~\ref{alg:Cautious Uncoloring}
runs in polynomial time. Let $G'$ be as above. If $C$ is an $O\left(\left|SB\right|\right)$-approximated
coloring then the set of the uncolored vertices ($H_{1}$) is of size
$O\left(\left|SB\right|\right)$ and the set of the colored vertices
agree with the planted coloring of $G'$. \end{lem}
\begin{proof}
(Of Lemma~\ref{lem: Uncoloring leaves us with a large set colored correctly.}).
Define
\[
B:=\left\{ v\in V\left(G'\right)\,|\, col_{P}\left(v\right)\neq col_{C}\left(v\right)\right\} \,.
\]

To derive the proof we show $B\subseteq H_{1}$ and that $\left|H_{1}\right|\le O\left(\left|SB\right|\right)$.
By Lemma~\ref{lem:An expander after set being removed has a large core.}
there exist a set $CC\subseteq V\left(G'\right)\setminus\left(SB\cup B\right)$
such that $\left|CC\right|\ge n-2\left|SB\cup B\right|$, and every
$v\in G_{CC}$ has at least $\left(\frac{1}{3}-2\epsilon\right)d$
neighbors in $CC$ with color $i$, for all $i\in\left\{ 1,2,3\right\} \setminus col_{P}\left(v\right)$.
To see this, apply Lemma~\ref{lem:An expander after set being removed has a large core.}
on $G$, the graph before the planting, with $S=SB\cup B$ to get
the set $CC$. Since $v\notin SB$ it holds that $v\in CC$ has at
at least $\left(\frac{1}{3}-2\epsilon\right)d$ neighbors in $CC$
with color $i$, for all $i\in\left\{ 1,2,3\right\} \setminus col_{P}\left(v\right)$.

It holds that no vertex from $CC$ is uncolored, (at the first iteration
of Step~1 this clearly holds. Assume that no vertex from $CC$ is
uncolored in the first $i$ iterations of Step~1, in the $i+1$ iteration
this still holds.).

~

Now we show that all of $V\left(G'\right)\setminus H_{1}$ is colored
correctly. Assume towards a contradiction that there exists a vertex
$v$ in $\left(V\left(G'\right)\setminus H_{1}\right)\setminus CC$
such that $v\in B$. Let $B_{V\left(G'\right)\setminus H_{1}}:=\left\{ v\in V\left(G'\right)\setminus H_{1}\,|\, col_{P}\left(v\right)\neq col_{C}\left(v\right)\right\} $.
Since that any vertex $v\in B_{V\left(G'\right)\setminus H_{1}}$
has degree at least $\left(\frac{2}{3}-2\epsilon\right)d$ neighbors
(as otherwise Step~1.a of Algorithm~\ref{alg:Cautious Uncoloring}
removes $v$) then $v$ has at least $\left(\frac{1}{6}-\epsilon\right)d$
neighbors colored by $col_{p}\left(v\right)$ which means they are
in $B_{V\left(G'\right)\setminus H_{1}}$. It follows that the sub-graph
of $G'$ induced on $B_{V\left(G'\right)\setminus H_{1}}$ has degree
at least $\left(\frac{1}{6}-\epsilon\right)d$. By the expander mixing
lemma, Lemma~\ref{lem:=00005BMixing Lemma=00005D}, it holds that
\begin{align*}
\left(\frac{1}{6}-\epsilon\right)d\left|B_{V\left(G'\right)\setminus H_{1}}\right| & \le E_{G'}\left(B_{V\left(G'\right)\setminus H_{1}},B_{V\left(G'\right)\setminus H_{1}}\right)\\
 & \le E_{G}\left(B_{V\left(G'\right)\setminus H_{1}},B_{V\left(G'\right)\setminus H_{1}}\right)\\
 & \le\frac{d}{n}\left|B_{V\left(G'\right)\setminus H_{1}}\right|^{2}+\lambda\left|B_{V\left(G'\right)\setminus H_{1}}\right|\,.
\end{align*}
It follows that if $\left|B_{V\left(G'\right)\setminus H}\right|>0$,
as we assumed, then $\left|B_{V\left(G'\right)\setminus H}\right|\ge\Omega\left(\frac{d-\lambda}{d}\right)n$
but as it holds that
\[
\left|B_{V\left(G'\right)\setminus H}\right|\le n-\left|CC\right|=O\left(\left|SB\right|\right)\,,
\]
 we derive a contradiction.
\end{proof}

\begin{algorithm}[H]
Input: A graph $G'$ and a partial coloring $C$. Denote the set of
uncolored vertices by $H_{1}$.
\begin{enumerate}
\item Repeatedly color any vertex that has $2$ neighbors colored with $i,j$,
(for $i\neq j$), by $\left\{ 1,2,3\right\} \setminus\left\{ i,j\right\} $.
\end{enumerate}
\protect\caption{\label{alg:Safe-Recoloring}Safe Recoloring}
\end{algorithm}

\begin{lem}
\label{lem:Safe Recoloring}Algorithm~\ref{alg:Safe-Recoloring}
runs in polynomial time. Let $G'$ be as above and $H$ be the set
of vertices that were left uncolored by Algorithm~\ref{alg:Safe-Recoloring}.
Let $G'_{H}$ be the induced sub-graph of $G'$ on the vertices of
$H$. If $H_{1}\le\frac{1}{11}n$ then with high probability (over
the distribution of $G'$) each connected component in $G'_{H}$ is
of size $O\left(\frac{\lambda^{2}}{d^{2}}\log n\right)$.\end{lem}
\begin{proof}
(Of Lemma~\ref{lem:Safe Recoloring}). Recall that $G$ is the graph
before the planted coloring process occurred. We use Claim~\ref{claim:sets with large neighborhod are not in H}
extensively.
\begin{claim}
\label{claim:sets with large neighborhod are not in H} Let $C\subset V\left(G\right)$
be any set of vertices of size $k$, let $D$ be $N_{G}\left(C\right)\setminus C$
and let $t=\left|D\right|$. If $t\ge2k$ then $C$ is in $H$ and
$D\subseteq V\left(G\right)\setminus H$ with probability at most
$2^{-\eta t}$, for some constant $\eta$.
\end{claim}
\begin{proof}
Partition the vertices of $D$ to disjoint subsets $\left\{ D_{v}\right\} $,
 where $v\in C$ and $D_{v}\subseteq N_{G}\left(v\right)$ (this can be done by choosing for each vertex in $D$ an arbitrary neighbor in $S$). Consider
some $D_{v}$ and assume that $v$ is colored by $1$. If both $v\in H$
and $D_{v}\subseteq V\left(G\right)\setminus H$ then the set of colors
that vertices in $D_{v}$ are using can be contained in exactly one
of the sets $\left\{ 1,2\right\} $ or $\left\{ 1,3\right\} $ (otherwise
Algorithm~\ref{alg:Safe-Recoloring} will color $v$ at Step~1).
Therefore the probability that both $v\in H$ and $D_{v}\subseteq V\left(G\right)\setminus H$
is at most $2\left(\frac{2}{3}\right)^{\left|D_{v}\right|}$. Hence,
the probability that both $C$ is in $H$ and $D\subseteq V\left(G\right)\setminus H$
is at most
\begin{align*}
	\prod_{D_{v}\in\left\{ D_{v}\right\} }2\left(\frac{2}{3}\right)^{\left|D_{v}\right|} & =2^{\left|\left\{ D_{v}\right\} \right|}\left(\frac{2}{3}\right)^{t}\\
	& \le2^{k}\left(\frac{2}{3}\right)^{t}\\
	& =\left(\frac{2}{3}\right)^{t+k\log_{\frac{2}{3}}2}\le\left(\frac{2}{3}\right)^{0.14t}\,.
\end{align*}

The last inequality follows by the assumption that $t\ge2k$.

\end{proof}
~

Let $C\subset H$ be any set of vertices of size $k$ such that $G'_{C}$
is connected. Note that if $G'_{C}$ is connected then $G{}_{C}$
is connected. Hence we can consider the set $\bar{C}$ to be a connected
set in $G$ such that $C\subseteq\bar{C}\subset H$ and $N\left(\bar{C}\right)\setminus\bar{C}\subseteq V\left(G\right)\setminus H$
(as long as there exist vertices in $\left(N\left(\bar{C}\right)\setminus\bar{C}\right)\cap H$
we repeatedly add them to $\bar{C}$, and it is straight forward to
show that $\bar{C}$ is connected in $G$). Thus if we show that with
high probability no connected subsets $\bar{C}$ of $G$ with cardinality
$\Omega\left(\frac{\lambda^{2}}{d^{2}}\log n\right)$ are in $H$
if $N\left(\bar{C}\right)\subseteq V\left(G\right)\setminus H$ then
the proof follows. The key point of the proof is that (since $H$
is small) no (small connected) subsets of $H$ with with a small neighborhood
are in $V\left(G\right)\setminus H$ (due to $G$'s expansion) and
when a connected set have a large enough neighborhood we can use Claim~\ref{claim:sets with large neighborhod are not in H}
to show they are likely to not be contained in $H$. To fully exploit
Claim~~\ref{claim:sets with large neighborhod are not in H} we
use the union bound in a delicate manner that depends on the neighborhood
size of the (connected) sets in order to sum sets with large neighborhood
size with small probabilities. In order to do so we use Lemma~\ref{lem:=000023connected com' with neighborhood t}.

~

~

Denote by $N_{G}\left(n,k,d\right)$ the number of connected components
of size $k$ with a neighborhood of size exactly $t$ in a graph $G$
with $n$ vertices.
\begin{lem}
\label{lem:=000023connected com' with neighborhood t}For any graph
$G$ the following holds,
\[
N_{G}\left(n,k,t\right)\le n{k+t \choose k}\,.
\]

\end{lem}
Several proofs of Lemma~\ref{lem:=000023connected com' with neighborhood t} are known, and
 one such a proof can be found in~\cite{krivelevich2015smoothed}. Lemma~\ref{lem:=000023connected com' with neighborhood t} holds
for any graph, regardless of its degree or expansion. Let $\alpha$
be such that $\lambda=\frac{d}{\alpha}$. Let
\[
p_{s,t}:=\max_{C\subset G,\left|C\right|=s,\left|N\left(C\right)\right|=t}\Pr\left[\text{\ensuremath{C}\,\ is in\,\ensuremath{H}}\right]\,.
\]
By the union bound the probability that there exists a connected component
in $G'_{H}$ of size $\Omega\left(\frac{\lambda^{2}}{d^{2}}\log n\right)$
is at most
\begin{align}
\sum_{s=\frac{\lambda^{2}}{d^{2}}\log n\,}^{\left|H\right|}\sum_{t=1}^{ds}N_{G}\left(n,s,t\right)P_{s,t} & =\sum_{s=\frac{\lambda^{2}}{d^{2}}\log n\,}^{\left|H\right|}\sum_{t=\left\lceil \frac{1}{2}\left(\frac{d}{\lambda}\right)^{2}s\right\rceil }^{ds}N_{G}\left(n,s,t\right)2^{-\eta t}\label{eq:alpha1}\\
 & \le\sum_{s=\frac{\lambda^{2}}{d^{2}}\log n\,}^{\left|H\right|}\sum_{t=\left\lceil \frac{\alpha^{2}}{2}s\right\rceil }^{ds}n{s+t \choose s}2^{-\eta t}\label{eq:alpha2}\\
 & \le\sum_{s=\frac{\lambda^{2}}{d^{2}}\log n\,}^{\left|H\right|}\sum_{t=\left\lceil \frac{\alpha^{2}}{2}s\right\rceil }^{ds}n\left(e\frac{s+t}{s}\right)^{s}2^{-\eta t}\label{eq:alpha3}\\
 & \le\sum_{s=\frac{\lambda^{2}}{d^{2}}\log n\,}^{\left|H\right|}nds\max_{\frac{\alpha^{2}}{2}\le c\le d}2^{\left(\log_{2}\left(e\left(c+1\right)\right)-\eta c\right)s}\nonumber \\
 & \le n^{4}\max_{\frac{\alpha^{2}}{2}\le c\le d}2^{\left(\log_{2}\left(e\left(c+1\right)\right)-\eta c\right)s}\\
 & \le n^{-\Omega\left(1\right)}\,.\label{eq:alpha4}
\end{align}
Equality~\ref{eq:alpha1} follows since by Lemma~\ref{lem:=00005BSpectral to vertex expansion.=00005D}
if $\lambda\le d/\sqrt{12}$ then for a set of size $s\le\left|H\right|\le\frac{1}{11}n$
it holds that its neighborhood is of size $t\ge\frac{1}{2}\left(\frac{d}{\lambda}\right)^{2}s$
(thus $N_{G}\left(n,s,t\right)=0$) and by Claim~\ref{claim:sets with large neighborhod are not in H}
(for some constant $\eta$). Inequality~\ref{eq:alpha2} follows
by Lemma~\ref{lem:=000023connected com' with neighborhood t}. Inequality~\ref{eq:alpha3}
follows by the binomial identity ${n \choose k}\le\left(e\frac{n}{k}\right)^{k}$.
Inequality~\ref{eq:alpha4} follows if $\alpha\ge\alpha'$ , where
$\alpha'$ is a large enough constant.
\end{proof}
\begin{algorithm}[H]
Input: A graph $G'$ and a partial coloring $C$.
\begin{enumerate}
\item If the induced graph of the uncolored vertices, $H$, contains a connected
component of size $\log n$ abort.
\item Enumerate over all possible coloring of each connected component of
$H$ and return a legal $3$-coloring of $G'$.\protect\caption{\label{alg:Brute-Force}Brute Force}

\end{enumerate}
\end{algorithm}

\section{Adversarial host and adversarial planting, $H_A/P_A$}
\label{app:AA}

Section~\ref{app:AA} contains the proofs for the model with an adversarial expander host graph and an adversarial planted balanced coloring, namely, $H_A/P_A$.
Section~\ref{sec:Coloring Expander Graphs With Adversarial Planting} shows how to obtain a $b$-partial coloring for graphs in the $H_{A}/P_{A}$ model (Theorem~\ref{thm:partial}).
Section~\ref{sec:On the hardness of coloring expander graphs with adversarial planting} proves Theorem~\ref{thm:A/A}
(hardness result for the $H_{A}/P_{A}$ model).

\subsection{\label{sec:Coloring Expander Graphs With Adversarial Planting}
Partial coloring of expander graphs with adversarial planting}

Given a $d$-regular $\lambda$-expander graph $G$ and a partition,
$C$, of $V(G)$ to $3$ sets, $c_{1},c_{2},c_{3}$, we denote by
$P_{C}\left(G\right)$ the graph obtained after removing all edges
from $G$ with both endpoints in the same set $c_{i}$. If all the
sets $c_{1},c_{2},c_{3}$ have the same cardinality we say that $C$
is balanced. In this section we consider the computational problem
of coloring $P_{C}\left(G\right)$ when $C$ is given by an adversary.

In this section we show the following theorem.
\begin{thm}
\label{thm:Coloring with balanced adversarial planting}(Restatement
of Theorem~\ref{thm:partial}). Let $G$ be a $d$-regular $\lambda$-expander
graph. If $\lambda\le c d$ (for some constant $c$)
then for every balanced 3-color planting $C$ Algorithm~\ref{alg:Iterative Coloring and Uncoloring}
outputs an $O\left(n\left(\frac{\lambda}{d}\right)^2\right)$-approximated coloring for $P_{C}\left(G\right)$.
\end{thm}
Throughout this section, let $\epsilon=0.01$.
\begin{defn}
The set $SB\subseteq V$, statistically bad, is the the following
vertex set. $v\in SB$ if $v$ has more than $\left(\frac{1}{3}+\epsilon\right)d$
or neighbors, or less than $\left(\frac{1}{3}-\epsilon\right)d$ or
neighbors, with original color $k\in\left\{ 1,2,3\right\} \setminus\left\{ col_{C}\left(v\right)\right\} $.
\end{defn}
It turns out that by applying Algorithm~\ref{alg:iterative recoloring}
and the uncoloring procedure of Algorithm~\ref{alg:Cautious Uncoloring}
(see Algorithm~\ref{alg:Iterative Coloring and Uncoloring} below)
we can find the planted coloring of a set of $\left(1-O\left(\frac{1}{d}\right)\right)n$
vertices.

\begin{algorithm}[H]
\begin{enumerate}
\item Color $P_{C}\left(G\right)$ with an $O\left(\left|SB\right|\right)$-approximated
coloring ,$C_{Alg}$ (use Algorithm~\ref{alg:iterative recoloring}).
\item Repeatedly uncolor vertices $v\in P_{C}\left(G\right)$ with less
than $\left(\frac{2}{3}-2\epsilon\right)d$ neighbors or with less
than $\left(\frac{1}{6}\right)d$ neighbors of color $l\in\left\{ 1,2,3\right\} \setminus col_{C_{Alg}}\left(v\right)$,
denote the set of uncolored vertices by $H_{1}$.
\item If $\left|H_{1}\right|=O\left(\log n\right)$ then enumerate over
all the possible colorings ($3^{\left|H_{1}\right|}$) of the vertices
in $H_{1}$.
\end{enumerate}
\protect\caption{\label{alg:Iterative Coloring and Uncoloring}Iterative Coloring and
Uncoloring}
\end{algorithm}

Throughout the rest of this section we denote $G'=P_{C}\left(G\right)$.
The main point of the proof of Theorem~\ref{thm:Coloring with balanced adversarial planting}
is that $e_{n-1}\left(P_{C}\left(G\right)\right),e_{n}\left(P_{C}\left(G\right)\right)$
are related to $C$ even when $C$ is arbitrary balanced (rather than
random as in the previous sections). This is stated in Lemma~\ref{lem: eigen values after arbitrary planting}
which is similar to Theorem~\ref{thm:Main}. To state Lemma~\ref{lem: eigen values after arbitrary planting}
let us define the following.
\begin{defn}
\label{Def:x,y-1}Given $C$ define the following vectors in $\mathbb{R}^{n}$.
$\left(\vec{1}_{n}\right)_{i}=1$, $\left(\vec{p_{i}}\right)_{j}:=\begin{cases}
1 & v_{j}\in c_{i}\\
0 & v_{j}\notin c_{i}
\end{cases}$, $\left(\vec{x}\right)_{i}:=\begin{cases}
2 & v_{i}\in c_{1}\\
-1 & v_{i}\in c_{2}\\
-1 & v_{i}\in c_{3}
\end{cases}$ and $\left(\vec{y}\right)_{i}:=\begin{cases}
0 & v_{i}\in c_{1}\\
1 & v_{i}\in c_{2}\\
-1 & v_{i}\in c_{3}
\end{cases}$.\end{defn}
\begin{lem}
\label{lem: eigen values after arbitrary planting}Let $G$ be a $d$
regular $\lambda$-expander, where $\lambda \le c d$
(for some constant $c$). If \textup{$G'=P_{C}\left(G\right)$ for a balanced
$C$  then}  with high probability the following holds.
\begin{enumerate}
\item The eigenvalues of $G'$ have the following spectrum.

\begin{enumerate}
\item $\lambda_{1}\left(G'\right)\ge\frac{2}{3}d-\sqrt{d\lambda}$.
\item $\lambda_{n}\left(G'\right)\le\lambda_{n-1}\left(G'\right)\le-\frac{1}{3}d-\sqrt{d\lambda}$.
\item $|\lambda_{i}\left(G'\right)|\le2\lambda+O\left(\sqrt{d\lambda}\right)$
for all $2\le i\le n-2$.
\end{enumerate}
\item The following vectors exist.

\begin{enumerate}
\item $\vec{\epsilon}_{\bar{x}}$ such that $\left\Vert \vec{\epsilon}_{\bar{x}}\right\Vert _{2}=O\left(\sqrt{\frac{\lambda}{d}}\right)$
and $\bar{x}+\vec{\epsilon}_{\bar{x}}\in\text{span}\left(\left\{ e_{i}\left(G'\right)\right\} _{i\in\left\{ n-1,n\right\} }\right)$.
\item $\vec{\epsilon}_{\bar{y}}$ such that $\left\Vert \vec{\epsilon}_{\bar{y}}\right\Vert _{2}=O\left(\sqrt{\frac{\lambda}{d}}\right)$
and $\bar{y}+\vec{\epsilon}_{\bar{y}}\in\text{span}\left(\left\{ e_{i}\left(G'\right)\right\} _{i\in\left\{ n-1,n\right\} }\right)$.
\end{enumerate}
\end{enumerate}
\end{lem}
\begin{proof}
The proof goes by showing an equivalent statement as in Lemma~\ref{lem:Existance of almost eiganvectors}.
In this proof we set $\epsilon$ to be $\sqrt{\frac{\lambda}{d}}$ (rather
than $0.01$). Denote $c_{G',i,j}^{+}=\left\{ v\in c_{i}|\left|N_{G'}\left(v\right)\cap c_{j}\right|\ge\left(\frac{1}{3}+\epsilon\right)dn\right\} $
and $c_{G',i,j}^{-}=\left\{ v\in c_{i}|\left|N_{G'}\left(v\right)\cap c_{j}\right|\le\left(\frac{1}{3}-\epsilon\right)dn\right\} $.
By the expander mixing lemma, Lemma~\ref{lem:=00005BMixing Lemma=00005D},
it follow that for all $i\neq j$
\begin{align*}
\left(\frac{1}{3}+\epsilon\right)d\left|c_{G',i,j}^{+}\right|\le\left|E_{G'}\left(c_{G',i,j}^{+},G'_{c_{j}}\right)\right| & =\left|E_{G}\left(c_{G,i,j}^{+},G{}_{c_{j}}\right)\right|\\
 & \le\frac{d}{3}\left|c_{G',i,j}^{+}\right|+\lambda\sqrt{\frac{n}{3}\left|c_{G',i,j}^{+}\right|}\,.
\end{align*}

The above simplifies to
\[
\left|c_{G',i,j}^{+}\right|\le\left(\frac{\lambda}{\epsilon d}\right)^{2}\frac{n}{3}\,,
\]
and similarly
\[
\left|c_{G',i,j}^{-}\right|\le\left(\frac{\lambda}{\epsilon d}\right)^{2}\frac{n}{3}\,.
\]
Recall the definition of $\vec{x}$, it follows that
\begin{align*}
\sum_{i=1}^{n/3}\left(A_{G'}\vec{x}-\left(-\frac{d}{3}\right)\vec{x}\right)_{i}^{2} & \le d^{2}\left(\left|c_{G',1,2}^{+}\right|+\left|c_{G',1,2}^{-}\right|+\left|c_{G',1,3}^{+}\right|+\left|c_{G',1,3}^{-}\right|\right)+\left(2\epsilon d\right)^{2}\frac{n}{3}\\
 & \le d^{2}\frac{n}{3}\left(4\left(\frac{\lambda}{\epsilon d}\right)^{2}+\left(2\epsilon\right)^{2}\right)\\
 & \le\frac{8}{3}nd\lambda\,.
\end{align*}
By similar calculations we can conclude that
\[
\sum_{i=1}^{n}\left(A_{G'}\vec{x}-\left(-\frac{d}{3}\right)\vec{x}\right)_{i}^{2}\le8nd\lambda \,,
\]
and if we normalize then
\[
\left\Vert A_{G'}\bar{x}-\left(-\frac{d}{3}\right)\bar{x}\right\Vert _{2}\le O\left(\sqrt{d\lambda }\right)\,.
\]
Similarly we can show equivalent statements with respect to $\bar{y},\bar{1}$
and $\left\{ \bar{p}_{i}\,|\, i\in\left\{ 1,2,3\right\} \right\} $.
By applying Lemma~\ref{lem:Almost Eigan Vectors implies more} we
can conclude an equivalent statement as in Lemma~\ref{lem:distance to S_delta},
namely the following claim.
\begin{claim}
\label{claim:AdvPlanting1}Define $S_{\delta}\left(x\right)\subset\{1,\ldots,n\}$
to be the set of those indices $i$ for which $|\lambda\left(G'\right)_{i}-\left(x\right)|\le\delta$.
The following vectors exist.
\begin{itemize}
\item $\vec{\epsilon}_{\bar{1}}$ such that $\left\Vert \vec{\epsilon}_{\bar{1}}\right\Vert _{2}\le\frac{O\left(\sqrt{d\lambda }\right)}{\delta}$
and $\bar{1}+\vec{\epsilon}_{\bar{1}}\in\text{span}\left(\left\{ e_{i}\left(G'\right)\right\} _{i\in S_{\delta}\left(\frac{2}{3}d\right)}\right)$.
\item $\vec{\epsilon}_{\bar{x}}$ such that $\left\Vert \vec{\epsilon}_{\bar{x}}\right\Vert _{2}\le\frac{O\left(\sqrt{d\lambda }\right)}{\delta}$
and $\bar{x}+\vec{\epsilon}_{\bar{x}}\in\text{span}\left(\left\{ e_{i}\left(G'\right)\right\} _{i\in S_{\delta}\left(-\frac{1}{3}d\right)}\right)$.
\item $\vec{\epsilon}_{\bar{y}}$ such that $\left\Vert \vec{\epsilon}_{\bar{y}}\right\Vert _{2}\le\frac{O\left(\sqrt{d\lambda }\right)}{\delta}$
and $\bar{y}+\vec{\epsilon}_{\bar{y}}\in\text{span}\left(\left\{ e_{i}\left(G'\right)\right\} _{i\in S_{\delta}\left(-\frac{1}{3}d\right)}\right)$.
\end{itemize}
\end{claim}
Using Claim~\ref{claim:AdvPlanting1} and applying the same arguments
as in the proof of Theorem~\ref{thm:Main} ends the proof.
\end{proof}

\begin{lem}
\label{lem: sb bound adversary}It holds that $\left|SB\right|\le O\left(n\left(\frac{\lambda}{\epsilon d}\right)^{2}\right)$.\end{lem}
\begin{proof}
Recall the definition of $c_{G',i,j}^{+}$ and $c_{G',i,j}^{-}$ from
the proof of Lemma~\ref{lem: eigen values after arbitrary planting}.
It holds that
\begin{align*}
SB & \le\sum_{i\neq j}\left(\left|c_{G',i,j}^{+}\right|+\left|c_{G',i,j}^{-}\right|\right)\\
 & \le4\left(\frac{\lambda}{\epsilon d}\right)^{2}n\,.
\end{align*}

\end{proof}

Now we continue with the proof of Theorem~\ref{thm:Coloring with balanced adversarial planting}.
\begin{proof}
{[}Theorem~\ref{thm:Coloring with balanced adversarial planting}{]}.
Using Lemma~\ref{lem: eigen values after arbitrary planting} and
the proof of Lemma~\ref{lem:Eigenvector approximation} we conclude
that Algorithm~\ref{alg:Spectral Clustering} outputs an $O\left(n\frac{\lambda}{d}\right)$-approximate
coloring of $G'$. The proof of Lemma~\ref{lemma:log(n)-Algorithmicly}
shows that Algorithm~\ref{alg:iterative recoloring} gets $O\left(\left|SB\right|\right)$-approximate
coloring of $G'$. Define
\[
B:=\left\{ v\in V\left(G'\right)\,|\, col_{C}\left(v\right)\neq col_{C_{Alg}}\left(v\right)\right\} \,.
\]
The proof of Lemma~\ref{lem: Uncoloring leaves us with a large set colored correctly.}
shows that $\left|H_{1}\right|\le O\left(\left|SB\right|\right)$
and $B\subseteq H_{1}$. By Lemma~\ref{lem: sb bound adversary}
it follows that $SB\le O\left(n\left(\frac{\lambda}{\epsilon d}\right)^{2}\right)$
and proof follows.
\end{proof}

When $G$ has an high degree (with respect to its expansion) we obtain
the following corollary.
\begin{cor}
Let $G$ be an $\Omega\left(\frac{n}{\log n}\right)$-regular
 $\lambda$-expander graph. If $\lambda=O\left(\sqrt{d}\right)$
then for every balanced $C$ Algorithm~\ref{alg:Iterative Coloring and Uncoloring}
colors $P_{C}\left(G\right)$.
\end{cor}

\subsection{\label{sec:On the hardness of coloring expander graphs with adversarial planting}Hardness of coloring expander graphs with adversarial planting}

It is well known~\cite{garey1976some} that $3$-coloring $4$-regular
graphs is NP-hard. In this section we show the following theorem.
\begin{thm}
(Restatement of Theorem~\ref{thm:A/A}). Let $c$ be a constant equals
$4$ and $d=\Omega\left(1\right)$ be large enough. Let $H$ be an
arbitrary $c$-regular graph with $\frac{n}{cd}$ vertices that is
$3$-colorable. Any (polynomial time) algorithm that (fully) colors
$d$-regular $O\left(\sqrt{d}\right)$-expander graphs with an adversarial
planted 3-coloring that is balanced can be used to color $H$. \end{thm}
\begin{proof}
Assume the graph $H$ is $3$-colorable in such a way that all the
color classes are of the same size (this can be obtained by taking
a disjoint union of $3$ copies of $H$), denote this coloring by
$P_{H}=P_{1}\cup P_{2}\cup P_{3}$. Let $G$ be any $d$-regular $O\left(\sqrt{d}\right)$-expander
graph with $\left(1-\frac{1}{cd}\right)n$ vertices. Let $A$ be an
algorithm that colors $d$-regular $O\left(\sqrt{d}\right)$-expander
graphs with an adversarial planted coloring. We show that this implies
that $A$ colors the disjoint union $H$ and $P\left(G\right)$ ($P\left(G\right)$
is an arbitrary balanced planted coloring of $G$), denote this union
by $\left(P\left(G\right)+H\right)$. Consider the graph $G_{H}$
obtained as follows. Add $3$ sets, $S_{1},S_{2},S_{3}$ of $\frac{n}{3cd}$
vertices to $G$ such that each of the added vertices has $d-c$ unique
neighbors in $G$ and all the vertices from the $i$-th set have their
neighborhood contained in $P_{i}$. Replace the induced sub-graph
on the vertices of $\bigcup S_{i}$ (it is an independent set) with
$H$ (permute the vertices of $H$ such that the sets $\left\{ S_{i}\right\} $
agrees with $P$ ). It is clear from this construction that there
exists an (adversarial) planting such that after it has been applied
on $G_{H}$ we obtain $\left(P\left(G\right)+H\right)$. The following
claim ends the proof.

~
\begin{claim}
$G_{H}$ is a $O\left(\sqrt{d}\right)$ expander.\end{claim}
\begin{proof}
Consider the following inequality from perturbation theory for matrices
that holds for any two symmetric matrices $A,N\in\mathbb{R}^{n,n}$
(see, for example, \cite{bhatia2013matrix})
\begin{equation}
\max_{i:1\le i\le n}\left|\lambda_{i}\left(A+N\right)-\lambda_{i}\left(A\right)\right|\le\max_{i:1\le i\le n}\left|\lambda_{i}\left(N\right)\right|\,.\label{eq:Sums of Eigenvalues}
\end{equation}

Namely, the inequality shows that by adding a matrix $N$ to a matrix
$A$, the eigenvalues of $A+N$ change by at most $\max_{i:1\le i\le n}\left|\lambda_{i}\left(N\right)\right|$.

We can write the adjacency matrix of $G{}_{H}$ as
\[
A_{G{}_{H}}=A'_{G}+A'_{H}+A'_{S}\,.
\]

Here $S$ is the disjoint union of $\frac{n}{cd}$ star graphs $S_{d-c}$
(the graph $S_{k}$ is a bipartite graph of $(k+1)$ vertices with
one vertex connected to all the other vertices) and $A'_{G}$ is the
adjacency matrix of the graph obtained from $G$ an independent set
of vertices were added to it (similarly for $H$ and $S$). Note that
adding an independent set to a graph only adds zero entries to its
spectrum.

It is a known fact that for every $i$ it holds that $\left|\lambda_{i}\left(S_{k}\right)\right|\le\sqrt{k}$.
Hence, since $S$ is the disjoint union of star graphs then every
$i$ it holds that $\left|\lambda_{i}\left(A_{S}\right)\right|\le\sqrt{d-c}$.
Since $H$ is $c$ regular, for every $i$ it holds that $\left|\lambda_{i}\left(H\right)\right|\le c$.
By Inequality~\ref{eq:Sums of Eigenvalues} it follows that for every
$i\in\left[n\right]$ it holds that
\begin{align}
\left|\lambda_{i}\left(A_{G}\right)\right|-\left(c+\sqrt{d-c}\right) & \le\left|\lambda_{i}\left(A_{G_{H}}\right)\right|\label{eq:G_H is expander}\\
 & \le\left|\lambda_{i}\left(A_{G}\right)\right|+\left(c+\sqrt{d-c}\right)\,.\nonumber
\end{align}
Recall that $c$ is a constant and $G$ is a $d$-regular $O\left(\sqrt{d}\right)$-expander
graph and note that (by construction) vertices of $G_{H}$ have degrees
between $d$ and $d+1$ (actually, this implies that $d\le\lambda_{1}\left(G_{H}\right)\le d+1$).
Hence by Inequality~\ref{eq:G_H is expander} $G{}_{H}$ is (roughly
$d$-regular) $O\left(\sqrt{d}\right)$-expander.
\end{proof}
~
\end{proof}

\section{Random host and adversarial planting, $H_R/P_A$}
\label{app:RA}

Section~\ref{app:RA} contains the proofs for the model with a random host graph and an adversarial planted balanced coloring, namely, $H_R/P_A$.
Section~\ref{sec:F 3-coloring random graphs with adversarial color planting}
proves Theorem~\ref{thm:R/A}a (a 3-coloring algorithm for $H_R/P_A$). Section~\ref{sec:G Hardness of 3-coloring random graphs with adversarial planted 3-coloring}
proves Theorem~\ref{thm:R/A}b (hardness for $H_R/P_A$).

\subsection{\label{sec:F 3-coloring random graphs with adversarial color planting}
3-coloring random graphs with adversarial color planting}

We start with a definition for a distribution for random graphs, following
\cite{erdds1959random},
\begin{defn}
A graph with $n$ vertices $G$ is distributed by $G_{n,d}$ if each
edge is included in the graph with probability $p=\frac{d}{n-1}$,
independently from every other edge.
\end{defn}
In this section we consider the following planting model. The host
graph $G\sim G_{n,d}$ is a random graph on $n$ vertices with average
degree $d$ and the planting is adversarial. As opposed to the previous
sections, here our techniques can be applied only to $3$ colors (or
less) planting (see Section~\ref{sub:G.2 Hardness result for random graphs with adversarial 4-color planting}
for details). Consider the following algorithm.

\begin{algorithm}[H]
\begin{enumerate}
\item Use Algorithm~\ref{alg:iterative recoloring} to color $P_{C}\left(G\right)$
with a partial coloring ,$C_{Alg}$.
\item Repeatedly uncolor vertices $v\in P_{C}\left(G\right)$ with less
than $\left(\frac{2}{3}-2\epsilon\right)d$ neighbors or with less
than $\left(\frac{1}{6}\right)d$ neighbors of color $l\in\left\{ 1,2,3\right\} \setminus col_{C_{Alg}}\left(v\right)$.
\item Repeatedly color every vertex that has $2$ neighbors colored with
$i,j$, (for $i\neq j$), by $\left\{ 1,2,3\right\} \setminus\left\{ i,j\right\} $.
\item For every vertex that has a colored neighbor create a variable that
takes values corresponding to the two remaining colors.
\item For each $i\in\left\{ 1,2,3\right\} $ do as follows.

\begin{enumerate}
\item Any uncolored vertex that has no colored neighbor is colored by $i$.
\item If there exists a legal coloring for $P_{C}\left(G\right)$ that agrees
with the above (use any algorithm that solves $2SAT$) return it.
\end{enumerate}
\item Return the partial coloring of $P_{C}\left(G\right)$.
\end{enumerate}
\protect\caption{\label{alg:3-coloring-for-random graphs with adv planting}3-coloring
for random graphs with adversarial planting }
\end{algorithm}

We prove the following theorem.
\begin{thm}
\label{thm:random graphs with Adv Planting}(Restatement of Theorem~\ref{thm:R/A}a).
Let $d=\omega\left(n^{\frac{2}{3}}\right)$, $G\sim G_{n,d}$ and
$P_{C}\left(G\right)$ be the graph $G$ with an adversarial 3-coloring.
With high probability, Algorithm~\ref{alg:3-coloring-for-random graphs with adv planting}
outputs a legal coloring for $P_{C}\left(G\right)$.
\end{thm}
The starting point in the proof of Theorem~\ref{thm:random graphs with Adv Planting}
is that random graphs are very good expanders. Specifically, it is
a well known fact that with high probability $\lambda_{2}\left(G\right)=\Theta\left(\sqrt{d}\right)$
(for example, see \cite{feige2005spectral,furedi1981eigenvalues,friedman1989second}).
This enables us to employ our techniques from Section~\ref{sec:Coloring Expander Graphs With Adversarial Planting}
(mainly to apply Theorem~\ref{thm:Coloring with balanced adversarial planting})
to get an $O\left(\frac{n}{d}\right)$-approximated coloring for $P_{C}\left(G\right)$.
To get a better coloring of $P_{C}\left(G\right)$ than suggested
by Theorem~\ref{thm:Coloring with balanced adversarial planting},
we use specific properties of random graphs. Mainly that every pair
of vertices has roughly $\frac{d^{2}}{n}$ common neighbors.
\begin{proof}
(Of Theorem~\ref{thm:random graphs with Adv Planting}). Recall that
with high probability $\lambda_{2}\left(G\right)=\Theta\left(\sqrt{d}\right)$.
Hence, by Theorem~\ref{thm:Coloring with balanced adversarial planting},
it follows that after Step~2 we have a partial coloring of $P_{C}\left(G\right)$
that colors all but $O\left(\frac{n}{d}\right)$ vertices exactly
as in the planted coloring. Clearly, in Step~3 all the vertices
we color are colored as in the planted coloring. Denote the set of
colored vertices by $A$ and the remaining vertices by $B$. Denote
by $B_{1}\subseteq B$ the set of vertices with a colored neighbor,
and by $B_{2}\subseteq B$ the rest. The proof follows by showing
that all the vertices in $B_{2}$ are of the same planted color. By
applying the Chernoff and the union bounds it follows that, with high
probability, for every pair of vertices $u,v\in V\left(G\right)$
there are $\Theta\left(\frac{d^{2}}{n}\right)$ common neighbors (assume
this event holds). By our assumption on $d$, it follows that given
a pair of vertices $u,v$ of different planted color classes at least
$\Theta\left(\frac{d^{2}}{n}\right)-O\left(\frac{n}{d}\right)>0$
vertices in $A$ are neighbors of both $u,v$ in $G$, take one such
vertex $w$. Hence, for every planting, $w\in A$ is still a neighbor
of at least one of $u,v$, say $v$. But this is a contradiction for
$v\in B_{2}$.
\end{proof}
When $n^{\frac{1}{2}}\le d\le n^{\frac{2}{3}}$ it holds that all
pairs have a lot of common neighbors and Theorem~\ref{thm:random graphs with Adv Planting}
does not give any advantage when that is the case. The following generalization
to Theorem~\ref{thm:random graphs with Adv Planting} deals with
that regime.
\begin{thm}
\label{thm:random graphs with Adv Planting-1}Let $G\sim G_{n,d}$
and $P_{C}\left(G\right)$ be the graph $G$ with an adversarial 3-coloring.
With high probability, Algorithm~\ref{alg:3-coloring-for-random graphs with adv planting}
outputs a partial coloring for $P_{C}\left(G\right)$ that is legal
on the colored vertices such that at most $\tilde{O}\left(\frac{n^{2}}{d^{3}}\right)$
vertices are not colored.
\end{thm}
Note that when $d\ge\sqrt{n}$ then indeed we get advantage from applying
Theorem~\ref{thm:random graphs with Adv Planting-1} (and when $d<\sqrt{n}$
then Theorem~\ref{thm:Coloring with balanced adversarial planting}
gives a better guarantee of $O\left(\frac{n}{d}\right)$-approximated
coloring).
\begin{proof}
(Of Theorem~\ref{thm:random graphs with Adv Planting-1}). Recall
that with high probability $\lambda_{2}\left(G\right)=\Theta\left(\sqrt{d}\right)$.
Hence, by Theorem~\ref{thm:Coloring with balanced adversarial planting},
it follows that after Step~2 we have a partial coloring of $P_{C}\left(G\right)$
that colors all but $O\left(\frac{n}{d}\right)$ vertices exactly
as in the planted coloring. Clearly, in Step~3 all the vertices
we color are colored as in the planted coloring. Denote the set of
colored vertices by $A$ and the remaining vertices by $B$. Denote
by $B_{1}\subseteq B$ the set of vertices with a colored neighbor,
and by $B_{2}\subseteq B$ the rest. Fix an arbitrary set $A'\subset V\left(G\right)$
of size $n\left(1-\frac{1}{d}\right)$ and let $B'=V\left(G\right)\setminus A'$.

For every vertex $w\in A'$ it hold that
\begin{align*}
\Pr_{G\sim G_{n,d}}\left[\left|N_{G}\left(w\right)\cap A'\right|\ge2\right] & =1-\left(1-p\right)^{\frac{n}{d}}-\frac{n}{d}p\left(1-p\right)^{\frac{n}{d}-1}\\
 & =\Theta\left(1\right)\,.
\end{align*}

Consider the set $C'=\left\{ w\in V\left(G\right)\,|\,\left|N\left(w\right)\cap A'\right|\ge2\right\} $.
By applying the Chernoff bound it follows that
\[
\Pr\left[\left|C'\right|<0.1n\right]\le2^{-\Omega\left(n\right)}\,.
\]

Condition on the event that $\left|C'\right|\ge0.1n$. For every vertex
$v\in B'$ define the set $B_{v,A'}=\left\{ u\in B'\,|\,\nexists w\in A'\, s.t.\,\left\{ u,v\right\} \subseteq N_{G}\left(w\right)\right\} $.
It holds that for any $t\ge0$
\begin{align*}
\Pr\left[\left|B_{v,A'}\right|\ge t\right] & \le\frac{n}{d}\left(1-\frac{t}{\left(\frac{n}{d}\right)^{2}}\right)^{0.1n}\\
 & \le\frac{n}{d}e^{-\frac{t}{\left(\frac{n}{d}\right)^{2}}0.1n}\,.
\end{align*}

Let $\eta$ be some constant and $Bad$ be the following event
\[
\exists A'\subseteq V\left(G\right)\, s.t.\,\left|A'\right|=\frac{n}{d},\,\exists v\in B'\, s.t.\,\left|B_{v,A'}\right|\ge\frac{n^{2}}{d^{3}}\log^{\eta}d\,.
\]
By the the above and the union bound it follows that for every $\eta>1$
\begin{align*}
\Pr_{G\sim G_{n,d}}\left[Bad\right] & \le\binom{n}{n/d}\,\left(2^{-\Omega\left(n\right)}+\frac{n^{2}}{d^{2}}e^{-\frac{t}{\left(\frac{n}{d}\right)^{2}}0.1n}\right)\\
 & \le e^{\Theta\left(\frac{n}{d}\log d\right)-\Theta\left(\frac{n}{d}\log^{\eta}d\right)}\,.
\end{align*}

The last inequality follows from the approximation $\binom{n}{k}\le\left(e\frac{n}{k}\right)^{k}$,
and the last term approaches zero as $n$ grows.

Assume that no legal coloring of $P_{C}\left(G\right)$ was found
in Step~5. It follows that the set $B_{2}$ was colored by $C$ with
at least two colors. Let $u$ be a vertex such that $col_{C}\left(u\right)=i$.
Let $j\neq i$ be any other color class in $C$. Suppose $J=\left|\left\{ v\,|\, col_{C}\left(v\right)=j\right\} \right|\ge\frac{n^{2}}{d^{3}}\log^{\eta}d$
then, conditioned on the complement event $Bad$, there exist a vertex
$v\in J$ and $w\in A$ such that $\left\{ u,v\right\} \subseteq N_{G}\left(w\right)$.
Since $u,v$ are in different color classes it follows that at least
one of them is not in $B_{2}$, as no meter what $col_{C}\left(w\right),$
at least one of $u,v$ remains a neighbor of $w$ in $P_{C}\left(G\right)$,
which is a contradiction. It follows that the remaining set of uncolored
vertices is of size at most $3\frac{n^{2}}{d^{3}}\log^{\eta}d$.

\end{proof}

\subsection{\label{sec:G Hardness of 3-coloring random graphs with adversarial planted 3-coloring}Hardness
of 3-coloring random graphs with adversarial planted 3-coloring}

The following definition can be found in \cite{alon2004probabilistic}
(see Chapter~4).
\begin{defn}
{[}balanced graph{]}. Given a graph $H$, denote by $\alpha$ its
average degree. A graph $H$ is balanced if every induced subgraph
of $H$ has an average degree of at most $\alpha$. \end{defn}
\begin{thm}
(Restatement of Lemma~\ref{lem:Q}). \label{thm:Hardness of 3-coloring random graphs with adversarial planted 3-coloring}
For $0<\epsilon\le\frac{1}{7}$ and $3<\alpha<4$, suppose that $\frac{k^{2}d^{2}}{n}\le\epsilon$
and $\frac{k^{4}n^{\alpha-2}}{d^{\alpha}}\le\epsilon^{2}$. Let $H$
be an arbitrary balanced graph on $k$ vertices with an average degree
$\alpha$ and let $G\sim G_{n,d}$ be a random graph. Then with probability
at least $1-4\epsilon$ (over choice of $G$), $G$ contains a set
$S$ of $k$ vertices such that:
\begin{enumerate}
\item The subgraph induced on $S$ is $H$.
\item No two vertices of $S$ have a common neighbor outside $S$.
\end{enumerate}
\end{thm}
\begin{proof}
Let $p=\frac{d}{n-1}$ be the edge probability in $G$. Suppose for
simplicity (an assumption that can be removed) that $k$ divides $n$.
Partition the vertex set of $G$ into $k$ equal parts of size $n/k$
each. Vertex $i$ of $H$ will be required to come from part $i$.
A set $S$ with such a property is said to {\em obey the partition}.

Let $X$ be a random variable counting the number of sets $S$ obeying
the partition that satisfy the theorem. Let $Y$ be a random variable
counting the number of sets $S$ obeying the partition that have $H$
as an edge induced subgraph (but may have additional edges, and may
not satisfy item~2 of the theorem).

\[
E[Y]=\left(\frac{n}{k}\right)^{k}p^{\frac{\alpha k}{2}}=\left(\frac{d^{\alpha}}{k^{2}n^{\alpha-2}}\right)^{\frac{k}{2}}\,.
\]

A set $S$ in $Y$ contributes to $X$ if it has no internal edges
beyond those of $H$ (which happens with probability at least $1-{k \choose 2}\frac{d}{n}$)
and no two of its vertices has a common neighbor outside $S$ (which
happens with probability at least $1-{k \choose 2}\frac{d^{2}}{n}$).
Consequently:

\[
E[X]\ge E[Y]\left(1-\frac{k^{2}d^{2}}{n}\right)\ge(1-\epsilon)E[Y].
\]

Now let us compute $E[Y^{2}]$. Given one occurrence of $H$, consider
another potential occurrence $H'$ that differs from it by $t$ vertices.
Since $H$ is balanced graph then
\begin{align*}
\left|E\left(G_{H'\setminus H}\right)\right|+\left|E\left(G_{H'\setminus H},G_{H}\right)\right| & \ge\frac{\alpha\left|V\left(G\right)\right|}{2}-\frac{\alpha\left|V\left(G_{H\cap H'}\right)\right|}{2}\\
 & \ge\frac{\alpha}{2}t\,.
\end{align*}
Hence, the probability that $H'$ is realized is at most $p^{\frac{\alpha t}{2}}$.
The number of ways to choose the $t$ other vertices is ${k \choose t}\left(\frac{n}{k}\right)^{t}$.
Hence the expected number of such occurrences is $\mu_{t}\le{k \choose t}\left(\frac{n}{k}\right)^{t}p^{\frac{\alpha t}{2}}={k \choose t}\left(\frac{d^{\alpha}}{k^{2}n^{\alpha-2}}\right)^{\frac{t}{2}}$.

We have:
\begin{enumerate}
\item $\mu_{k}<E[Y]$.
\item $\sum_{t\le\frac{k}{2}}\frac{\mu_{t}}{E[Y]}\le2^{k}\left(\frac{d^{\alpha}}{k^{2}n^{\alpha-2}}\right)^{\frac{-k}{4}}=\left(\frac{16k^{2}n^{\alpha-2}}{d^{\alpha}}\right)^{\frac{k}{4}}$.
\item $\sum_{t\ge\frac{k}{2}}^{k-1}\frac{\mu_{t}}{E[Y]}\le\sum_{t\ge\frac{k}{2}}^{k-1}k^{k-t}\left(\frac{d^{\alpha}}{k^{2}n^{\alpha-2}}\right)^{\frac{t-k}{2}}=\sum_{t\ge\frac{k}{2}}^{k-1}\left(\frac{d^{\alpha}}{k^{4}n^{\alpha-2}}\right)^{\frac{t-k}{2}}$.
The term $t=k-1$ dominates (when $d^{\alpha}\ge2k^{4}n^{\alpha-2}$)
and hence the sum is at most roughly $\sqrt{\frac{k^{4}n^{\alpha-2}}{d^{\alpha}}}$.
\end{enumerate}
Recall that $\frac{k^{4}n^{\alpha-2}}{d^{\alpha}}\le\epsilon^{2}$.
Then $\sum\mu_{i}\le\left(1+\epsilon\right)E[Y]$. Hence $E[Y^{2}]\le(1+\epsilon)\left(E[Y]\right)^{2}$.
Recall that $X\le Y$ and that $E[X]\ge(1-\epsilon)E[Y]$. Hence $E[X^{2}]\le\frac{1+\epsilon}{(1-\epsilon)^{2}}E[X]^{2}\le(1+4\epsilon)E[X]^{2}$
(the last inequality holds because $\epsilon\le\frac{1}{7}$). We
get that $\sigma^{2}[X]=E[X^{2}]-E[X]^{2}\le4\epsilon E[X]^{2}$.
By Chebychev's inequality we conclude that $Pr[X\ge0]\ge1-\frac{\sigma^{2}[X]}{E[X]^{2}}\ge1-4\epsilon$.
\end{proof}
\textbf{Remark:} The proof of Theorem~\ref{thm:Hardness of 3-coloring random graphs with adversarial planted 3-coloring}
shows that the number of copies of $H$ in $G$ is likely to be close
to its expectation, and hence large (this will be useful in the next
section). Also, simple modifications to the proof can be used in order
to show the existence of many disjoint copies (where the number grows
as $\epsilon$ decreases).

\begin{cor}
\label{cor:Hardness of 3-coloring random graphs with adversarial planted 3-coloring}
For every $0.467<\delta<\frac{1}{2}$ and $\epsilon<\frac{1}{8}$
there is some $\rho>0$ such that the following holds for every large
enough $n$. Let $H$ be an arbitrary balanced graph with average
degree $3.75$ on $k=n^{\rho}$ vertices. Let $G$ be a random graph
on $n$ vertices with average degree $d=n^{\delta}$ (which we refer
to as $G_{n,d}$). Then with probability larger than $1-4\epsilon$
(over choice of $G$), $G$ contains a set $S$ of $k$ vertices such
that:
\begin{enumerate}
\item The subgraph induced on $S$ is $H$.
\item No two vertices of $S$ have a common neighbor outside $S$.
\end{enumerate}
\end{cor}
\begin{proof}
In Theorem~\ref{thm:Hardness of 3-coloring random graphs with adversarial planted 3-coloring},
choose $\epsilon<\frac{1}{8}$, $\alpha = 3.75$,  and choose $k$ such that $\frac{k^{2}d^{2}}{n}\le\epsilon$
and $\frac{k^{4}n^{\alpha-2}}{d^{\alpha}}\le\epsilon^{2}$. Specifically,
one may choose $k=\min\left[\sqrt{\epsilon}n^{\frac{1-2\delta}{2}},\sqrt{\epsilon}n^{\frac{\alpha\delta-\alpha+2}{4}}\right]=n^{\Omega(1)}$. \end{proof}
\begin{thm}
(Restatement of Lemma~\ref{lem:3.75}). \label{thm:dense graph reduction}Coloring
a balanced graph with an average degree $3.75$ with a balanced $3$-coloring
is NP-hard. \end{thm}
\begin{proof}
It is known that $3$-coloring $4$-regular graphs is NP-hard, see~\cite{garey1976some}
(and actually, with slight modifications, this proof shows it as well).
Therefore it is enough to show that there exists a polynomial time
reduction $R$ such that for any given 4-regular graph $H$ it holds
that
\begin{enumerate}
\item $R\left(H\right)$ is a balanced graph with an average degree of $3.75$.
\item $H$ is a 3-colorable graph if and only if $R\left(H\right)$ is $3$-colorable
and given a (legal) $3$-coloring to $R\left(H\right)$ one can (legally)
$3$-color $H$ in a polynomial time.
\item If $R\left(H\right)$ is $3$-colorable then it has a balanced coloring.
\end{enumerate}
The reduction is as follows. For every vertex $v$ of $H$ consider
its four edges $e_{1},e_{2},e_{3},e_{4}$, replace $v$ by the graph
in Figure~\ref{fig:The-construction-of} (denote this graph by $R\left(H,v\right)$)
and then connect edge $e_{i}$ to vertex $u_{i}$. Note that the average
degree of $R\left(H\right)$ is $3.75$. Also note that in any legal
$3$-coloring of $R\left(H\right)$ the vertices $v_{i},u_{i}$ get
all the same color and the second assertion of $R$ follows.

We show that $R\left(H\right)$ is a balanced graph. For a vertex
$v$ of $H$ let $R\left(H,v_{i}\right)$ be the set $\left\{ v_{i},u_{i},A_{i},B_{i}\right\} $.
Consider a subset $S^{*}$ of the vertices of $R\left(H\right)$ such
that the average degree on the induced subgraph $R\left(H\right)_{S^{*}}$
is maximized to $\alpha^{*}$. Let $\alpha^{*}\left(R\left(H,v_{i}\right)\right)$
be the average degree of the vertices $R\left(H,v_{i}\right)\cap S^{*}$
in $R\left(H\right)_{S^{*}}$. As the sets $R\left(H,v_{i}\right)$
are disjoint and their union is the vertex set of $R\left(H\right)$
it follows that $\alpha^{*}$ is upper bounded by $\alpha^{*}\left(R\left(H,v_{i}\right)\right)$
for some $v_{i}$. But for every $v_{i}$ and $S^{*}$ we are averaging
at most $4$ vertices of degree bounded by $4$, where at least one
of them is of degree bounded by $3$. It follows that
\[
\alpha^{*}\le\max_{v_{i}}\alpha^{*}\left(R\left(H,v_{i}\right)\right)\le\frac{3*4+3}{4}=3.75\,.
\]

To show the third assertion of $R$ it is enough to take a disjoint
union of $3$ copies of the above construction (note that the disjoint
union of two balanced graphs is balanced).

\begin{figure}
\centering{}\includegraphics[scale=0.4]{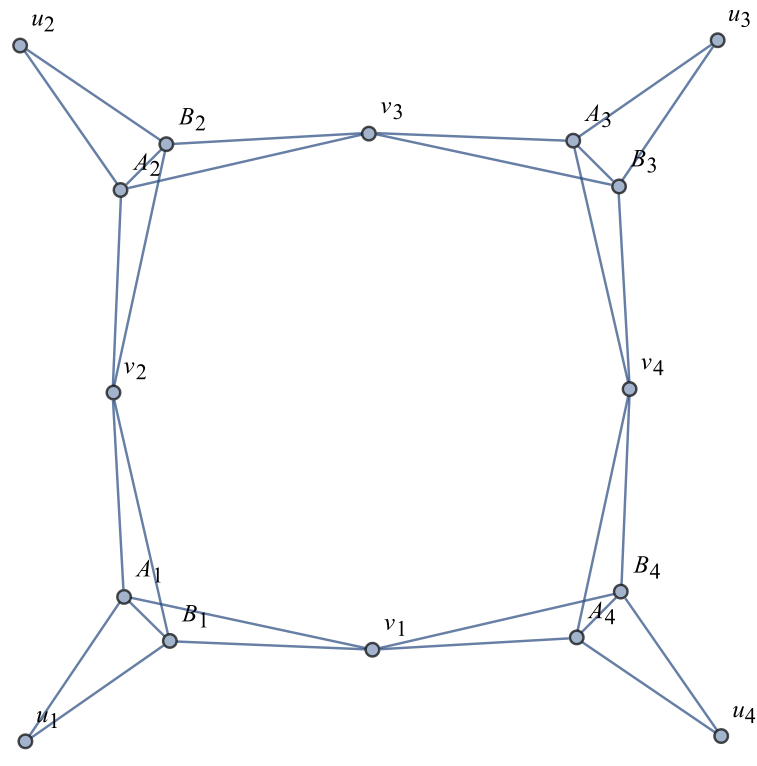}\protect\caption{\label{fig:The-construction-of}The construction of Theorem~\ref{thm:dense graph reduction}.}
\end{figure}

\end{proof}

{\bf Remark:} The construction and analysis
of Theorem~\ref{thm:dense graph reduction} can be modified to show that for every $\epsilon > 0$, 3-coloring of
balanced graphs with average degree $\left(\frac{10}{3}+\epsilon\right)$ is NP-hard. Every vertex $v_i$ in the
graph of Figure~\ref{fig:The-construction-of} is replaced by a 4-vertex gadget with five edges of structure similar to the graph induced on $v_1,A_1,B_1,v_2$, with the $v$ vertices as endpoints of the gadget. For example, if $v_2$ is replaced then one endpoint is connected to $A_1$ and $B_1$, and the other endpoint is connected to $A_2$ and $B_2$. Observe that the two endpoints of the gadget must have the same color in every legal 3-coloring. Each such replacement increases the number of vertices by three and the number of edges by five, hence bringing the average degree closer to $\frac{10}{3}$.
Repeating this replacement recursively (until the distance between $A_1$ and $A_2$ becomes $\Omega(\frac{1}{\epsilon})$) gives a balanced graph with average degree below $\left(\frac{10}{3}+\epsilon\right)$. Further details omitted.

~

For the sake of intuition, we temporarily restrict attention to algorithms that we refer to as {\em
decomposable} (a restriction that will be lifted later).
\begin{defn}
An algorithm $A$ for 3-coloring is {\em decomposable} if for every
disconnected input graph $G$, algorithm $A$ is applied independently
to each of $G$'s connected components.
\end{defn}
Natural 3-coloring algorithms are decomposable. In fact, we are not
aware of any coloring algorithm that is not decomposable. Moreover,
in works on random and semi-random models of inputs, it makes sense
to require coloring algorithms to be decomposable, as an algorithm
that is not decomposable would presumably involve aspects that are
very specific to the model and would not generalize to other models.

One can imagine that in some contexts the use of algorithms that are
not decomposable may offer advantages. This may happen if the input
graph is generated in such a way that the structure of one component
contains hints as to how to color other components. Perhaps the simplest
form of a hint is the following. Suppose that the input graph is known
to be generated with a balanced coloring (in which each color class
is of size $n/3$), and furthermore, is known to be generated such
that in each component the 3-coloring is unique. Then for an input
graph with two components, once one colors the first component, one
knows how many vertices of each color class there should be in the
second component. This simple form of a hint saves at most polynomial
factors in the running time, because it involves only $O(\log n)$
bits of information, and hence the hint can be guessed. Nevertheless,
it illustrates the point that under some generation models of input
graphs, it is possible that algorithms that are not decomposable will
be faster than algorithm that are decomposable.

Here we consider 3-coloring $G_{n,d}$ with an adversarially planted
balanced 3-coloring. We show that given the decomposable-algorithm
assumption, a hardness result can be derived. In Section~\ref{sub:Hardness-Result-Without-DECOM-ASS}
a full proof is given without the decomposability assumption.

\begin{thm}
\label{thm:Hardness of 3-coloring random graphs with adversarial planted 3-coloring 2}
Suppose that for some $0.467<\delta<\frac{1}{2}$ and $d=n^{\delta}$
there is a decomposable algorithm that with probability at least $\frac{1}{2}$
(over choice from $G_{n,d}$ and for every adversary) 3-colors $G_{n,d}$
with an adversarially planted balanced 3-coloring. Then P=NP. \end{thm}
\begin{proof}
By Theorem~\ref{thm:dense graph reduction} it follows that 3-coloring
balanced graphs of average degree $3.75$ is NP-hard.

Suppose there was an algorithm $A$ for 3-coloring $G_{n,d}$ with
an adversarially planted balanced 3-coloring as in the statement of
the theorem. Consider now an arbitrary balanced graph $H$ with average
degree $3.75$ of size $k=n^{\rho}$ (where $\rho$ is as in Corollary~\ref{cor:Hardness of 3-coloring random graphs with adversarial planted 3-coloring}),
and associate with it an adversary $H'$. On input a random graph
$G$ from $G_{n,d}$, the graph has probability more than $\frac{1}{2}$
of satisfying the conclusion of Theorem~\ref{thm:Hardness of 3-coloring random graphs with adversarial planted 3-coloring}.
The adversary $H'$ (who is not computationally bounded) does the
following.
\begin{enumerate}
\item If $G$ does not satisfy the conclusion of Theorem~\ref{thm:Hardness of 3-coloring random graphs with adversarial planted 3-coloring}
with respect to $H$, then the adversary $H'$ plants in it a random
balanced coloring.
\item If $G$ satisfies the conclusion of Theorem~\ref{thm:Hardness of 3-coloring random graphs with adversarial planted 3-coloring}
with respect to $H$ then the adversary $H'$ leaves $H$ untouched
and then:

\begin{enumerate}
\item If $H$ is 3-colorable, for each color class of $H$ it colors its
neighborhood outside $H$ with the same color as in $H$, and then
completes to a balanced planted 3-coloring at random. Observe that
this planted coloring disconnects $H$ from the rest of $G$, because
all original edges between $H$ and the rest of $G$ are between pairs
of vertices of the same color.
\item If $H$ is not 3-colorable, the adversary removes all edges between
$H$ and the rest of $G$, and randomly produces a balanced planted
coloring of the rest of $G$.
\end{enumerate}
\end{enumerate}
Case~2 above happens with probability greater than $\frac{1}{2}$,
by Corollary~\ref{cor:Hardness of 3-coloring random graphs with adversarial planted 3-coloring}.

Suppose that $H$ is not 3-colorable (case 2(b)). Then $A$ must fail
to 3-color $H$.

Suppose now that $H$ is 3-colorable (case 2(a)). Because $A$ is
decomposable, it must color $H$ without seeing the rest of $G$.
Because $A$ succeeds for every adversary on at least half the inputs
(over choice from $G_{n,d}$), and for adversary $H'$ over half the
inputs generate $H$, $A$ must succeed to 3-color $H$. (We assumed
here that $A$ is deterministic. If $A$ is randomized then choose
$\epsilon<\frac{1}{16}$ and then $A$ must succeed with probability
at least $\frac{2}{3}$. In this case the conclusion will be that
NP has randomized polynomial time algorithms with one sided error.)

Hence the output of $A(H)$ determines whether $H$ is 3-colorable.
As this applies to every $H$, and the sizes of $H$ and $G$ are
polynomially related, this implies that $A$ solves in polynomial
time an NP-hard problem, implying $P=NP$.
\end{proof}

There are two weaknesses of Theorem~\ref{thm:Hardness of 3-coloring random graphs with adversarial planted 3-coloring 2}.
One is that it requires $d>n^{0.467}$: at lower densities the input
graph is unlikely to contain a given $H$ with average degree $3.75$. The degree $d$ can be lowered to roughly $n^{0.4}$ using the remark following the proof of Theorem~\ref{thm:dense graph reduction}. However, it cannot be lowered below $n^{1/3}$ (using our techniques), because of Proposition~\ref{pro:sparseQ}.
The other weakness is that it requires $A$ to be decomposable. The decomposability
weakness can be overcome using the following approach.

Suppose there was an algorithm $A$ for 3-coloring $G_{n,d}$ with
an adversarially planted balanced 3-coloring, that succeeds with probability
at least $\frac{1}{2}$ over choice of $G$ (for every adversary).
Given a 3.75-balanced graph $H$ on $k$ vertices, give $A$ as input
a graph $G'$ composed of two disjoint parts. One is $H$ and the
other is a random subgraph of size $n-k$ of a random graph from $G_{n,d}$
with a randomly planted balanced 3-coloring. This would prove Theorem~\ref{thm:Hardness of 3-coloring random graphs with adversarial planted 3-coloring 2}
if the distribution generated by this process is statistically close
to the one generated by the adversary $H'$. The techniques in~\cite{Juels98hidingcliques}
can be extended in order to prove statistical closeness
and this is done in the next section.

\subsection{\label{sub:Hardness-Result-Without-DECOM-ASS}Hardness result without
the decomposable-algorithm assumption}


Let $H$ be an arbitrary balanced graph
with average degree $\alpha$ and $k$ vertices, for $\alpha=3.75$.
Let $G$ be a graph with $n$ vertices. Assume that $k$ divides $n$
(this assumption can be removed) and consider a fixed partition
of the vertex set of $G$ to $k$ disjoint subsets of vertices, each
of size $\frac{n}{k}$. Let $C_{H}\left(G\right)$ be the number of
induced sub-graphs of $G$ that are isomorphic to $H$ such that they
obey the partition (see the definition in the proof of Theorem~\ref{thm:Hardness of 3-coloring random graphs with adversarial planted 3-coloring 2})
and let $E_{H}$ be $\mathbb{E}_{G\sim G_{n,d}}\left[C_{H}\left(G\right)\right]$.

Note that
\begin{equation}
E_{H}=\left(\frac{n}{k}\right)^{k}p^{\frac{\alpha}{2}k}\left(1-p\right)^{\binom{k}{2}-\frac{\alpha}{2}k}\,,\label{eq:Expected number of induced subgraphs}
\end{equation}

where $p=\frac{d}{n-1}$.

~

We consider the following distribution of random graphs with the graph
$H$ being planted as an induced sub-graph.
\begin{defn}
\label{def:g_n,d,H}
A graph $G$ with $n$ vertices is distributed by $G_{n,d,H}$ if
it is created by the following random process.
\begin{enumerate}
\item Take a random graph $G'$ distributed by $G_{n,d}$.
\item Choose a random subset $K$ of $k$ vertices from $G'$ that obeys
the partition.
\item Replace the induced subgraph of $G'$ on $K$ by $H$ (we say that
$H$ is randomly planted in $G'$).
\end{enumerate}
\end{defn}
Given a graph $G$, we denote by $p\left(G\right)$ the probability
to output $G$ according to $G_{n,d}$ and by $p'\left(G\right)$
the probability to output $G$ according to $G_{n,d,H}$.
\begin{claim}
\label{claim: Probability strech when planted}For any given graph
$G$ it holds that $p'\left(G\right)=\frac{C_{H}\left(G\right)}{E_{H}}p\left(G\right)$.\end{claim}
\begin{proof}
Let $e$ be the number of edges in $G$ and consider $p'\left(G\right)$.
Out of the $\left(\frac{n}{k}\right)^{k}$ options to choose $K$
(in $G_{n,d,H}$) only $C_{H}\left(G\right)$ options are such that
the induced sub graph on $K$ is $H$ so that the resulting graph
could be $G$. Given that we chose a suitable $K$, the rest of the
edges ($e-\frac{\alpha}{2}k$) should agree with $G$. It follows
that
\begin{align*}
p'\left(G\right) & =\frac{C_{H}\left(G\right)}{\left(\frac{n}{k}\right)^{k}}p^{e-\frac{\alpha}{2}k}\left(1-p\right)^{\binom{n}{2}-\binom{k}{2}-e+\frac{\alpha}{2}k}\\
 & =\frac{C_{H}\left(G\right)}{E_{H}}p^{e}\left(1-p\right)^{\binom{n}{2}-e}\\
 & =\frac{C_{H}\left(G\right)}{E_{H}}p\left(G\right)\,.
\end{align*}

The second equality follows from Equation~\ref{eq:Expected number of induced subgraphs}.
\end{proof}
~

Recall the definition of the random variables $X,Y$ from the proof
of Theorem~\ref{thm:Hardness of 3-coloring random graphs with adversarial planted 3-coloring 2}.
Moreover, the following claim follows from the proof of Theorem~\ref{thm:Hardness of 3-coloring random graphs with adversarial planted 3-coloring 2}.
\begin{claim}
\label{claim:Variance of the number of induced  subgraphs}
\[
\sigma^{2}\left[X\right]=4\epsilon\mathbb{E}\left[X\right]^{2}\,.
\]
\end{claim}
\begin{thm}
\label{thm:Hardnes adv planting on gnp full}(Restatement of Theorem~\ref{thm:R/A}b).
Let $G\sim G_{n,d}$ be a random graph, where $d=n^{\delta}$ for $0.467\le\delta<\frac{1}{2}$,
	and let $\gamma$ be an arbitrary small constant.
	 If there exists
a randomized algorithm that colors $G\sim G_{n,d}$ after an adversarial
color-planting with probability  $\gamma$ (over the distribution $G_{n,d}$ and the
randomness of the algorithm) then $RP=NP$.\end{thm}
\begin{proof}
Consider $\epsilon=\frac{\gamma}{100}$ and set $k=\min\left[\sqrt{\epsilon}n^{\frac{1-2\delta}{2}},\sqrt{\epsilon}n^{\frac{\alpha\delta-\alpha+2}{4}}\right]=n^{\Omega(1)}$
(so that the conditions of Corollary~\ref{cor:Hardness of 3-coloring random graphs with adversarial planted 3-coloring}
and Theorem~\ref{thm:Hardness of 3-coloring random graphs with adversarial planted 3-coloring}
hold). Assume that there exists an algorithm $A$ as in the theorem.
The probability measure of graphs from $G_{n,d}$ that $A$ colors
with respect to \emph{all} possible color planting with probability
(over the randomness of $A$) at least $\frac{\gamma}{2}$ is at least
$\frac{\gamma}{2}$. This holds by averaging and because we can consider
an adversary that, given any input graph, simulates $A$ and try all
possible color planting in order fail $A$ with the largest probability
(over the randomness of $A$).

By Chebychev's inequality and Claim~\ref{claim:Variance of the number of induced  subgraphs}
it holds that
\[
\Pr_{G\sim G_{n,d}}\left[X\le\rho\mathbb{E}\left[X\right]\right]\le\frac{4\epsilon}{\left(1-\rho\right)^{2}}\,.
\]
 Note that $X\le C_{H}\left(G\right)\le Y$. By the proof of Theorem~\ref{thm:Hardness of 3-coloring random graphs with adversarial planted 3-coloring}
it hold that
\[
\mathbb{E}\left[C_{H}\left(G\right)\right]\le\mathbb{E}\left[Y\right]\le\frac{1}{1-\epsilon}\mathbb{E}\left[X\right]\,.
\]

Hence
\[
\Pr_{G\sim G_{n,d}}\left[X\le\rho\left(1-\epsilon\right)\mathbb{E}\left[C_{H}\left(G\right)\right]\right]\le\frac{4\epsilon}{\left(1-\rho\right)^{2}}\,.
\]

Set $\rho$ to be such that $\frac{4\epsilon}{\left(1-\rho\right)^{2}}\le\frac{1}{4}\gamma$
and that $\frac{1}{2}\rho\ge\frac{1}{10}$. Therefore that the probability
measure of graphs that $A$ colors (with probability, over the randomness
of $A$, of at least $\frac{\gamma}{2}$ for every color planting)
with $X\ge\rho\left(1-\epsilon\right)\mathbb{E}\left[C_{H}\left(G\right)\right]$
is at least $\frac{1}{4}\gamma$. By Claim~\ref{claim: Probability strech when planted}
it follows that the probability measure with respect to $G_{n,d,H}$
of graphs that $A$ colors (after any adversarial planting with probability
at least $\frac{\gamma}{2}$) is at least $\frac{1}{2}\gamma\rho\left(1-\epsilon\right)\ge\frac{1}{11}\gamma$.

Fix $H$ to be an arbitrary balanced graph with average degree $3.75$
with $k$ vertices that has a balanced $3$-coloring. Now we show
that we can use $A$ to color $H$.

Given $G\sim G_{n,d,H}$ we consider the following distribution $\tilde{G}_{H}$
for graphs with an adversarial coloring. If any vertex that is not
on the induced planted (by $G_{n,d,H}$) graph $H$ has two or more
neighbors in the planted induced graph, denote this event by $S_{1}$,
then color $G$ arbitrarily. Otherwise for the planted graph use a
coloring that agrees with the coloring of $H$. For the rest of the vertices, if a
vertex has a neighbor in the planted graph color it by the same color
of its neighbor and the rest of the vertices are colored in such away
that the coloring is balanced. More specifically, from all the balanced
coloring we choose one at random (again, if there is no possible balanced coloring then we  color
$G$ arbitrarily). Since $k=o\left(\frac{n}{d}\right)$ then, by the
union bound and the Chernoff bound, no vertex in $G$ has more then $\left(1+c_{1}\right)d$
neighbors with probability at most $n2^{-\Omega\left(c_{1}d\right)}\le\gamma/44$.
Hence there are such balanced colorings with high probability. Denote
the event that no such balanced coloring is possible by $S_{2}$.
The event $S_{1}$ happens with probability at most $k^{2}\frac{d^{2}}{n}\le\gamma/44$.
A graph from $\tilde{G}_{H}$ is called \emph{good} if the events
$S_{1},S_{2}$ do not hold. Given that $G$ is a good graph, all possible
balanced coloring on $G\setminus H$ are equally distributed.

If an adversary knows the subset $S$ of vertices that the distribution
$G_{n,d,H}$ plants $H$ on then the distribution $\tilde{G}_{H}$
can be created by this adversary. A subtle point is that the adversary
can guess $S$. Given a graph $G\sim G_{n,d,H}$ the adversary can
calculate for every subset $S'$ (that obeys the partition) that satisfies
$G_{S'}=H$ what is the probability that $S'$ is the planted subset,
then choose a subset $S'$ with the calculated probability and behave
as $S'$ is the planted subset. It follows that the distribution after
the above preprocessing is the same distribution as if the adversary
knows $S$.

In total, the probability measure of good graphs in $\tilde{G}_{H}$
that $A$ can color (with probability, over the randomness of $A$,
of at least $\frac{\gamma}{2}$ for every color planting) is at least
$\gamma/22$. The key point is that the conditional, on being good,
distribution of graphs $\tilde{G}_{H}$ can be sampled in a polynomial
time by taking the vertex disjoint union of the graph $H$ and a random
graph from $G_{n-k,d}$ with a random balanced planted coloring (and
apply an appropriate random permutation). If we run algorithm $A$
on $\Omega\left(\frac{44}{\gamma^{2}}\right)$ random instances from
$\tilde{G}_{H}$ (which again we can sample efficiently) then with
high probability $A$ will color $H$ (that is a graph with a $poly\left(n\right)$
vertices) with high probability. By Theorem~\ref{thm:dense graph reduction}
the proof follows.
\end{proof}

\subsection{\label{sec:Proof-of-Proposition sparce copy}Proof of Proposition~\ref{pro:sparseQ}}
\begin{prop}
(Restatement of Proposition~\ref{pro:sparseQ}). Let ${\cal Q}$
be an arbitrary class of graphs. Then either there is a polynomial
time algorithm for solving 3-colorability on every graph in ${\cal Q}$,
or ${\cal Q}$ contains graphs that are unlikely to appear as subgraphs
of a random graph from $G_{n,p}$, if $p=n^{-2/3}$.\end{prop}
\begin{proof}
We say that a class ${\cal Q}$ is {\em 3-sparse} if every graph in $Q$ has average degree at most~3, and furthermore, has no subgraph of average degree above~3. There are two cases to consider.

Suppose that ${\cal Q}$ is 3-sparse. In this case, there is a polynomial time algorithm that solves 3-colorably on all graphs from ${\cal Q}$. Let $Q \in {\cal Q}$ be an arbitrary such graph. Iteratively remove from $Q$ vertices of degree less than~3 until no longer possible, and let $Q'$ be the remaining graph. Every 3-coloring of $Q'$ can be extended to $Q$ by inductive coloring. Hence it remains to 3-color $Q'$. If $Q'$ is empty then we are done. If $Q'$ is nonempty, then 3-sparseness of $Q$ implies that $Q'$ is 3-regular. In this case, Brook's theorem~\cite{brooks1941colouring}
(see~\cite{lovasz1975three} for an algorithmic version of it) implies that we can decide whether $Q'$ is $3$-colorable, and if so, 3-color it in polynomial time.

Suppose that ${\cal Q}$ is not 3-sparse. Then it contains some graph $Q$ and within it some subgraph $Q'$ such that $Q'$  contains at least $\frac{3k}{2}+1$ edges, where $k$ denotes the number of vertices in $Q'$. The probability that a random $G_{n,p}$ graph with $p = n^{-2/3}$ contains an induced copy of $Q'$ is at most $n^k p^{3k/2 + 1} = p = n^{-2/3}$.
\end{proof}

\section{Extending the results to more than 3 colors}
\label{sec:k>3}
In this section we elaborate on our results when using $k\ge 4$ colors. When stating our results we shall use the notation $O_k(.)$ to denote hidden constants whose value may depend on the number $k$ of colors.
We derive the following theorems.

 \begin{thm}
 	\label{thm:partial-k}
 	(Generalization of Theorem~\ref{thm:partial}).
 	For any positive  $k$ there exists a  constant $c_k$, such that if the
 	average degree in the host graph satisfy $c_k<d<n$ then the following holds. In all four
 	models ($H_{A}/P_{A}$, $H_{A}/P_{R}$, $H_{R}/P_{A}$, $H_{R}/P_{R}$)
 	there is a polynomial time algorithm that finds a $b$-partial coloring
 	for $b=O_k(\left(\frac{\lambda}{d}\right)^{2}n)$. For the models with
 	random host graphs ($H_{R}$) and/or random planted colorings ($P_{R}$),
 	the algorithm succeeds with high probability over choice of random
 	host graph $H$ and/or random planted coloring $P$.
 \end{thm}

 \begin{thm}
 	\label{thm:A/A-k}
 	(Generalization of Theorem~\ref{thm:A/A}).
 	For any positive $k$ there exists a  constant $C_k$, such that  the following holds.
 	In the $H_{A}/P_{A}$ model, for every $d$ in the
 	range $C_k<d<n^{1-\epsilon}$ (where $\epsilon>0$ is arbitrarily small), it is NP-hard to $k$-color
 	a graph with a planted $k$-coloring, even when $\lambda=O_k(\sqrt{d})$.
 \end{thm}

 \begin{thm}
 	\label{thm:A/R-k}
 	(Generalization of Theorem~\ref{thm:A/R}).
 	For any positive $k$ there exist constants $0<c_k<1$ and $C_k>1$, such that  the following holds.
 	In the $H_{A}/P_{R}$ model there is a polynomial time algorithm with the following
 	properties. For every $d$ in the range $C_k<d\le n-1$ and every $\lambda\le c_kd$,
 	for every host graph within the model, the algorithm with high probability
 	(over the choice of random planted $k$-coloring) finds a legal $k$-coloring.
 \end{thm}

 \begin{thm}
	\label{thm:Generalization of R/Ab}
	(Generalization of Theorem~\ref{thm:R/A}b).
	Let $G\sim G_{n,d}$ be a random graph, where $d=n^{\delta}$ for $0.467\le\delta<1$, let $k \ge 4$, and let
	$\gamma$ be an arbitrary small constant. If there exists a randomized
	algorithm that $k$-colors $G\sim G_{n,d}$ after an adversarial color-planting
	with probability $\gamma$ (over the distribution $G_{n,d}$ and the randomness
	of the algorithm) then $RP=NP$.\end{thm}

It turns out that theorems~\ref{thm:partial},~\ref{thm:A/A},~\ref{thm:A/R} and~\ref{thm:R/A}b extend to $k \ge 4$. Theorem~\ref{thm:R/A}a does not extend to $k \ge 4$ as this would contradict Theorem~\ref{thm:Generalization of R/Ab}.

~

Theorems ~\ref{thm:partial},~\ref{thm:A/A} and~\ref{thm:A/R} follow in a rather straightforward way from our proofs of the respective theorems~\ref{thm:partial},~\ref{thm:A/A} and~\ref{thm:A/R}. We provide some more details.

We start with the generalization of Theorem~\ref{thm:Main} to $k \ge 4$. The statement of Theorem~~\ref{thm:Main} uses two predefined vectors ($\bar{x}$ and $\bar{y}$) that encodes the planted 3-coloring (see Definition~\ref{Def:x,y}). Here we describe how to define $k - 1$ vectors when $k \ge 4$.

Let $P_k(G)$ be the graph $G$ after a random $k$-color-planting has been applied. Let $V_1,V_2,...,V_k$ be a partition of $V$ induced by the color planting and let $P$ be a function that maps each vertex to its color class.
Let $P_k=\mathbf{1}_{k \times k}-I_k$ be the matrix obtained by the all-one matrix minus the identity matrix and let the vectors  $p_0,p_1,...,p_{k-1}$ be an orthogonal set of eigenvectors of $P$ where $p_0$ is the all-one vector.
For a vector $\vec{x}$ we denote by $\vec{x}(j)$ its $j$-th coordinate and recall  that we denote by  $\bar{x}$ the vector  $\frac {\vec{x}}{\left\Vert \vec{x}\right\Vert_2}$ .


\begin{defn}
\label{Def:x(p)} We define the following $k$ vectors $\vec{x}_0, \ldots, \vec{x}_{k-1}$ in $\mathbb{R}^{n}$. For each $j$, $0 \le j \le k-1$,  the vector $\vec{x}_j$  gives to coordinate $i$ the value $p_j(P(i))$.

\end{defn}

	One can check that Definition~\ref{Def:x,y} is an instantiation of Definition~\ref{Def:x(p)} in the case $k=3$. Moreover Theorem~\ref{thm:Main-k} bellow follows by a direct generalization of the proof of Theorem~\ref{thm:Main}.

\begin{thm}
\label{thm:Main-k}
Let $G$ be a $d$ regular $\lambda$-expander and $P_{k}\left(G\right)$ be the graph $G$ after a random $k$-color planting, where
$\lambda\le \Theta(d)$.
With high probability the following holds.

\begin{enumerate}
\item The eigenvalues of $P_{k}\left(G\right)$ have the following spectrum.

\begin{enumerate}
\item $\lambda_{1}\left(P_{k}\left(G\right)\right)\ge\left(1-2^{-\Omega(d)}\right)\frac{k-1}{k}d$.
\item $\lambda_{n-k+2}\left(P_{k}\left(G\right)\right)\le-\frac{1}{k}d\left(1-\frac{k}{\sqrt{d}}\right)$.
\item $|\lambda_{i}\left(P_{k}\left(G\right)\right)|\le2\lambda+O_k\left(\sqrt{d}\right)$
for all $2\le i\le n-k+1$.
\end{enumerate}
\item The following vectors exist.

$\vec{\epsilon}_{\bar{x}_j}$ such that $\left\Vert \vec{\epsilon}_{\bar{x}_j}\right\Vert _{2}=O_k\left(\frac{1}{\sqrt{d}}\right)$
and $\bar{x}_j+\vec{\epsilon}_{\bar{x}_j}\in\text{span}\left(\left\{ e_{i}\left(P_k(G)\right)\right\} _{i\in\left\{ n-k+2,...,n\right\} }\right)$, for $1\le j\le k-1$.
\end{enumerate}

\end{thm}

Lemma~\ref{lem: eigen values after arbitrary planting} is a variant of Theorem~\ref{thm:Main} for the case of adversarial balanced 3-coloring.   Lemma~\ref{lem: eigen values after arbitrary planting} can be generalized in a similar manner as the above for the case of adversarial balanced $k$-coloring.




We now present a clustering algorithm (which is a generalization of Algorithm~\ref{alg:Spectral Clustering}) to the case planted  $k \ge 4$-coloring. We chose here a randomized version due to improvement in the running times over the deterministic version.


\begin{algorithm}[H]
	Input: A graph $P_k(G)$ and three positive constants $c\le\frac{1}{2}$,
	 $c^1_k$ and $c^2_k$ (the last two constants depend on $k$).
	\begin{enumerate}
		\item Compute the eigenvectors $e_{n-l}:=e_{n-l}\left(P_k(G)\right)$ for $0\le l \le k-2$.
		\item \label{step:rankclus} Choose uniformly at random a set of $k$ vertices $v_{1},v_{2},...,v_{k}\in V\left(P_k(G)\right)$.
		\item If there are $1 \le i < j \le k$ satisfying
\[
		\sum_{l=0}^{l=k-2}\left(\left(e_{n-l}\right)_{v_i}-\left(e_{n-l}\right)_{v_{j}}\right)^{2}<\frac{4}{c^1_k n}\,.
		\]
		go to Step~2.
		\item Put a vertex $u\in V\left(G'\right)$ in $S_{i}$ if  it
		holds that
		\[
		\sum_{l=0}^{l=k-2}\left(\left(e_{n-l}\right)_{u}-\left(e_{n-l}\right)_{v_{i}}\right)^{2}<\frac{1}{c^1_k n}\,.
		\]
		\item If for every $i=1,2,...,k$ it holds that $\left|S_{i}\right|\ge\left(\frac{1}{k}-\frac{1}{c^2_k d^{2c}}\right)n$
		then output a coloring $C$ of $G'$ that sets $col_{C}\left(u\right)=i$
		for every $1 \le i \le k$ and $u\in S_{i}$, and colors the remaining vertices (if there are any) arbitrarily. Otherwise go to Step~\ref{step:rankclus}.
		
	\end{enumerate}
	\protect\caption{\label{alg:Random Spectral k-Clustering}Random Spectral $k$-Clustering}
\end{algorithm}

The following lemma shows that under sufficient conditions the above clustering algorithm, Algorithm~\ref{alg:Random Spectral k-Clustering}, outputs a good approximated coloring. These conditions are stated in the lemma and by  Theorem~\ref{thm:Main-k} these conditions  hold with high probability.

\begin{lem}
	\label{lem:Eigenvector approximation-k}(Generalization of Lemma~\ref{lem:Eigenvector approximation}). For any positive $k$ there exist constants $c^1_k$ and $c^2_k$ such that the following holds.
	Let $P_k(G)$ be as above and  $c$ be a positive constant. Suppose that the following
	vectors exist for $0\le l \le k-1$ (recall Definition~\ref{Def:x(p)}):
	$\vec{\epsilon}_{\bar{x}(l)}$
	such that $\left\Vert \vec{\epsilon}_{\bar{x}_l}\right\Vert _{2}=O_k\left(d^{-c}\right)$
	and $\bar{x}_l+\vec{\epsilon}_{\bar{x}_l}\in\text{span}\left(\left\{ e_{i}\left(P_k(G)\right)\right\} _{i\in\left\{ n-k+1,n\right\} }\right)$. Then Algorithm~\ref{alg:Random Spectral k-Clustering}  outputs an $O_k\left(nd^{-2c}\right)$-approximated
	coloring.
	The expected running time of Algorithm~\ref{alg:Random Spectral k-Clustering} is $\tilde{O}(\frac{k^k}{k!} n)$.
	\end{lem}
	
	An example of how to use Lemma~\ref{lem:Eigenvector approximation-k} is as follows. For $P_k(G)$ by Theorem~\ref{thm:Main-k} we have $\left\Vert \vec{\epsilon}_{\bar{x}(l)}\right\Vert _{2}=O_k\left(\frac{1}{\sqrt{d}}\right)$. Thus running Algorithm~\ref{alg:Random Spectral k-Clustering} with $c=\frac{1}{2}$ results in $O_k\left(\frac{n}{d}\right)$-approximated coloring.

	The analysis of  Algorithm~\ref{alg:Random Spectral k-Clustering} running time is discussed in the poof of Lemma~\ref{lem:Eigenvector approximation}.
	 We note that  one can show  an expected running time of $\tilde{O}(\frac{k^k}{k!}+n)$ for Algorithm~\ref{alg:Random Spectral k-Clustering}.

\subsection{\label{sub:G.2 Hardness result for random graphs with adversarial 4-color planting}Hardness
result for random graphs with adversarial 4-color planting }

In this section we extend Theorem~\ref{thm:Hardnes adv planting on gnp full}
to the case of adversarial planting with $k>3$ colors. For simplicity, we present the proof only for the case $k=4$, but it is not difficult to
extend it to any constant $k$.

Let us begin with an overview of the proof. Recall the overview of the proof of Theorem~\ref{thm:R/A}b given in Section~\ref{R/Ahard}, and the terminology that was used there. The hardness of 3-coloring in $H_R/P_A$ was by reduction from the class ${\cal Q}$ of balanced graphs with average degree $3.75$, on which 3-coloring is NP-hard. For this class, 4-coloring is easy (by inductive coloring), but nevertheless we shall use the same class ${\cal Q}$ in the hardness result for 4-coloring. We have already seen (in Lemma~\ref{lem:Q}) that for every graph $Q \in {\cal Q}$ of size $n^{\epsilon}$, random graphs of sufficiently high average degree $n^{\delta}$ are likely to contain $Q$ as a subgraph. For hardness of 3-coloring, given a random $G_{n,n^{\delta}}$ host graph $H$, the adversary plants in $H$ a 3-coloring that isolates a copy of $Q$. For hardness of 4-coloring, it is useless to plant in $H$ a 4-coloring that isolates a copy a copy of $Q$, because 4-coloring of $Q$ is easy. Instead, the plan is for the adversary to plant in $H$ a 4-coloring in which all neighbors of $Q$ outside $Q$ have the same color. This leaves only three colors for $Q$, and hence 4-coloring of $H$ would imply 3-coloring of $Q$.

To make this plan work, one needs every vertex of $Q$ to have at least one neighbor in $H - Q$, and all vertices of $Q$ combined should have less than $n/4$ neighbors in $H - Q$. Luckily, for a random copy of $Q$ in $H$, both these properties happen with overwhelming probability as long as $\epsilon + \delta < 1$. Consequently, the hardness result for 4-coloring has the following two advantages over the one for 3-coloring (that required vertices in $Q$ not to have common neighbors in $H - Q$). One is that there is no need to require that $\delta < \frac{1}{2} - \epsilon$. In particular, this shows that the positive results of Theorem~\ref{thm:R/A}a do not hold when $k\ge4$. The other advantage is that the fraction of host graphs $H \in_R G_{n,n^{\delta}}$ on which the reduction fails is smaller than the corresponding fraction in the proof of Theorem~\ref{thm:R/A}b. (For simplicity, we shall not address this point in our proofs.)

There is a property that is required for hardness of 4-coloring but was not needed for the hardness of 3-coloring. That property is that after the adversary plants the 4-coloring, we need that in every 4-coloring of $G$, there is some color $c$ such every vertex of $Q$ has at least one neighbor in $H - Q$ of  color $c$. This ensures that every 4-coloring of $G$ indeed 3-colors $Q$. The way we show that this property holds is through  Lemma~\ref{lem:Uniqueness of random coloring} that shows that in a sufficiently dense random graph, a planted random $k$-coloring is almost surely the only legal $k$-coloring of the resulting graph.

This completes the overview of our proof approach.

 We start with a variant of Theorem~\ref{thm:Hardness of 3-coloring random graphs with adversarial planted 3-coloring}.
\begin{thm}
	\label{thm:(k=00003D4)Hardness of 3-coloring random graphs with adversarial planted 3-coloring}For
	$0<\epsilon\le\frac{1}{7}$ and $3<\alpha<4$, suppose that $\frac{k^{2}d}{n}\le\epsilon$
	and $\frac{k^{4}n^{\alpha-2}}{d^{\alpha}}\le\epsilon^{2}$. Let $H$
	be an arbitrary balanced graph on $k$ vertices with an average degree
	$\alpha$ and let $G\sim G_{n,d}$ be a random graph. Then with probability
	at least $1-4\epsilon$ (over choice of $G$), $G$ contains a set
	$S$ of $k$ vertices such that:
	\begin{enumerate}
		\item The subgraph induced on $S$ is $H$.
		\item Each vertex of $S$ has a neighbor outside $S$.
		\item The total sum of degrees of vertices  in  $S$  is at most $\frac{n-k}{4}$.
	\end{enumerate}
\end{thm}
\begin{proof}
	Let $p=\frac{d}{n-1}$ be the edge probability in $G$. Suppose for
	simplicity (an assumption that can be removed) that $k$ divides $n$.
	Partition the vertex set of $G$ into $k$ equal parts of size $n/k$
	each. Vertex $i$ of $H$ will be required to come from part $i$.
	A set $S$ with such a property is said to {\em obey the partition}.
	
	Let $X$ be a random variable counting the number of sets $S$ obeying
	the partition that satisfy the theorem. Let $Y$ be a random variable
	counting the number of sets $S$ obeying the partition that have $H$
	as an edge induced subgraph (but may have additional edges, and may
	not satisfy item~2 of the theorem).
	
	\[
	E[Y]=\left(\frac{n}{k}\right)^{k}p^{\frac{\alpha k}{2}}=\left(\frac{d^{\alpha}}{k^{2}n^{\alpha-2}}\right)^{\frac{k}{2}}\,.
	\]

	A set $S$ in $Y$ contributes to $X$ if the following conditions hold
	\begin{enumerate}
		\item It has no internal edges
		beyond those of $H$ (which happens with probability at least $1-{k \choose 2}\frac{d}{n}$).
		
		\item Every vertex of it has a neighbor outside $S$ (which happens
		with probability at least $1-ke^{-\frac{dn}{4}}$).
		
		\item  The total sum of degrees of vertices  in  $S$  is at most $\frac{n-k}{4}$.
		By the union and the Chernoff bounds no
		vertex in $G$ has more then $\left(1+c_{1}\right)d$ neighbors with probability at
		most $n2^{-\Omega\left(c_{1}d\right)}$. Since $k=o\left(\frac{n}{d}\right)$ this assertion is satisfied (for $S$ in $Y$) with probability $n2^{-\Omega\left(c_{1}d\right)}$ as well.
		
	\end{enumerate}
	  Consequently:
	
	\[
	E[X]\ge E[Y]\left(1-\frac{k^{2}d^{2}}{n}\right)\ge(1-\epsilon)E[Y].
	\]

	Using $\frac{k^{4}n^{\alpha-2}}{d^{\alpha}}\le\epsilon^{2}$ the computation
	of $E[Y^{2}]$ and the rest of the proof are the same as in Theorem~\ref{thm:Hardness of 3-coloring random graphs with adversarial planted 3-coloring}.

\end{proof}

\begin{cor}
	\label{cor: (k=00003D4)Hardness of 3-coloring random graphs with adversarial planted 3-coloring}
	For every $\delta\ge 0.467$ and $\epsilon<\frac{1}{8}$ there
	is some $\rho>0$ such that the following holds for every large enough
	$n$. Let $H$ be an arbitrary balanced graph with average degree
	$3.75$ on $k=n^{\rho}$ vertices. Let $G$ be a random graph on $n$
	vertices with average degree $d=n^{\delta}$ (which we refer to as
	$G_{n,d}$). Then with probability larger than $1-4\epsilon$ (over
	choice of $G$), $G$ contains a set $S$ of $k$ vertices such that:
	\begin{enumerate}
		\item The subgraph induced on $S$ is $H$.
		\item Each vertex of $S$ has a neighbor outside $S$.
		\item The total sum of degrees of vertices  in  $S$  is at most $\frac{n-k}{4}$.
	\end{enumerate}
\end{cor}
\begin{proof}
	In Theorem~\ref{thm:Hardness of 3-coloring random graphs with adversarial planted 3-coloring},
	choose $\epsilon<\frac{1}{8}$, $\alpha =3.75$, and choose $k$ such that $\frac{k^{2}d}{n}\le\epsilon$
	and $\frac{k^{4}n^{\alpha-2}}{d^{\alpha}}\le\epsilon^{2}$. Specifically,
	one may choose $k=\min\left[\sqrt{\epsilon}n^{\frac{1-\delta}{2}},\sqrt{\epsilon}n^{\frac{\alpha\delta-\alpha+2}{4}}\right]=n^{\Omega(1)}$.
\end{proof}

\begin{thm}
	\label{thm:Generalization of R/Ab for the k=4 case}
	(Generalization of Theorem~\ref{thm:R/A}b for the $k=4$ case).
	Let $G\sim G_{n,d}$ be a random graph, where $d=n^{\delta}$ for $0.467\le\delta<1$, and let
	$\gamma$ be an arbitrary small constant. If there exists a randomized
	algorithm that colors $G\sim G_{n,d}$ after an adversarial color-planting
	with probability $\gamma$ (over the distribution $G_{n,d}$ and the randomness
	of the algorithm) then $RP=NP$.\end{thm}


\begin{proof}
	Suppose that there exists an algorithm $A$ as in the theorem. Consider $\epsilon=\frac{\gamma}{100}$ and set $k=\min\left[\sqrt{\epsilon}n^{\frac{1-\delta}{2}},\sqrt{\epsilon}n^{\frac{\alpha\delta-\alpha+2}{4}}\right]=n^{\Omega(1)}$
	(so that the conditions of Corollary~\ref{cor: (k=00003D4)Hardness of 3-coloring random graphs with adversarial planted 3-coloring}
	and Theorem~\ref{thm:(k=00003D4)Hardness of 3-coloring random graphs with adversarial planted 3-coloring}
	hold).  Let $H$ be an arbitrary balanced graph with average degree $3.75$
	with $k$ vertices that has a balanced $3$-coloring. We show
	how $A$ can be used in order to color $H$.
	Recall Definition~\ref{def:g_n,d,H} for the graphs distribution  $G_{n,d,H}$.
	By the proof of Theorem~\ref{thm:Hardnes adv planting on gnp full}
	it follows that the probability measure with respect to $G_{n,d,H}$
	of graphs that $A$ colors (after any adversarial planting, with probability
	at least $\frac{\gamma}{2}$) is at least $\frac{1}{2}\gamma\rho\left(1-\epsilon\right)\ge\frac{1}{11}\gamma$.

	Given $G\sim G_{n,d,H}$ we consider the following distribution $\tilde{G}_{H}$
	for graphs with an adversarial coloring. Let $S_1$ be the event that some vertex in
	the induced planted graph $H$ (in $G_{n,d,H}$) has no neighbors
	outside $H$. Let $S_2$ be the event that the total number of neighbors (among the remaining vertices in $G$) that the induced planted graph $H$ has is larger than $n/4$. If either event $S_1$ or $S_2$ happen (these are considered bad events), then plant an arbitrary balanced 4-coloring in $G$. If neither event $S_1$ nor $S_2$ happens (this will be shown to be the typical case), plant in $G$ a balanced 4-coloring chosen uniformly at random, conditioned on the following two events: the planted graph is colored by a balanced 3-coloring that agrees with the coloring of $H$, and all neighbors of $H$ (outside $H$) are colored by the fourth color.

 The event $S_{1}$ happens with probability at most $k^{2}e^{-\frac{dn}{4}}\le\gamma/44$. As to event $S_2$, since $k=o\left(\frac{n}{d}\right)$
	then, by the union bound and the Chernoff bound, no vertex in $G$ has more
	then $\left(1+c_{1}\right)d$ neighbors with probability at most $n2^{-\Omega\left(c_{1}d\right)}\le\gamma/44$.
	Consequently, event $S_{2}$ happens with probability at most $\gamma/44$ as well.
	
	One last event $S_{3}$ to consider is that the $4$-coloring of the
	graph induced on vertices not in $H$ is not unique. By Lemma~\ref{lem:Uniqueness of random coloring}
	below and since a random graph is a $O\left(\sqrt{d}\right)$-expander~\cite{feige2005spectral,furedi1981eigenvalues,friedman1989second},
	it follows that $S_{3}$ happens with probability at most $\frac{\gamma}{44}$
	as well. A graph from $\tilde{G}_{H}$ is called \emph{good} if the
	events $S_{1},S_{2},S_{3}$ do not hold.
	

	In total, the probability measure of good graphs in $\tilde{G}_{H}$
	that $A$  colors with probability
	of at least $\frac{\gamma}{2}$ for every color planting (over the randomness of $A$) is at least
	$\gamma/50$. Moreover, the
	distribution $\tilde{G}_{H}$ of good graphs can be sampled in a polynomial
	time as follows:
	\begin{enumerate}
		\item Take  the vertex disjoint
		 union
		of the graph $H$ and a random graph from $G'\sim G_{n-k,p}$ (with $p = d/n$).
		 \item Add edges between $H$ and $G'$ (each edge is added independently with
		probability $p$).
		\item 4-color $G'$ randomly conditioned on three color classes each containing exactly $\frac{n}{4} - \frac{k}{3}$ vertices and the remaining color class (of size $\frac{n}{4}$) containing all neighbors of $H$. Remove all monochromatic edges from $G'$.
		
	\end{enumerate}
	
	By the above, the probability that  this procedure fails to construct a good graph  is bounded by some small constant. Every 4-coloring of a good graph gives 3-coloring of $H$ (because the 4-coloring of $G'$ is unique, and in this 4-coloring every vertex of $H$ is adjacent to some vertex of $G'$ of the fourth color class).
	Running algorithm $A$ on $\Omega\left(\frac{50}{\gamma^{2}}\right)$
	random instances from $\tilde{G}_{H}$ (which we can sample
	efficiently) then with high probability $A$ will 4-color at least one of these instances.  By Theorem~\ref{thm:dense graph reduction}
	the proof follows.\end{proof}
\begin{lem}
	\label{lem:Uniqueness of random coloring} Let  $k\in\mathbb{N}^{+}$
	be a constant. Let $G$ be a $d$-regular $\lambda$-expander
	graph with $d=\Omega(\log (n))$ and $\lambda\le c \frac{1}{k^{2.5}} d$, where $c$ is a sufficiently small constant. Let $G'$ be the graph $G$ with a random balanced $k$-color planting.
	Then with high probability $G'$ has only one legal $k$-coloring (namely, the planted $k$-coloring).\end{lem}
\begin{proof}
	
	Assume, towards a contradiction, that there exist two different $k$-colorings $\chi_{1}$ (the
	planted coloring) and $\chi_{2}$ of $G'$. Define their distance $d\left(\chi_{1},\chi_{2}\right)$
	to be the Hamming distance between two strings that represents
	the colors of $G'$'s vertices with respect to $\chi_{1},\chi_{2}$
	(choose a naming of the colors that minimizes the distance).
	Throughout  the proof we treat $\chi_{i}$
	as a string with the naming of the colors that meets the above definition
	of distance.

	Let $\delta$ be a constant which is much smaller than $\frac{1}{k}$, and let $col_{\chi_{i}}(v)$ denote the color of vertex $v$ in the coloring $\chi_{i}$.
    Let $SB\subseteq V$ denote the following vertex set. $v\in SB$ if there is some color class other than $col_{\chi_{1}}(v)$ such that $v$ has either more
    than
    $\left(\frac{1}{k}+\delta \right)d$
     neighbors or less than
    $\left(\frac{1}{k}-\delta\right)d$
     neighbors in that color class.
    By a similar proof as of Lemma~\ref{lem:SB bound the random planting Case},
	it follows that with high
	probability there are no vertices in $SB$. We assume throughout the proof that indeed $SB$ is empty, and that $\delta$ is negligible compared to $\frac{1}{k}$.

    We first show that  $d(\chi _1,\chi _2)\ge\frac{d/k-\alpha d}{d} n$.
    Let $H$ be the sub-graph of $G$ induced on the vertices
	$\left\{ v\in V\left(G\right)\,|\,\chi_{1}\left(v\right)\neq\chi_{_{2}}\left(v\right)\right\} $.
	Since that $\chi_{2}$ is legal then for $v\in H$ it holds that any neighbor $u$ of $v$ such that  $col{\chi_{1}}(u)=col{\chi_{2}}(v)$ is in $H$ as well. By our assumption on $SB$, it follows that $H$ has a minimal degree of roughly $d/k$. Recall that $E_{G}\left(H,H\right)$ is twice the number of edges in the induced subgraph $H$. By
	the expander mixing lemma it holds that
	\begin{align*}
		\frac{d}{k}\left|H\right| & \le\left|E_{G}\left(H,H\right)\right| \\
								  & \le\frac{d}{n}\left|H\right|^{2}+\lambda\left|H\right|\,.
	\end{align*}
	Hence $\left|H\right|\ge\frac{d/k-\lambda}{d} n$.
	
	Now we show the following claim
	\begin{claim}
		\label{claim:upper bound on distance from planted coloring}
		$d(\chi _1,\chi _2) \le \epsilon n$, for
		$$\epsilon=k\left((k-1)\epsilon_1+(k-1)\left(\frac{\lambda}{d}\right)^2\frac{1}{\frac{1}{k} (\frac{1}{k}-(k-1)\epsilon_1)}\right)\,,$$
		where  $\epsilon_1=k(k-1)\left(\frac{\lambda}{d}\right)^2$.	
	\end{claim}
	 We prove Claim~\ref{claim:upper bound on distance from planted coloring} below. Note that if $\lambda \le c \frac{1}{k^{2.5}}d$, for sufficiently small constant $c$, then $\epsilon n < \frac{d/k-\lambda}{d} n$. Hence, there is no other  legal coloring for $G'$ other than the planted coloring.
\end{proof}
In remains to prove Claim~\ref{claim:upper bound on distance from planted coloring}.
	 \begin{proof}
	(of Claim~\ref{claim:upper bound on distance from planted coloring}). Let $col_{i,j}$ be  vertices with color $j$ in coloring $\chi _{i}$.
	We claim that $\left|col_{2,j}\right|\le (\frac{1}{k}+\epsilon_1)n$ for all $j\in \left\{1,2,..,k\right\}$.
	Given that $\left|col_{2,j}\right|\ge \frac{1}{k}n$. We claim that for some
	$j_{1}\neq j_{2}\in\left\{1,2,..,k\right\}$ it holds that
	\begin{equation}
	\label{eq:x1timesx2}
	\left|col_{2,j}\cap col_{1,j_{1}}\right|\left|col_{2,j}\cap col_{1,j_{2}}\right|\ge \frac{1}{k}n\frac{|col_{2,j}|-\frac{1}{k}n}{k-1}\,.
	\end{equation}
	To see this, denote by $x_i$ a possible size for  $col_{2,j}\cap col_{1,j_{1}}$. Clearly $\sum_{i=1}^{k} x_k =\left|col_{2,j}\right|$.
For notational convenience we may assume that $x_1$ and $x_2$ are largest among the $\{x_i\}$.
Fixing the sum $x_1 + x_2$, the product $x_1 x_2$ is minimized when their values are unbalanced as possible.
In the worst case, $x_1=\frac{n}{k}$. Now to get minimal $x_1 x_2$, we assume $\left\{x_i\right\}_{i\neq 1}$ are all equal and the lower bound is established.

	Since $\chi_2$ is legal, by Equation~\ref{eq:x1timesx2} and  the expander mixing lemma, it holds that
		$$ \frac{d}{n}\frac{1}{k}n\frac{|col_{2,j}|-\frac{1}{k}n}{k-1} \le\lambda \sqrt{\frac{1}{k}n\frac{|col_{2,j}|-\frac{1}{k}n}{k-1}}\,.$$
	Thus $\left|col_{2,j}\right|\le (\frac{1}{k}+\epsilon_1)n$ and $col_{2,j}\ge (\frac{1}{k}-(k-1)\epsilon_1)n$.
	
	By the above, for every $j$ there exists $j_1$ such that  $\left|col_{1,j_1}\cap col_{2,j}\right|\ge \frac{1}{k} (\frac{1}{k}-(k-1)\epsilon_1)n$.
	Since $\chi _2$ is legal and by the expander mixing lemma, for every $j_2\neq j_1$ it holds that
	$$ \frac{d}{n} |col_{1,j_2}\cap col_{2,j}| \left|col_{1,j_1}\cap col_{2,j}\right| \le \lambda\sqrt{|col_{1,j_2}\cap col_{2,j}| \left|col_{1,j_1}\cap col_{2,j}\right|}\,.$$
		Therefore $|col_{1,j_2}\cap col_{2,j}|\le \left(\frac{\lambda}{d}\right)^2\frac{1}{\frac{1}{k} (\frac{1}{k}-(k-1)\epsilon_1)}n$. Since this holds for any $j_2\neq j_1$ then
	\begin{align*}
	\left|col_{1,j_1}\cap col_{2,j}\right| &\ge \left(\frac{1}{k}-(k-1)\epsilon_1\right)n -\sum_{j_{2} \neq j_{1}}|col_{1,j_2}\cap col_{2,j}| \\
	&\ge    \left(\frac{1}{k}-(k-1)\epsilon_1-(k-1)\left(\frac{\lambda}{d}\right)^2\frac{1}{\frac{1}{k} (\frac{1}{k}-(k-1)\epsilon_1)}\right)n \\
	&=\left(\frac{1}{k}-\frac{1}{k}\epsilon\right)n\,.
	\end{align*}
	In other words, it holds that for every color $j\in \left\{1,2,..,k\right\}$ in $\chi _1$ there exists a matching color $j_1$ in $\chi _2$ such that both colors agree on  at least $\left(\frac{1}{k}-\frac{1}{k}\epsilon\right)n$ vertices (it is easy to see that this is a matching).
	Hence $d(\chi _1 , \chi _2)\le \left(\frac{k}{k}\epsilon\right)n=\epsilon n$,  as desired.
	\end{proof}

\section{NP-hard subgraphs are not always an obstacle to efficient coloring}
\label{sec:counterintuitive}

Consider the following proposition.

\begin{prop}
Let $Q(V_Q,E_Q)$ be an arbitrary 3-colorable 4-regular graph on $k = \frac{n^{1-\epsilon}}{4\log n}$ vertices.
There exists a spectral expander graph $H$ (with $n$ vertices) such that if a random 3-coloring $P$ is planted in $H$ and monochromatic edges are dropped, then with high probability the resulting graph $G$ contains a vertex induced copy of $Q$.
\end{prop}

\begin{proof}
(Sketch.) For simplicity, the host graph $H$ in our reduction will not be regular, but rather only nearly regular. However, it is not difficult to modify the reduction so that $H$ is regular.


 Based on $Q$, the host graph $H$ is constructed in two steps:

\begin{enumerate}

\item Select an arbitrary $d$-regular spectral expander $H'(V,E)$ on $n$ vertices, with $d = n^{\epsilon}$.

\item Let $S \subset V$ be an arbitrary independent set in $H'$ on $3k\log n$ vertices. Partition $S$ into $k$ equal size parts $S_1, \ldots, S_k$. For each $1 \le i < j \le k$, if $(i,j)\in E_Q$ then add to $E$ all $9\log^2 n$ edges between $S_i$ and $S_j$.

\end{enumerate}

$H$ constructed above is a spectral expander. This follows from Inequality~\ref{eq:Sums of Eigenvalues} in Section~\ref{sec:On the hardness of coloring expander graphs with adversarial planting}, which implies that no eigenvalue of $H'$ is shifted by more than $12\log n$ (the maximum number of edges added to a vertex in $H'$ in order to construct $H$).

%
%

Let $\chi$ be a legal 3-coloring of $Q$. Then with high probability over the choice of planted coloring in $H$, for every $1 \le i \le k$ there is some {\em faithful} vertex $v \in S_i$ such that $P(v) =  \chi(i)$.  Edges in $H$ connecting faithful vertices are not dropped in $G$. Hence taking one faithful vertex from each part gives a copy of $Q$.

\end{proof}

The above proposition shows that any 3-colorable 4-regular graph $Q$ can be embedded as a vertex induced subgraph in a graph $G$ generated by the $H_A/P_R$ model. In general, such graphs $Q$ are NP-hard to 3-color. Nevertheless our Theorem~\ref{thm:A/R} shows that the graph $G$ can be 3-colored in polynomial time. As $Q$ is a subgraph of $G$, this also produces a legal 3-coloring of $Q$, thus solving an NP-hard problem. However, this of course does not imply that P=NP. Rather, it implies that it is NP-hard to find in $G$ a subgraph that is isomorphic to $Q$.


\section{\label{sec:Useful-Facts}Useful facts}
In this section we state several well known theorems which we use throughout this manuscript.

The following theorem is due to~\cite{chernoff1952measure}.
\begin{thm}
\label{thm:=00005BChernoff=00005D}{[}Chernoff{]}. Let $X_{1},...,X_{n}$
be independent random variables.. Assume that $0\le X_{i}\le1$ always,
for each $i$. Let $X=\sum_{i=1}^{n}X_{i}$ and $\mu=\mathbb{E}\left[X\right]$.
For any $\epsilon\ge0$
\[
\Pr\left[X\ge\left(1+\epsilon\right)\mu\right]\le\exp\left(-\frac{\epsilon^{2}}{2+\epsilon}\mu\right)\,,
\]
 and
\[
\Pr\left[X\le\left(1-\epsilon\right)\mu\right]\le\exp\left(-\frac{\epsilon^{2}}{2}\mu\right)\,.
\]

\end{thm}
%
The following theorem can be found in~\cite{mcdiarmid1989method}.
\begin{thm}
\label{thm:=00005BMcDiarmid=00005D}{[}McDiarmid{]}. Let $X_{1},...,X_{m}\in\mathcal{X}^{m}$
be a set of $m\ge1$ independent random variables and assume that
there exist $c>0$ such that $f:\mathbb{R}^{n}\rightarrow\mathbb{R}$
satisfies the following conditions:
\[
\left|f\left(x_{1},...,x_{i},...,x_{m}\right)-f\left(x_{1},...,x'_{i},...,x_{m}\right)\right|\le c\,,
\]
for all $i\in\left[1,m\right]$ and any points $x_{1},...,x_{m},x'_{i}\in\mathcal{X}$
. Let $f(S)$ denote the random variable $f(X_{1},...,X_{m})$, then,
for all $\epsilon>0$, the following inequalities hold:
\[
\Pr\left[\left|f(S)-\mathbb{E}\left[f(S)\right]\right|\ge\epsilon\right]\le2\exp\left(-\frac{2\epsilon^{2}}{mc^{2}}\right)
\]

\end{thm}
The following theorem is well known,  see for example~\cite{horn2012matrix} (p.185).
\begin{thm}
\label{thm:=00005BCauchy-interlacing-theorem=00005D.}{[}Cauchy interlacing
theorem{]}. Let $A$ be a symmetric $n\times n$ matrix. Let $B$
be a $m\times m$ sub-matrix of $A$ , where $m\le n$.

If the eigenvalues of $A$ are $\alpha_{1}\le...\le\alpha_{n}$ and
those of $B$ are $\beta_{1}\le...\le\beta_{j}\le...\le\beta_{m}$,
then for all $j\le m$,
\[
\alpha_{j}\leq\beta_{j}\leq\alpha_{n-m+j}.
\]

\end{thm}
%
Recall that for $S,T\subseteq V\left(G\right)$, $E_{G}\left(S,T\right)$ denotes
the number of edges between $S$ and $T$ in $G$. If $S,T$ are not
disjoint then the edges in the induced sub-graph of $S\cap T$ are
counted twice.
The following lemma can be found in~\cite{alon1988explicit}.
\begin{lem}
\label{lem:=00005BMixing Lemma=00005D}{[}Expander Mixing Lemma{]}.
Let $G=(V,E)$ be a $d$-regular $\lambda$-expander graph with $n$
vertices, see Definition~\ref{def:expander}. Then for any two, not
necessarily disjoint, subsets $S,T\subseteq V$ the following holds:

\[
\left|E_{G}(S,T)-\frac{d\cdot|S|\cdot|T|}{n}\right|\leq\lambda\sqrt{|S|\cdot|T|}\,.
\]

\end{lem}

The following lemma is well known, a proof can be found for example in~\cite{vadhan2012pseudorandomness} (Chapter 4).
\begin{lem}
\label{lem:=00005BSpectral to vertex expansion.=00005D}{[}Spectral
to Vertex Expansion{]}. If $G$ is a $d$-regular $\lambda$-expander
graph, then, for every $\alpha\in\left[0,1\right]$, $G$ is an $\left(\alpha n,\frac{1}{\left(1-\alpha\right)\left(\frac{\lambda}{d}\right)^{2}+\alpha}\right)$
vertex expander, see Definition~\ref{def:=00005BVertex-expander=00005D.}.
\end{lem}

\end{document}